\documentclass[12pt]{article}
\usepackage{amsthm,amsmath,amsfonts,amssymb,cases}
\newcommand{\R}{{\mathord{\mathbb R}}}

\newcommand{\N}{{\mathord{\mathbb N}}}
\newcommand{\C}{{\mathord{\mathbb C}}}

\def\chib {\overline{\chi}}

\newcommand{\HH}{\mathcal{H}}

\newcommand{\FF}{\mathcal{F}}

\newcommand{\WW}{\mathcal{W}}
\newcommand{\hh}{\mathfrak{h}}
\newcommand{\kk}{\mathfrak{k}}

\newcommand{\umm}{\underline{m}}
\newcommand{\unn}{\underline{n}}
\newcommand{\upp}{\underline{p}}
\newcommand{\uqq}{\underline{q}}
\newcommand{\uzz}{\underline{0}}

\def\e {{e}}

\newcommand{\ran}{{\rm Ran}}

\newcommand{\ben}{\begin{displaymath}}
\newcommand{\een}{\end{displaymath}}
\newcommand{\beqn}{\begin{equation}}
\newcommand{\eeqn}{\end{equation}}
\newcommand{\beqna}{\begin{eqnarray*}}
\newcommand{\eeqna}{\end{eqnarray*}}
\DeclareMathOperator*{\esssup}{ess\,sup}

\def\inf{{\rm inf}\,}

\def\supp{\operatorname{supp}}

\newcommand{\sfrac}[2]{\textrm{\footnotesize $\frac{#1}{#2}$}}

\newtheorem{lemma}{Lemma}
\newtheorem{theorem}[lemma]{Theorem}
\newtheorem{remark}[lemma]{Remark}
\newtheorem{proposition}[lemma]{Proposition}
\newtheorem{corollary}[lemma]{Corollary}
\newtheorem{definition}[lemma]{Definition}

\begin{document}
\title{Ground States  in the Spin Boson Model}
\author{\vspace{5pt} D. Hasler $^1$\footnote{
E-mail: dghasler@wm.edu} and I.
Herbst$^2$\footnote{E-mail: iwh@virginia.edu.} \\
\vspace{-4pt} \small{$1.$ Department of Mathematics,
College of William and Mary} \\ \small{Williamsburg, VA 23187-8795, USA}\\
\vspace{-4pt}
\small{$2.$ Department of Mathematics, University of Virginia,} \\
\small{Charlottesville, VA 22904-4137, USA}\\}
\date{}
\maketitle

\begin{abstract}
 We prove that the Hamiltonian of the model describing a spin which is linearly
coupled to a field of relativistic and massless bosons, also known as the spin-boson model, admits a ground state
for small values of the coupling constant $\lambda$.
We show that the ground state energy is an analytic function of  $\lambda$ and
that the corresponding ground state can also be chosen to be an analytic function of $\lambda$.
No infrared regularization is imposed.
Our proof is based on a modified version of the BFS operator theoretic
renormalization analysis. Moreover, using a positivity argument we prove
that the ground state of the spin-boson model is unique. We show that
the expansion coefficients of the ground state and the ground state energy
can be calculated using regular analytic perturbation theory.
\end{abstract}

\section{Introduction}
\label{sec:int}

The spin boson model describes a quantum mechanical two level system which
is linearly coupled to the quantized field of bosons. We assume
that the quantized field is a relativistic field  of  massless bosons, and
we do not impose any infrared regularization.
In that case the
spin-boson model can be used  as a simplified caricature
describing
an atom coupled to the quantized electromagnetic
field. The two level system is a coarse approximation of
the energy levels of the atom. This model has been extensively
investigated, see for example \cite{spohn89,HS95,G00} and references therein.

Our first result states that
for all values of the coupling constant a possible ground state of
the spin boson model must be  unique. This result
is shown using a positivity argument with respect to a suitable choice of measure space.

Our second result  is that  the spin boson model admits a ground state for small values
of the coupling constant.  Quantum mechanical systems which are coupled to a relativistic field
of massless bosons typically do not admit ground states  unless
 cancellations of infrared divergences occur.
The reason the spin boson model admits
a ground state originates from the fact that the coupling matrix has no
diagonal entries, see  \eqref{eq:defofspinboson}.
In non-relativistic quantum electrodynamics (qed)
the gauge symmetry seems to be responsible for the existence of ground states
of molecules  \cite{BFS99,GLL01}.

Our third and main result is that a suitable choice of the ground state
as well as its energy  are analytic functions of the coupling
constant.
In non-relativistic qed expansions  of the
ground state and its energy as the coupling constant tends to zero  have
 recently attracted attention.
In \cite{BFP06,BFP09} it was proven that there  exists  an asymptotic expansion involving
coefficients which depend on the coupling parameter and may contain 
logarithmic expressions. Other expansion algorithms were employed for example in  
 \cite{BCVV09,HHS05,CH04} and it was shown that 
logarithmic terms can occur in
non-relativistic qed. 
On the other hand
it was shown that an atom in the dipole
approximation of qed (which effectively leads to an infrared regularization) has a ground state and ground state energy
which are analytic functions of the coupling constant \cite{GH09}.
We hope that our analyticity result concerning the spin boson model,
 will help to shed light on the nature
of infrared divergences occurring in such  expansions.

Once the analyticity of the ground state and its energy
have been shown, it is natural to ask
whether the coefficients of their power series
expansions can be obtained from regular analytic
perturbation theory. We prove that this is indeed the
case and illustrate how the ground state and its energy can be calculated using
Rayleigh-Schr\"odinger perturbation theory.
To this end we artificially introduce an infrared cutoff in
the Hamiltonian and show that the ground state
and the ground state energy are continuous
functions of that cutoff. Validity of
Rayleigh-Schr\"odinger perturbation theory will then
follow from the uniqueness property of the ground state.
In view of the explicit form of the  Rayleigh-Schr\"odinger coefficients
it is rather surprising that these coefficients are
infrared finite. The coefficients are given
as a sum of terms. While infrared divergent terms
occur 
our analyticity result
implies that the sum of these terms must be finite in the limit when the infrared cutoff is removed.

Let us now address the proof of the main results.
The ground state energy is embedded in the continuous
spectrum, see Proposition \ref{prop:halfline}. In such a situation  regular perturbation theory is
typically not applicable and other methods have to be employed.
To prove the existence result as well as the analyticity result for
the spin-boson model we use   a variant of the operator theoretic renormalization analysis as introduced
in \cite{BFS98} and further developed in \cite{BCFS03}.
The analysis as outlined in these papers
is not directly applicable to problems which are  infrared critical.
 To be able to apply a renormalization procedure, we first perform two initial
so called Feshbach transformations. This converts
the spectral problem of the original Hamiltonian into a problem
involving sums of normal ordered  operators containing only an even number of creation and
annihilation operators. We then must prove that on the  space of such operators
the renormalization procedure converges. To show this in a proper
way  we have to provide a detailed exposition  of the operator theoretic renormalization
transformation.

In \cite{GH09}  the  analyticity of  the ground state as well as the ground state energy of
an atom in the dipole approximation of non-relativistic qed was proven.
We want to point out that also in \cite{GH09}
operator theoretic renormalization was used in the proof, with a somewhat different representation of
 the spectral parameter.
Whereas the
problem considered in \cite{GH09}  was infrared regular,  the problem
considered in this paper is not subject to an infrared regularization.
Moreover, in \cite{GH09} the proof used that renormalization preserves
analyticity on the space of operators, in  this paper we use that renormalization preserves
analyticity on the space  of integral kernels.

In the next section we introduce the model and state the main results, which will then
be proven in later sections.

\section{Model and Statement of Results}
\label{sec:mod}
For a Hilbert space  $\hh$ we introduce the bosonic Fock space
$$
\Gamma(\hh) := \bigoplus_{n=0}^\infty S_n (\hh^{\otimes n}) \; ,
$$
where $S_n$ denotes the orthogonal projection onto the subspace of totally
symmetric tensors in $\hh^{\otimes n}$, and $S_0(\hh^{\otimes 0}) := \C$.
We introduce the vacuum vector $\Omega := (1,0,0,...) \in \FF(\hh)$.
Henceforth we  fix $\hh$ to be  $L^2(\R^3)$ and set $\FF := \FF(\hh)$.
We shall identify vectors $\psi \in \FF$ with sequences $(\psi_n)_{n=0}^\infty$ of
$n$-particle wave functions, $\psi_n(k_1,...,k_n)$, which are totally symmetric
in their $n$ arguments, and $\psi_0 \in \C$. The scalar product of two vectors
$\psi$ and $\phi$ is inherited from $\hh$ and is given by
$$
\langle \psi, \phi \rangle = \sum_{n=0}^\infty \int \overline{\psi_n(k_1,...,k_n)} \phi_n(k_1,...,k_n) d^3k_1...d^3k_n \; .
$$
For $g \in \hh$ one associates a creation operator defined as follows. For $\eta \in S_n (\hh^{\otimes n}) $,
$a^*(g) \eta$ is given by
$$
a^*(g) \eta = \sqrt{n+1} S_{n+1} ( g \otimes \eta ) \; .
$$
This defines a closable linear operator whose closure is also denoted by $a^*(g)$.
The annihilation operator $a(g)$ is defined to be the adjoint of $a^*(g)$.
Formally, we write
\beqn \label{eq:formala}
a(g) = \int \overline{g(k)} a(k) d^3k  , \quad a^*(g) = \int g(k) a^*(k) d^3k ,
\eeqn
where $a(k)$ and $a^*(k)$ are operator-valued distributions. They satisfy the
so called canonical commutation relations
$$
[a(k), a^*(k') ] = \delta(k - k') , \quad [a^{\#}(k) , a^{\#}(k') ] = 0 \; ,
$$
where $a^{\#}$ stands for  $a$ or $a^*$.

Let  $h$ be a measurable function on $\R^3$. We define the operator $d \Gamma(h)$
in $\FF$, as follows  on vectors $\psi$ in its domain
\begin{equation} \label{eq:defhf}
(d \Gamma(h) \psi)_n(k_1,...,k_n) = \sum_{j=1}^n h(k_j) \psi_n(k_1,...,k_n) \; .
\end{equation}
The domain of $d \Gamma(h)$ consists of all vectors $\psi$ such that $d \Gamma(h) \psi$  is a vector in  $\FF$.
We define the free-field Hamiltonian
$H_f := d \Gamma(\omega)$,  where $\omega(k) := |k|$.
The Hilbert space is given by
$$
\HH := \C^2 \otimes \FF \; .
$$
We define  the following Hamilton operator with coupling parameter $\lambda \in \C$
\begin{equation} \label{eq:defofspinboson}
H_\lambda := \tau  \otimes 1 + 1 \otimes H_f + \lambda  \sigma_x \otimes \phi(f)
\; ,
\end{equation}
where
$$
\phi(f) :=  \int  \frac{1}{\sqrt{\omega(k)}} ( f(k) a^*(k) + \overline{f}(k) a(k) )\frac{d^3k}{4 \pi} \ ,
$$
and
$$
\tau  :=  \sigma_z + 1 = \left( \begin{array}{cc} 2 & 0 \\ 0 & 0  \end{array}
\right) \quad , \quad \sigma_x = \left( \begin{array}{cc} 0 & 1 \\
1 & 0
\end{array} \right) .
$$

 Throughout this paper we shall assume that
 $f/\sqrt{\omega} \in \hh$ and $f/\omega \in \hh$. It is well known that creation and annihilation operators
are infinitesimally small with respect to the free-field Hamiltonian,
see Lemma \ref{thm:estimates1} in the Appendix A.
Thus the operator $H_\lambda$ is a self-adjoint operator on the natural domain of
$H_0$.
The main results of this paper hold under the following hypothesis.

\vspace{0.5cm}

\noindent
(H)  $f \in \hh$ and $\| f \|_\infty < \infty$.

\vspace{0.5cm}

Note that (H) implies that $f/\sqrt{\omega}$ and $f/{\omega} \in \hh$.
 We will use
the following notation
$$
B_r := D_r := \{ z \in \C | |z| < r \} .
$$
A main result  of this paper is the following theorem.

\begin{theorem}  \label{thm:main1} Assume (H). There exists a $\lambda_0 > 0$ such that for all $\lambda \in B_{\lambda_0}$,
 $H_\lambda$ has
an eigenvalue $E(\lambda)$ with eigenvector $\psi(\lambda)$ and eigenprojection $P(\lambda)$ satisfying,
\begin{itemize}
\item[(i)] for $\lambda \in \R \cap B_{\lambda_0}$, $E(\lambda) = {\rm inf}  \sigma ( H_\lambda) $ and $E(\lambda)$ is non-degenerate,
\item[(ii)] $\lambda \mapsto E(\lambda)$ and  $\lambda \mapsto \psi(\lambda)$
are analytic on $B_{\lambda_0}$,
\item[(iii)] $\lambda \mapsto P(\lambda)$ is analytic on $B_{\lambda_0}$ and $P(\lambda)^* = P(\overline{\lambda})$.

\end{itemize}
\end{theorem}

\begin{remark} Since we had the application to non-relativistic qed in mind, we chose Hypothesis (H).
Using a different norm for the Banach spaces one could also show that the conclusion of Theorem \ref{thm:main1} holds under the
assumptions $\omega^{-1 - \epsilon} f \in \hh$ and  $\omega^{-1/2} f \in \hh$, for any $\epsilon > 0$. Moreover, the assertion of the
Theorem \ref{thm:main1} without uniqueness holds if  $\tau$ and $\sigma_x$ are replaced by hermitian $N \times N$ matrices $T$ and $A$, respectively, such that $T$ has a
unique ground state and its eigenprojection $P$ satisfies $P A P = 0$ and $(1-P)A(1-P) = 0$.
\end{remark}

The above result is non-trivial since the ground state energy is not
isolated from the rest of the spectrum. In that situation regular analytic perturbation
theory is not applicable. We prove the existence and analyticity results of  Theorem \ref{thm:main1} using an operator theoretic
renormalization analysis. Since that method  yields the existence of a ground state but not
its uniqueness, we complement the existence
with the following  uniqueness theorem, which we prove in the next section.

\begin{theorem} \label{thm:uniqueness} Suppose $\lambda \in \R$, $\omega^{-1/2} f \in \hh$, and $\omega^{-1} f \in \hh$.
Suppose $E = \inf \sigma (H_\lambda)$ is an eigenvalue. Then $E$ is simple.
\end{theorem}

Once Theorem \ref{thm:main1} has been established, one
knows that  the eigenvalue  of $H_\lambda$ and the associated eigenprojection have
power series expansions with nonzero radius of convergence,
\begin{equation}
{P}(\lambda) = \sum_{n=0}^\infty {P}^{(n)} \lambda^n \quad , \quad {E}(\lambda) = \sum_{n=0}^\infty {E}^{(n)} \lambda^n .
\end{equation}
It is natural to ask whether the expansion coefficients can be obtained
by means of analytic  perturbation theory. This is indeed the case, as we now
outline.  For details see Theorem \ref{thm:perturb1} in Section  \ref{sec:ana}.   We introduce a  cutoff $\sigma > 0$ and
define the infrared regularized Hamiltonian
$$
H_{\lambda,\sigma} := H_{0} + \lambda T_\sigma ,
$$
with $T_\sigma : = \sigma_x\otimes \phi(\chi_\sigma f)$, where   $\chi_\sigma(k) = 1$ if $|k|\geq  \sigma$ and 0 otherwise,
and with $$H_0 := \tau \otimes 1 + 1 \otimes H_f .$$
This effectively turns the ground state energy into an isolated eigenvalue, after
 a trivial part of the Hamiltonian has been factored out. In this situation regular perturbation theory
becomes applicable. It is straight forward to show using analytic perturbation theory, see the proof of Theorem   \ref{thm:perturb1},
that for each $\sigma >0$ there exists a $\lambda_0(\sigma) > 0$
such that  for all $\lambda \in B_{\lambda_0(\sigma)}$, the Hamiltonian
$H_{\lambda,\sigma}$ has an eigenvalue $\widehat{E}_\sigma(\lambda)$  with eigenprojection
 $\widehat{P}_\sigma(\lambda)$. Furthermore,  we have convergent power series expansions (see Kato's book  \cite{K})
\begin{equation} \label{eq:powerperturb}
\widehat{P}_\sigma(\lambda) = \sum_{n=0}^\infty \widehat{P}_\sigma^{(n)} \lambda^n \quad , \quad \widehat{E}_\sigma(\lambda) = \sum_{n=0}^\infty \widehat{E}_\sigma^{(n)} \lambda^n .
\end{equation}
Using  analytic perturbation theory one can show that
\begin{equation} \label{eq:perturbproj0}
\widehat{P}^{(n)}_{\sigma} = -  \sum_{\substack{ \nu_1 + ... + \nu_{n+1} = n, \\ \nu_i \geq 0} } S_\sigma^{(\nu_1)} T_\sigma S_\sigma^{(\nu_2)} ... T_\sigma S_\sigma^{(\nu_{n+1})}
\end{equation}
where
\begin{equation} \label{eq:perturbproj00}
S_\sigma^{(\nu)} = \left\{ \begin{array}{ll} - P_{\Omega_\downarrow}   \quad &, \ \nu = 0 \\    H_{0}^{-\nu}
\overline{P}_{\Omega_\downarrow}   Q_\sigma \quad &, \ \nu \geq 1  , \end{array} \right.
\end{equation}
$P_{\Omega_{\downarrow}}$ denotes the orthogonal projection onto, $\Omega_{\downarrow}$,  the ground state of $H_0$, i.e.,
\begin{equation} \label{eq:perturbproj000}
\Omega_{\downarrow}  :=  \left( \begin{array}{c}  0 \\ 1  \end{array} \right)   \otimes \Omega ,
\end{equation}
$\overline{P}_{\Omega_{\downarrow}} = 1 - P_{\Omega_{\downarrow}}$,
 and
 $Q_\sigma$ denotes the orthogonal projection in $\FF$ onto the natural embedding of
$\FF(\hh_\sigma^{(+)})$  in $\FF$, with $\hh_\sigma^{(+)} :=  L^2(\{ k \in \R^3 | |k| \geq \sigma \})$.
Moreover, the coefficients of the energy expansion can be obtained using the relation
$$
\widehat{E}_\sigma^{(n)} = {\rm tr}(T_\sigma \widehat{P}^{(n-1)}_{\sigma}/n) ,
$$
which can be found in  \cite{K} (Page 80, Eq. (2.34)), and is in fact easy to see.
Analytic perturbation theory does not allow us to control the radius of convergence $\lambda_0(\sigma)$ as $\sigma$ tends to zero.  That is, we cannot
rule out the possibility that $\lambda_0(\sigma) \to 0$ in this limit.  In order  to control the radius of convergence of \eqref{eq:powerperturb}
we have to resort back to renormalization.
Using a continuity argument in connection with the renormalization procedure we obtain the following theorem, which essentially states that the
 ground state energy and the eigenprojection depend continuously on $\sigma$.

 \begin{theorem}  \label{thm:main2} Assume (H). There exists a $\lambda_0 > 0$ such that for all $\lambda \in B_{\lambda_0}$
 and all $\sigma \geq 0$,
 $H_{\lambda,\sigma}$ has
an eigenvalue $E_\sigma(\lambda)$ with eigenvector $\psi_\sigma(\lambda)$ and eigen-projection $P_\sigma(\lambda)$ satisfying (i)-(iii) of
Theorem \ref{thm:main1}. Moreover, $E_\sigma(\lambda)$, $\psi_\sigma(\lambda)$, and  $P_\sigma(\lambda)$, as well
as
the expansion coefficients $E_\sigma^{(n)}$ and $P_\sigma^{(n)}$, in
\begin{align} \label{eq:coeffnocutE}
E_\sigma(\lambda) &= \sum_{n=0}^\infty E_\sigma^{(2n)} \lambda^{2n} \\
P_\sigma(\lambda) &= \sum_{n=0}^\infty P_\sigma^{(n)} \lambda^n  \label{eq:coeffnocutP} ,
\end{align}
are continuous functions of $\sigma \in [0,\infty)$.
\end{theorem}

By the uniqueness of the ground state, we know from Theorem  \ref{thm:main2} and the result from perturbation theory (for details see Theorem   \ref{thm:perturb1})
 that for any $\sigma > 0$ there exists an open ball, $B_{\lambda_0(\sigma)}$,  such that
$\widehat{P}_\sigma(\lambda)= {P}_\sigma(\lambda)$ and  $\widehat{E}_\sigma(\lambda)= {E}_\sigma(\lambda)$ for all $\lambda \in B_{\lambda_0(\sigma)} \cap \R$.
By analytic continuation it follows that $\widehat{P}_\sigma(\lambda)$ and $\widehat{E}_\sigma(\lambda)$ have an analytic extension to a ball, $B_{\lambda_0}$, which is independent
of $\sigma > 0$. Moreover,  these extensions  agree with ${P}_\sigma(\lambda)$ and ${E}_\sigma(\lambda)$ on that ball, respectively.
Thus  we have shown that Theorem  \ref{thm:main2} implies  the following corollary.

\begin{corollary} Assume (H). There exists a $\lambda_0 > 0$ such that for all $\sigma > 0$,
$\widehat{P}_\sigma(\lambda)$ and  $\widehat{E}_\sigma(\lambda)$  have an analytic extension to $B_{\lambda_0}$,
and  on $B_{\lambda_0}$  they agree with  ${P}_\sigma(\lambda)$ and  ${E}_\sigma(\lambda)$.
In particular, for any $\sigma > 0$ we have
$$
\widehat{P}_\sigma^{(n)} =  P_\sigma^{(n)} \quad , \quad  \widehat{E}^{(n)}_\sigma = E_\sigma^{(n)} ,
$$
and the following limits exist,
\begin{equation} \label{eq:coeffconvinsigma0}
\lim_{\sigma \downarrow 0} \widehat{P}_\sigma^{(n)} =       P^{(n)}  , \quad
\lim_{\sigma \downarrow 0} \widehat{E}_\sigma^{(n)} =       E^{(n)}  .
\end{equation}
\end{corollary}

Note that the existence of the limit  \eqref{eq:coeffconvinsigma0} is in view of Equation \eqref{eq:perturbproj0} not obvious.
In particular certain summands in that sum are divergent as $\sigma \to 0$. But the total sum must be
convergent by \eqref{eq:coeffconvinsigma0}. We note the  following remark
illustrating this observation (see Section \ref{sec:ana}).

\begin{remark}   \label{thm:formalperturb00} Consider the sum
 \eqref{eq:perturbproj0}.
For $n \leq 3$ all terms in the sum  converge as $\sigma$ tends to 0.
For $n = 4$, there are terms which diverge. Let
\begin{eqnarray*}
A_\sigma &:=& S_\sigma^{(1)} T_\sigma S_\sigma^{(1)} T_\sigma S_\sigma^{(1)} T_\sigma S_\sigma^{(1)} T_\sigma S_\sigma^{(0)} \\
B_\sigma &:=&  S_\sigma^{(2)}  T_\sigma S_\sigma^{(1)} T_\sigma  S_\sigma^{(0)} T_\sigma S_\sigma^{(1)} T_\sigma S_\sigma^{(0)} .
\end{eqnarray*}
Then $\lim_{\sigma \downarrow 0} A_\sigma$ and  $\lim_{\sigma \downarrow 0} B_\sigma$ diverge but
$\lim_{\sigma \downarrow 0}( A_\sigma +  B_\sigma)$ converges. 
\end{remark}

It would be interesting to understand the nature of the cancellations occurring in  the  coefficients \eqref{eq:perturbproj0}
in a systematic way. Moreover, a sufficiently good  estimate on these coefficients could possibly provide
 an alternate way to prove Theorem  \ref{thm:main1}.

Let us now outline the paper. In section \ref{sec:uni}, we prove that a possible ground
state of the spin-boson model has to be  unique. We use this result to establish
the equivalence of expansion coefficients obtained on the one hand by
perturbation theory and on the other hand by operator theoretic renormalization.

Since  Theorem \ref{thm:main1} is a special case of   Theorem \ref{thm:main2},
we only prove the latter.
The proof  is based on the
operator theoretic renormalization analysis, as outlined in
 \cite{BCFS03}. Sections \ref{sec:smo}--\ref{sec:con}
 are devoted
 to the renormalization analysis.
 In Section \ref{sec:smo}, we introduce  the smooth Feshbach map associated to
 a pair of operators and we review
some of its isospectrality properties, which will be needed later.
In Section \ref{sec:ban}, we define a  Banach space
of integral kernels and show its bijective correspondence
to a subspace of Hamiltonians acting on Fock-space.
In Section \ref{sec:ren:def}, we define the renormalization transformation on the
level of operators. In  Section \ref{sec:ren:ker}, we derive the
induced action of the renormalization transformation on
the space of integral kernels.  In Section \ref{sec:ren:ana}, we show that
the renormalization transformation preserves analyticity and continuity properties
of the integral kernels.
In Section \ref{sec:cod}, we show that the renormalization transformation acts as
a contraction in a subset of the Banach space of integral kernels for which the sum of
the number of creation and annihilation operators is even.
In Section \ref{sec:con}, we construct the eigenvector and the corresponding  eigenvalue
of $H_{\lambda,\sigma}$. We collect certain convergence estimates which are uniform and which will be needed to prove
the analyticity of the ground state. This section contains the main results needed
from the operator theoretic renormalization analysis to prove Theorem  \ref{thm:main2}.
In Section \ref{sec:ini},
we perform the initial two Feshbach transformations. This
allows us to turn the spectral problem of the spin-boson Hamiltonian
into a spectral problem of a new operator involving a sum of normal ordered monomials in creation and annihilation operators where for each summand the total number of creation and annihilation operators is even.
Moreover, we present a basic estimate which allows us to initiate the renormalization procedure.
In Section \ref{sec:prov}, we put all the pieces together and
prove  Theorem \ref{thm:main2}. For this we will mainly use results
stated in Sections \ref{sec:con} and \ref{sec:ini}.

In Section \ref{sec:ana}, we discuss analytic perturbation theory and
Remark  \ref{thm:formalperturb00}. 
In Appendix A, we collect a few basic estimates and identities involving creation and annihilation
operators. In Appendix B, we discuss Wick's theorem and a generalization thereof.

\section{Uniqueness}
\label{sec:uni}
In this section we prove Theorem
\ref{thm:uniqueness}. It involves a special choice of $L^2$ space and a positivity argument.
We first introduce the notation
$$
\Phi(f) = a^*(f) + a(f) \quad , \quad f \in \hh = L^2(\R^3)
$$
and prove a lemma.

\begin{lemma} \label{lemma:zorn} Given $f_0 \in \hh$ then there exists a real Hilbert space $\kk \subset \hh$ with the properties
\begin{itemize}
\item[(1)] $\kk$ is invariant under $\{ e^{- t \omega} | t \geq 0 \}$.
\item[(2)] $\kk +i \kk = \hh$
\item[(3)] $[\Phi(f), \Phi(g) ] = 0$ if $f,g \in \kk$
\item[(4)] $f_0 \in \kk$.
\end{itemize}
\end{lemma}
\begin{proof}
Given $f \in \hh \setminus \{0\}$, let $V_f$ be the real Hilbert space given by
$$
V_f = \{ g(\omega) f \in \hh   \ | \  g \ {\rm a \ real \ measurable \ function } \}
$$
It is easy to see that $V_f$ is closed. We consider the family
$\mathfrak{H}$ of superorthogonal sets of vectors $\{f_j \in \hh \ | \ 0 \leq j < N \}$, $N \leq \infty$, where
superorthogonal means that  $V_{f_j}$ is orthogonal to $f_k$   for all $k\neq j$.
We order the set $\mathfrak{H}$ by inclusion. An easy application of Zorn's lemma shows there is
a maximal element, $ \tau = \{f_j \in \hh\ \ | \ 0 \leq j < N_{\tau} \}$, of $\mathfrak{H}$. Let us write
$$
\kk = \bigoplus_{j=0}^{N_\tau - 1} V_{f_j} ,
$$
where in the direct sum we only allow linear combinations with real coefficients so that $\kk$ is a real Hilbert space.
The properties (1) and (4) are clear while (3) follows from
$$
[ \Phi(f), \Phi(g) ] = 2 i {\rm Im} ( f, g)
$$
To see that $(f,g)$ is real for $f,g \in \kk$ note $(f,g) = 0$ if $f$ and $g$ are in different ${V_{f_j}}$'s while
if $f = h_1(\omega) f_j \in \hh$ and $g = h_2(\omega) f_j \in \hh$ with $h_1$ and $h_2$ real then
$(f,g) = \int h_1(\omega) h_2(\omega) | f_j(\omega) |^2 d^3k$ is clearly real. To see (2) note that if $h \in \hh$ is orthogonal
to $\kk$, then by an approximation argument the same is true of all $g(\omega) h \in \hh$ with $g$ measurable. Thus if $ h \neq 0 ,\  \tau \cup {\{h}\} \in \mathfrak{H}$ and
$\tau$ is not maximal. Thus $ h = 0$. Let
$B = \{ v_j | j \in \N \}$, $v_j \in \kk$, be an orthonormal basis for $\kk$. Then by what we have just proved, $B$ is an orthonormal basis for $\hh$.
If $g \in \hh$ then $g = \sum_{j=1}^\infty (a_j + i b_j ) v_j$ with $a_j, b_j$ real and $\sum_{j=1}^\infty \left( |a_j|^2 + |b_j|^2 \right) < \infty$.
Then
$$
g = \sum_{j=1}^\infty a_j v_j + i \sum_{j=1}^\infty b_j v_j \in \kk + i \kk .
$$
\end{proof}

\vspace{0.5cm}
\noindent
{\it Proof of Theorem \ref{thm:uniqueness}}.
From the lemma and the fact that the closure of the linear span of
$$
\{ e^{ i \Phi(f)} \Omega | f \in \kk \}
$$
is in fact all of Fock space, the spectral theorem shows that $\FF$ is unitarily equivalent to $L^2(Q, d \mu)$ for some probability measure
space $(Q,\mu)$ (we suppress the $\sigma$-algebra). In this representation $\Omega$ is the function 1 and we can take all the $\Phi(f)$'s, $f \in \kk$ to be
real Gaussian random variables with $\Phi(f+g) = \Phi(f) + \Phi(g)$ for $f,g \in \kk$.
Following \cite{sim74}, in the new representation $e^{- t H_f}$ is a positivity preserving operator on $L^2(Q , d \mu )$. Let $U = u \otimes 1$, with $u = e^{-i (\pi/4) \sigma_y}$.
Note that $u \sigma_x u^{-1} = \sigma_z$, $u \sigma_z u^{-1} = - \sigma_x$, and thus taking $\lambda = 1$ without loss of generality,
$$
H := U H_1 U^{-1} = \widetilde{H}_0 + \sigma_z \otimes \phi(f)
$$
where
$$
\tilde{H}_0 = 1 - \sigma_x \otimes 1 + 1 \otimes H_f .
$$
We write $\omega^{-1/2} f = f_0$ so that $\phi(f) = \Phi(f_0)$.
$\widetilde{H}_0$ has a non-degenerate ground state in $\C^2 \otimes \FF$, namely
$$
\left( \begin{array}{c} 1 \\ 1 \end{array} \right) \otimes \Omega =: \Psi_0 .
$$
We note that
$$
\C^2 \otimes \FF \cong L^2( \{ - 1, 1 \} \times Q , dp \otimes d \mu )
$$
where $p(\{1\}) = p(\{-1\}) = 1$.  In this representation, if $f \in L^2(\{-1,1\} \times Q ; dp \otimes d \mu)$ then
\begin{align}
\left( ( \sigma_x \otimes 1 ) f \right)(\pm 1, \cdot ) &= f( \mp 1 , \cdot ) \\
\left( ( 1 \otimes e^{- t H_f}  ) f \right)(s , \cdot ) &= e^{- t H_f} f ( s , \cdot ) \\
\left( ( \sigma_z \otimes \Phi(f_0) ) f \right)(\pm 1, \cdot ) &=  \pm \Phi(f_0) f( \pm 1 , \cdot ).
\end{align}
In addition $e^{- t ( - \sigma_x \otimes 1 )}$ is positivity preserving (clear by expanding the exponential in a power
series) and thus so is $e^{- t \widetilde{H}_0}$. The operator $\widetilde{H}_0$ has a non-degenerate ground state
given by the function 1. A direct application of Theorem XIII.43 of \cite{reesim4} then implies
$L^\infty(\{-1,1\} \times Q) \cup \{ e^{- \widetilde{H}_0 } \}$ acts irreducibly in $L^2(\{-1,1\} \times Q)$. Since
$\sigma_z \otimes \Phi(f_0)$ is infinitesimally $\widetilde{H}_0$ bounded, Theorem XIII.45 of \cite{reesim4} then shows
$L^\infty(\{-1,1\} \times Q ) \cup \{ e^{-H} \}$ acts irreducibly in $L^2(\{-1,1\} \times Q )$. Finally according
to Theorem XIII.43 of \cite{reesim4}, if $E = \inf \sigma(H)$ is an eigenvalue of $H$, then it is non-degenerate.
\qed

\vspace{0.5cm}

We end this section with a proof that $\sigma(H_{\lambda})$ is a half line.  In fact using the ideas developed in Lemma \ref{lemma:zorn} we prove a bit more:

\begin{proposition} \label{prop:halfline}
Suppose $\lambda \in \R$, $\omega^{-1/2} f \in \hh$, and $\omega^{-1} f \in \hh$.
Let $E = \inf \sigma (H_\lambda)$. Then $\sigma_{ac}(H_{\lambda}) = [E, \infty)$.

\end{proposition}

Here $\sigma_{ac}(H_{\lambda})$ is the absolutely continuous spectrum of $H_{\lambda}$.

\begin{proof}

Using the notation of Lemma \ref{lemma:zorn} and Theorem \ref{thm:uniqueness},  let $\hh_1 = V_{f_0} + iV_{f_0}$.  Using polar coordinates $(u,t)$ where $u \in S^2$ and $t \in (0, \infty)$ we have $f_0(t) \in L^2(S^2)$ for a.e. $t$.  According to \cite{Dix} the space orthogonal to $f_0(t)$ in $L^2(S^2)$ has an orthonormal basis $\{e_j(t)| j\in \N\}$ where the vectors $e_j(t)$ are measurable in the variable $t$.  The space of functions $\sum_{n = 1}^{\infty}g_n(t)e_n(t)$  with $g_n \in L^2(t^2dt)$ and $\sum_{n = 1}^{\infty}\int_0 ^\infty |g_n(t)|^2t^2dt < \infty$ is exactly $\hh_2:= \hh_1^{\bot}$.  Fix an orthonormal basis $\{\hat{e}_j| j\in \N\}$ for $L^2(S^2)$ and note that defining $u(t) :  f_0(t)^{\bot} \rightarrow L^2(S^2)$ by linearity and continuity from $u(t)\e_j(t): = \hat {e}_j$, then $u(t)$ is unitary and $U$ given by $Ug(t) = u(t)g(t)$ is a unitary map of $\hh_2$ onto $\hh = L^2(\R^3)$.\\
We now factor the Hilbert space $\HH = \C^2 \otimes \FF$  as  $\HH = \C^2 \otimes \FF(\hh_1)\otimes \FF(\hh_2)$.  With respect to this factorization we write
$H_\lambda = \tau  \otimes 1 \otimes 1 + 1 \otimes H_{f}^{(1)} \otimes 1 + \lambda \sigma_x \otimes \phi(f)\otimes 1 + 1 \otimes 1 \otimes H_{f}^{(2)}$
where $H_f^{(j)}$ is the restriction of $H_f$ to $\FF(\hh_j)\cap {D(H_f)}$.  We define $\tilde{H}_\lambda$ by the equation $H_\lambda = \tilde{H}_\lambda \otimes 1 + 1 \otimes H_f^{(2)}$.  Let $\Gamma(U): \FF(\hh_2) \rightarrow \FF(\hh)$ be the unitary operator satisfying $\Gamma(U)\Omega = \Omega$ and $\Gamma(U)S_n(g_1\otimes\cdots\otimes g_n) = S_n(Ug_1\otimes\cdots\otimes Ug_n)$.  It is easy to see that $\Gamma(U)H_f^{(2)} = H_f\Gamma(U)$ so that $H_f^{(2)}$ restricted to $\Omega^\bot$ is absolutely continuous.  Since clearly $E = \inf \sigma(\tilde{H}_\lambda)$ and the convolution of an absolutely continuous measure with another measure is absolutely continuous, the proposition easily follows.
\end{proof}

\section{Smooth Feshbach Property}
\label{sec:smo}

In this section we follow  \cite{BCFS03,GH08}. We introduce the Feshbach map and state its basic isospectrality
properties.  This will be needed to define the renormalization transformation and to construct the ground state
and the ground state energy.

Let $\chi$ and $\overline{\chi}$ be commuting, nonzero bounded operators, acting on a separable Hilbert space $\HH$
and satisfying $\chi^2 + \overline{\chi}^2=1$. A {\it Feshbach pair} $(H,T)$ for $\chi$ is a pair of
closed operators with the same domain,
$$
H,T : D(H) = D(T) \subset \HH \to \HH
$$
such that $H,T, W := H-T$, and the operators
\begin{align*}
&W_\chi := \chi W \chi ,  &  &W_{\overline{\chi}} := \overline{\chi} W  \chib  \\
&H_\chi :=T + W_\chi   ,  &  &H_{\overline{\chi}} := T +  W_{\chib} ,
\end{align*}
defined on $D(T)$ satisfy the following assumptions:
\begin{itemize}
\item[(a)] $\chi T \subset T \chi$ and $\chib T \subset T \chib$,
\item[(b)] $T, H_{\chib} : D(T) \cap \ran \chib \to \ran \chib$ are bijections with bounded inverse,
\item[(c)] $\chib H_{\chib}^{-1} \chib W \chi : D(T) \subset \HH \to \HH$ is a bounded operator.
\end{itemize}
\begin{remark} \label{rem:abuse}
By abuse of notation we write  $ H_{\chib}^{-1} \chib$ for $ \left( H_{\chib} \upharpoonright \ran \chib \cap D(T) \right)^{-1}   \chib$ and
likewise  $ T^{-1} \chib$ for $ \left( T \upharpoonright \ran \chib \cap D(T)\right)^{-1}   \chib$.
\end{remark}
An operator $A:D(A) \subset \HH \to \HH$ is called {\it bounded invertible} in a subspace $V \subset \HH$
($V$ not necessarily closed), if $A: D(A) \cap V \to V$ is a bijection with bounded inverse.
Given a Feshbach pair $(H,T)$ for $\chi$, the operator
\begin{align*}
&F_\chi(H,T) := H_\chi - \chi W \chib H_{\chib}^{-1} \chib W \chi
\end{align*}
on $D(T)$ is called the { \it Feshbach map of} $H$. The mapping $(H,T) \mapsto F_\chi(H,T)$ is called the
{\it Feshbach map}. The auxiliary operators
\begin{align*}
&  Q_\chi := Q_\chi(H,T) := \chi - \chib H_{\chib}^{-1} \chib W \chi , \\
& Q_\chi^{\#} := Q_\chi^\#(H,T) :=  \chi -  \chi W  \chib H_{\chib}^{-1} \chib   .
\end{align*}
are by conditions (a), (c), bounded, and $Q_\chi$ leaves $D(T)$ invariant. The Feshbach map is
isospectral in the sense of the following theorem.
\begin{theorem} \label{thm:fesh}
Let $(H,T)$ be a Feshbach pair for $\chi$ on a  Hilbert space $\HH$. Then the following holds:
\begin{itemize}
\item[(i)] Let $V$ be subspace with $\ran \chi \subset V \subset \HH$,
\begin{align*}
&T : D(T) \cap V \to V , &   {\rm and} &  & \chib T^{-1} \chib V \subset V .
\end{align*}
Then $H :D(H) \subset \HH  \to \HH$ is bounded invertible if and only if $F_\chi(H,T) : D(T) \cap V \to V$ is bounded invertible
in $V$. Moreover,
\begin{align*}
& H^{-1}  = Q_\chi F_\chi(H,T)^{-1} Q^{\#}_\chi + \chib H_{\chib}^{-1} \chib , \\
& F_\chi(H,T)^{-1} = \chi H^{-1}\chi + \chib T^{-1} \chib \; .
\end{align*}
\item[(ii)]  $\chi \ker H \subset \ker F_\chi(H,T)$ and $Q_\chi \ker F_\chi(H,T) \subset \ker H$. The mappings
\begin{align*}
&\chi : \ker H \to \ker F_\chi(H,T) , \\
&Q_\chi : \ker F_\chi(H,T)  \to \ker H  ,
\end{align*}
are linear isomorphisms and inverse to each other.
\end{itemize}
\end{theorem}

The proof of Theorem \ref{thm:fesh} can be found in \cite{GH08}. The next lemma
gives sufficient conditions for two operators to be a Feshbach pair. It follows
from a Neumann expansion, \cite{GH08}.

\begin{lemma} \label{fesh:thm2}
Conditions {\rm (a), (b)}, and {\rm (c)} on Feshbach pairs are satisfied if:
\begin{itemize}
\item[(a')] $\chi T \subset T \chi$ and $\chib T \subset T \chib$,
\item[(b')] $T$ is bounded invertible in $\ran \chib$,
\item[(c')] $\| T^{-1} \chib W \chib \| < 1$,  $\|  \chib W T^{-1} \chib \| < 1$, and $T^{-1} \chib W \chi$ is a bounded operator.
\end{itemize}
\end{lemma}

Moreover we need the identity given in the following Lemma. The identity follows after some manipulations
of the definitions. A proof can be found for example in \cite{GH08}.

\begin{lemma} \label{lem:feshbasic} Let $(H,T)$ be a Feshbach pair for $\chi$. Then $H Q_{\chi} = \chi F_\chi(H,T)$ on $D(T)$.
\end{lemma}

\section{Banach Spaces of Hamiltonians}
\label{sec:ban}

In this section we introduce Banach spaces of integral kernels, which
parameterize  certain subspaces of the space of bounded operators on Fock space.
These subspaces are  suitable to study an iterative application of the Feshbach map
and to formulate the contraction property. We mainly follow  the exposition in \cite{BCFS03}.
However, we use a different class of  Banach spaces. 
The renormalization transformation will be defined on operators acting on the reduced Fock space
$$
\mathcal{H}_{\rm red}:= P_{\rm red} \FF ,
$$
where we introduced the notation  $P_{\rm red}:=  \chi_{[0,1]}(H_f)$.
We will investigate bounded operators in $\mathcal{B}(\mathcal{H}_{\rm red})$ of the form
\beqn \label{eq:standardexp}
T  +  W  ,
\eeqn
where $T = t(H_f)$ with  $t \in C^1([0,1])$  and the interaction term $W$ is  given formally by
\beqn \label{eq:sum}
W[w] := \sum_{m+n \geq 1} H_{m,n}(w_{m,n})
\eeqn
with
\beqn \label{eq:defhmn11}
H_{m,n}(w_{m,n}) := P_{\rm red} \int_{B_1^{m+n}} \frac{ dK^{(m,n)}}{|K^{(m,n)}|^{1/2}} a^*(k^{(m)}) w_{m,n}(H_f, K^{(m,n)}) a(\tilde{k}^{(n)}) P_{\rm red} \upharpoonright \mathcal{H}_{\rm red}
  \; ,
\eeqn
where $w_{m,n} \in L^\infty([0,1]\times B_1^m \times B_1^n)$  is an integral kernel and $w = (w_{m,n})_{m,n \in \N_0}$ a sequence of integral kernels.
We have used and will henceforth use the following notation.
\begin{align*}
& B_1 := \{ x \in \R^3 | |x|< 1 \} \\
& k^{(m)} := (k_1, ... , k_m ) \in \R^{3m}  , \quad \widetilde{k}^{(n)} := (\widetilde{k}_1, ... , \widetilde{k}_n ) \in \R^{3n} , \\
&  K^{(m,n)}  := (k^{(m)}, \widetilde{k}^{(n)}) , \quad d K^{(m,n)}  :=  \frac{d k^{(m)}}{(4\pi)^m}   \frac{d  \widetilde{k}^{(n)}}{(4 \pi)^n}  ,  \\
&   d k^{(m)}   :=  \prod_{i=1}^m d^3 k_i , \quad  d  \widetilde{k}^{(n)}  := \prod_{j=1}^n d^3 \widetilde{k}_j , \\
& a^*(k^{(m)}) :=  \prod_{i=1}^m a^*(k_i) , \quad a(\widetilde{k}^{(m)}) :=  \prod_{j=1}^m a(\widetilde{k}_j) \\
& | K^{(m,n)}| := | k^{(m)} | \cdot | \widetilde{k}^{(n)}| , \quad | k^{(m)} | := |k_1| \cdots |k_m | , \quad  | \tilde{k}^{(m)} | := |\tilde{k}_1| \cdots |\tilde{k}_m | , \\
& \Sigma[k^{(m)}] := \sum_{i=1}^m |k_i |  \; .
\end{align*}
For $w_{0,0} \in L^\infty([0,1])$, we define
$$
H_{0,0}(w_{0,0}) := w_{0,0}(H_f) .
$$
Note that \eqref{eq:defhmn11} is understood in the sense of forms, i.e. for $\psi, \phi$ two vectors in $\HH_{\rm red}$ with finitely many particles we define,
\begin{align} \label{eq:defhmnrig}
\langle \psi, H_{m,n}(w_{m,n}) \phi \rangle &=  \int_{B_1^{m+n}} \frac{ dK^{(m,n)}}{|K^{(m,n)}|^{1/2}} \left\langle a(k^{(m)})P_{\rm red} \psi,  w_{m,n}(H_f, K^{(m,n)}) a(\tilde{k}^{(n)}) P_{\rm red} \phi \right\rangle .
\end{align}
A  vector $\psi \in \FF$ is said to have finitely many particles if only  finitely many
$\psi_n$ are nonzero. For the precise meaning of the vectors   $a(k^{(m)})P_{\rm red} \psi$  and $a(\tilde{k}^{(n)}) P_{\rm red} \phi$ see
\eqref{eq:defofa} in Appendix A.
As shown in the proof of the next lemma, Lemma \ref{lem:operatornormestimates}, the  quadratic form \eqref{eq:defhmnrig} is bounded and thus
defines a bounded  operator.
Note that in view of the pull-through formula, Lemma \ref{lem:pullthrough}, the operator in \eqref{eq:defhmn11} is equal to the restriction of
\beqn \label{eq:defintegralkernel}
\int_{B_1^{m+n}} \frac{ dK^{(m,n)}}{|K^{(m,n)}|^{1/2}} a^*(k^{(m)})  \chi(H_f +  \Sigma[k^{(m)}] \leq 1 ) w_{m,n}(H_f, K^{(m,n)})
\chi(H_f +  \Sigma[\tilde{k}^{(n)}] \leq 1) a(\tilde{k}^{(n)} ) \;
\eeqn
to the subspace $\HH_{\rm red}$.
Thus we can restrict attention to integral kernels $w_{m,n}$ which are  essentially supported on the set
\begin{eqnarray*}
Q_{m,n} &:=& \left\{ ( r , K^{(m,n)}) \in [0,1] \times B_1^{m+n}  \ | \ r  \leq 1 - \max(\Sigma[k^{(m)}],
\Sigma[\widetilde{k}^{(m)}]) \right\}  \quad , \quad m + n \geq 1 ,  \\
Q_{0,0} &:=& [0,1] .
\end{eqnarray*}
Moreover, note that  integral kernels can always be assumed to be symmetric. That is, they lie in the range of the symmetrization operator,
which is defined as follows,
\begin{eqnarray} \label{eq:symmetrization}
w_{M,N}^{({\rm sym})}(r,K^{(M,N)}) := \frac{1}{N!M!} \sum_{\pi \in S_M} \sum_{\widetilde{\pi} \in S_N} {w}_{M,N}(r,
k_{\pi(1)},\ldots,k_{\pi(N)}, \widetilde{k}_{\widetilde{\pi}(1)},\ldots,\widetilde{k}_{\widetilde{\pi}(M)}).
\end{eqnarray}
To be able to relate the integral kernels with bounded operators we need the following lemma.
\begin{lemma} For $w_{m,n} \in L^\infty([0,1] \times B_1^m \times B_1^n)$\label{lem:operatornormestimates} we have
\beqn \label{eq:operatornormestimate1}
\|H_{m,n}(w_{m,n}) \|_{\rm op}  \leq  \| w_{m,n} \|_{\infty}  ( n! m!)^{-1/2}  \; ,
\eeqn
where $\| \cdot \|_{\rm op}$ denotes the operator norm of  $\HH_{\rm red}$.
\end{lemma}
\begin{proof}
For $\psi, \phi \in \mathcal{H}_{\rm red}$ with finitely many particles  we estimate  by
means of the Cauchy-Schwarz inequality,
\begin{eqnarray*}
| \langle \psi , H_{m,n}(w_{m,n}) \phi  \rangle  |  &&\leq \| w_{m,n} \|_\infty
\int_{S_{m,n}} \frac{d K^{(m,n)}}{|K^{(m,n)}|^{1/2}} \| a(k^{(m)}) \psi \|
\| a(\tilde{k}^{(n)}) \varphi \| \\
&&\leq \| w_{m,n} \|_\infty  D_m(\psi)^{1/2} D_n(\varphi)^{1/2}
\left[ \int_{S_{m,n}} \frac{d K^{(m,n)}}{|K^{(m,n)}|^{2}}  \right]^{1/2}
\end{eqnarray*}
where
$$
D_m(\psi) := \int_{B_1^m} |k^{(m)}| \| a(k^{(m)}) \psi \|^2 d k^{(m)}  \; ,
$$
and $S_{m,n} = \{ K^{(m,n)} \in  B_1^{m+n}  \ | \Sigma[k^{(m)}] \leq 1 ,
\Sigma[\widetilde{k}^{(m)}] \leq 1 \}$. By Corollary \ref{lem:multanihiest2} we have
$$
D_m(\psi) \leq  \| H_f^{m/2}   \psi \|^2   \leq \| \psi \|^2 .
$$
We calculate
\begin{equation} \label{eq:intofwKminus2}
 \int_{S_{m,n}} \frac{d K^{(m,n)}}{|K^{(m,n)}|^{2}} =  \frac{1}{{n! m!}} .
\end{equation}
Collecting estimates  the lemma follows.
\end{proof}

The renormalization procedure will involve kernels which lie in the following Banach spaces.
We shall identify the space $L^\infty(B_1^{m+n}; C[0,1])$ with a subspace   of $L^\infty([0,1]\times B_1^{m+n})$ by
setting $w_{m,n}(r,K^{(m,n)}) := w_{m,n}(K^{(m,n)})(r)$ for $w_{m,n} \in L^\infty(B_1^{m+n}; C[0,1])$. The norm in $L^\infty(B_1^{m+n}; C[0,1])$
is given by
$$
\| w_{m,n} \|_{\underline{\infty}} := \esssup_{K^{(m,n)} \in B_1^{m+n}}  {\rm sup}_{r \geq 0}| w_{m,n}(K^{(m,n)})(r)|  .
$$
We note that for $w \in L^\infty(B_1^{m+n}; C[0,1])$ we have $ \| w \|_\infty  \leq  \| w \|_{\underline{\infty}}$.
Conditions (i) and (ii) of the
following definition are needed for the injectivity property stated in Theorem \ref{thm:injective}, below.

\begin{definition} \label{def:wgartenhaag}
We define  $\WW_{m,n}^\#$ to be the Banach space consisting of functions $w_{m,n} \in L^\infty(B_1^{m+n};C^1[0,1])$ satisfying the following properties:
\begin{itemize}
\item[(i)]  $ w_{m,n} (1 - \chi_{Q_{m,n}} ) = 0$
\item[(ii)]  $w_{m,n}(\cdot,k^{(m)}, \widetilde{k}^{(n)})$ is totally
symmetric in the variables $k^{(m)}$ and $\widetilde{k}^{(n)}$
\item[(iii)] the following norm is finite
$$
\| w_{m,n} \|^\# := \| w_{m,n} \|_{\underline{\infty}} + \| \partial_r w_{m,n} \|_{\underline{\infty}} .
$$
\end{itemize}
Hence  for almost all  $K^{(m,n)} \in B_1^{m+n}$ we have $w_{m,n}(\cdot,K^{(m,n)}) \in C^1[0,1]$, where
the derivative is denoted by $\partial_r w_{m,n}$.
For $0<\xi < 1$, we define the Banach space
$$
 \mathcal{W}^\#_{\xi} := \bigoplus_{(m,n) \in \N_0^2  } \mathcal{W}_{m,n}^\#  \
$$
to consist of all sequences $w =( w_{m,n})_{m,n \in \N_0}$ satisfying
$$
\| w \|_\xi^\# := \sum_{(m,n)\in \N_0^2} \xi^{-(m+n)} \| w_{m,n}\|^\# < \infty   .
$$
\end{definition}
\begin{remark}
{ We shall also use the norm  $\| w_{m,n} \|^\#$  for any integral kernel  $w_{m,n} \in L^\infty(B_1^{m+n};C^1[0,1])$.
 Note that $\| w_{m,n}^{({\rm sym})} \|^\# \leq \| w_{m,n} \|^\#$.}
\end{remark}

Given $w \in \mathcal{W}_\xi^\#$, we
write   $w_{\geq r}$  for the vector in $\mathcal{W}_\xi^\#$  given by
$$
(w_{\geq r})_{(m,n)}  :=  \left\{ \begin{array}{ll} w_{m,n} & , \quad  {\rm if}  \ m+n \geq r \\ 0 & , \quad {\rm otherwise} . \end{array} \right.
$$
We will use the following balls to define the  renormalization transformation
\begin{align*}
\mathcal{B}^\#(\alpha,\beta,\gamma) := \left\{ w \in \mathcal{W}_\xi^\# \left|  \| \partial_r w_{0,0}  - 1 \|_\infty \leq \alpha , \
|w_{0,0}(0) |  \leq \beta
, \ \| w_{\geq 1} \|_{\xi}^\# \leq \gamma   \right. \right\} .
\end{align*}
For $w \in \mathcal{W}^\#_{\xi}$, it is easy to see  using \eqref{eq:operatornormestimate1} that the sum
$$
H(w) := \sum_{m,n} H_{m,n}(w)
$$
with $H_{m,n}(w) := H_{m,n}(w_{m,n})$
converges in operator norm with bound
\beqn \label{eq:opestimatgeq12}
\| H(w) \|_{\rm op} \leq  \| w\|_\xi^\# .
\eeqn
In fact using \eqref{eq:operatornormestimate1}, we see that
\beqn \label{eq:opestimatgeq1}
\| H(w_{\geq r}) \|_{\rm op} \leq   \xi^r   \| w_{\geq r} \|_\xi^\#  .
\eeqn
To identify  $H(w)$ with expressions of the form \eqref{eq:standardexp} we will set
$T[w] := w_{0,0}(H_f)$.
We will use the following theorem from  \cite{BCFS03}.
Note that in the  theorem  stated in  \cite{BCFS03}
 the integral kernels are not restricted to $Q_{m,n}$. But this seems to be necessary
for the injectivity property. We  sketch  a proof along the lines of   \cite{BCFS03}.

\begin{theorem} \label{thm:injective} The map $H : \WW_\xi^\# \to \mathcal{B}(\HH_{\rm red})$ is injective and  bounded.
\end{theorem}

\begin{proof} The boundedness follows from \eqref{eq:opestimatgeq12}.
Assume that $H(w)=0$. We want to show that this implies that $w=0$. First we show that $H(w)=0$ implies $w_{0,0}=0$. To show this pick a non-negative function
$f \in C_0^\infty(\R^3)$ with ${\rm supp} f \subset B_1$ and $\int f^2(x) d^3 x = 1$. Define $f_{\epsilon,k}(x) := \epsilon^{-3/2} f(\epsilon^{-1}(x-k))$ for $k \in B_1$.
A straight forward computation gives
\begin{eqnarray*}
\lefteqn{ \langle a^*(f_{\epsilon,k}) \Omega, H(w) a^*(f_{\epsilon,k}) \Omega \rangle } \\
&& = \int_{B_1}  f_{\epsilon,k}^2(x)  w_{0,0}(|x|) d^3 x  +
 \int_{B_1^2}  f_{\epsilon,k}(x_1)  w_{1,1}(0,x_1,\widetilde{x}_1)f_{\epsilon,k}(\widetilde{x}_1) \frac{d X^{(1,1)}}{|X^{(1,1)}|^{1/2}} .
\end{eqnarray*}
As $\epsilon$ tends to zero, the second term on the right hand side converges to zero, because $f_{\epsilon,k}$ converges
weakly to zero in  $L^2(B_1)$ and the integral operator $w_{1,1}(0,x_1,\widetilde{x}_1)/ |X^{(1,1)}|^{1/2}$ is compact.
The first term converges in this limit to $w_{0,0}(|k|)$. Since by assumption $H(w)=0$, this implies
$w_{0,0}=0$.   To show that for $m+n \geq 0$ also $w_{m,n}$ has to be zero  we proceed by induction. We prove
that $w_{m,n} = 0$ for all $m+n \leq l-1$ implies that $w_{m,n}=0$ for $m+n = l$. Thus fix $(\widehat{m},\widehat{n})$ with $\widehat{m}+\widehat{n}=l$.
Let $g_1,...,g_{\widehat{m}},h_1,...,h_{\widehat{n}} \in L^2(B_1)$ and set
\begin{align*}
\psi = a^*(g_1) \cdots a^*(g_{\widehat{m}}) a^*(f_{\epsilon,k}) \Omega , \quad
\phi =a^*(h_1) \cdots a^*(h_{\widehat{n}}) a^*(f_{\epsilon,k}) \Omega .
\end{align*}
\begin{align*}
\langle \psi  , H(w) \phi \rangle =  \langle \psi , H_{\widehat{m},\widehat{n}} \phi \rangle +  \langle \psi, H_{\widehat{m}+1,\widehat{n}+1} \phi \rangle ,
\end{align*}
where we used that by the induction hypothesis $w_{m,n}=0$ if $m+n \leq l-1$.
As $\epsilon$ tends to zero, the second term on the right hand side converges to zero, because $w-\lim_{\epsilon \downarrow 0} f_{\epsilon,k}=0$
 in  $L^2(B_1)$. The first term on the right hand side converges in this limit to $(\widehat{m}+1)!(\widehat{n}+1)!$ times
$$
\int_{B_1^{\widehat{m}+\widehat{n}}} \frac{ dX^{(\widehat{m},\widehat{n})}}{|X^{(\widehat{m},\widehat{n})}|^{1/2}}
\overline g_1(x_1) \cdots \overline g_{\widehat{m}}(x_{\widehat{m}})  w_{\widehat{m},\widehat{n}}(|k|, X^{(\widehat{m},\widehat{n})}) h_1(\widetilde{x}_1) \cdots h_{\widehat{n}}(\widetilde{x}_{\widehat{n}})  ,
$$
other contributions to $\langle \psi , H_{\widehat{m},\widehat{n}} \phi \rangle$ vanish in this limit, again  because $w-\lim_{\epsilon \downarrow 0} f_{\epsilon,k}=0$ in  $L^2(B_1)$.
Since $H(w)=0$ and the choice of the functions $h_i$ and $g_i$ and the choice of $k \in B_1$ was arbitrary, we conclude that $w_{\widehat{m},\widehat{n}} = 0$. This shows
Theorem \ref{thm:injective}. \end{proof}

The renormalization transformation will be defined on kernels which depend on a spectral parameter. To account for
that, we introduce the following Banach space.

\begin{definition}
Let   $\WW_\xi$ denote  the Banach space consisting of strongly analytic functions on $D_{1/2}$ with values
in $\WW_\xi^\#$ and norm given by
$$
\| w(\cdot ) \|_\xi := \sup_{z \in D_{1/2}} \| w(z) \|_\xi^\# .
$$
\end{definition}
For $w \in \WW_\xi$ we will use the notation $w_{m,n}(z, \cdot) := w(z)_{m,n}(\cdot)$.
We extend the definition of $H(\cdot)$  to $\WW_\xi$ in the natural way: for $w \in \WW_\xi$, we set
$$
\left( H(w) \right) (z) := H(w(z))
$$
and likewise for $H_{m,n}(\cdot), \, W[\cdot], \, T[\cdot]$.
The renormalization transformation will be defined on the following balls in $\mathcal{W}_\xi$,
\begin{align*}
\mathcal{B}(\alpha,\beta,\gamma) := \left\{ w \in \mathcal{W}_\xi \left| \sup_{z \in D_{1/2}}  \| \partial_r w_{0,0}(z)  - 1 \|_\infty \leq \alpha , \
\sup_{z \in D_{1/2}} | w_{0,0}(z,0) + z | \leq \beta
, \ \| w_{\geq 1} \|_{\xi} \leq \gamma  \right. \right\} .
\end{align*}
Note that this set defines a basis of neighborhoods of the  point $w^*$ satisfying $H(w^*(z)) = H_f - z$, i.e.,  ${w^*}_{0,0}(z,r) = r - z $ and ${w^*_{\geq 1}} = 0$ , since
\begin{align*}
\{ w \in \WW_\xi | \| w - w^* \|_\xi \leq \epsilon  \}   \subset \mathcal{B}(\epsilon, \epsilon,  \epsilon)
 \subset \{ w \in \WW_\xi | \| w - w^* \|_\xi \leq  4 \epsilon  \} .
\end{align*}
To state the contraction property of the renormalization transformation in Section \ref{sec:contract}, we will need to introduce the  balls of  even integral kernels
$$
\mathcal{B}_0(\alpha,\beta,\gamma)  := \{ w \in \mathcal{B}(\alpha,\beta,\gamma)  | w_{m,n} = 0 \quad {\rm if } \quad m+n = {\rm odd} \ \} .
$$
We say that a kernel $w \in \WW_\xi$ is symmetric if $w_{m,n}(\overline{z}) = \overline{w_{n,m}(z)}$. Note that because of  Theorem
\ref{thm:injective} we have for $w \in \WW_\xi$,
\begin{equation} \label{eq:symmetricsa}
w \ {\rm  is \ symmetric } \   \Leftrightarrow \ H(w(\overline{z}))= H(w(z))^* .
\end{equation}
To show  the continuity of the ground state and the ground state energy as a function of the infrared cutoff we need to introduce a coarser norm

in $\WW_{m,n}^\#$. The supremum norm is to fine. For $w_{m,n} \in L^\infty(B_1^{m+n};C[0,1])$ we define the norm
$$
\| w_{m,n} \|_2  := \left[ \int \frac{dK^{(m,n)}}{|K^{(m,n)}|^2} \sup_{r \in [0,1]} | w_{m,n}(r,K^{(m,n)}) |^2 \right]^{1/2} .
$$
Observe that by  \eqref{eq:intofwKminus2}  we have
\begin{equation} \label{eq:infinity2ineq}
\| w_{m,n} \|_2 \leq \frac{ \|w_{m,n} \|_{\infty} }{\sqrt{ n! m!}}   
\end{equation}
We have the following lemma.
\begin{lemma} For $w_{m,n} \in L^\infty( B_1^m \times B_1^n; C[0,1])$\label{lem:operatornormestimates2} we have
\beqn \label{eq:operatornormestimate12}
\|H_{m,n}(w_{m,n}) \|_{\rm op}  \leq  \| w_{m,n} \|_{2}    \; .
\eeqn
\end{lemma}
\begin{proof}
We use the notation and the estimates used in the proof of Lemma  \ref{lem:operatornormestimates}.
For $\psi, \phi \in \mathcal{H}_{\rm red}$ with finitely many particles  we estimate
\begin{eqnarray*}
\lefteqn{\left| \left\langle \psi,  H_{m,n}(w_{m,n})  \phi \right\rangle \right| }\\
&&\leq  \int_{B_1^{m+n}} \frac{ \sup_{r \in [0,1]} |  w_{m,n}(r,K^{(m,n)}) |}{|K^{(m,n)}|^{1/2}} \| a(k^{(m)})\psi \| \| a(\widetilde{k}^{(n)}) \phi\| d K^{(m,n)} \\
&&\leq  D_m(\psi)^{1/2} D_n(\phi)^{1/2} \left[  \int_{B_1^{m+n}} \frac{ \sup_{ r \in [0,1]}  |  w_{m,n}(r,K^{(m,n)} )|^2}{|K^{(m,n)}|^{2}}  d K^{(m,n)}  \right]^{1/2} . 
\end{eqnarray*}
Now observe that by Corollary \ref{lem:multanihiest2} we have
$ D_m(\psi)  \leq \| \psi \|^2$.
\end{proof}
\begin{definition}
Let $S$ be topological space. We say that the mapping $w : S \to \WW_\xi^\#$ is componentwise $L^2$--continuous (c-continuous) if for all
$m,n \in \N_0$  the map $s \mapsto w(s)_{m,n}$ is continuous with respect to $\| \cdot \|_2$, that is
$$
\lim_{s \in S, s \to s_0}  \left\| w(s_0)_{m,n} -  w(s)_{m,n} \right\|_2 =  0
$$
for all $s,s_0 \in S$.
\end{definition}
The above notion of continuity for integral kernels, yields continuity of the associated operators with respect to the norm topology. This is
the content of the following lemma.
\begin{lemma}  \label{lem:pcontop}
Let $w: S \to \WW_\xi^\#$ be c-continuous and uniformly bounded, i.e., $\sup_{s \in S} \| w(s) \|^\#_\xi < \infty$. Then
$H(w(\cdot)) : S \to \mathcal{B}(\HH_{\rm red})$ is continuous, with respect to the norm topology.
\end{lemma}
\begin{proof}
From Lemma \ref{lem:operatornormestimates2}
it follows that   $H_{m,n}(w(s)) \stackrel{\| \cdot \|_{\rm op}}{\longrightarrow}   H_{m,n}(w(s_0))$ as $s$ tends to $s_0$.
The lemma now follows from a simple argument using  the estimate \eqref{eq:opestimatgeq1}
 and the uniform bound on $w(\cdot)$.
\end{proof}

\section{Renormalization Transformation: Definition}
\label{sec:ren:def}

In this section we define the renormalization transformation as in \cite{BCFS03}. It is a combination of
the Feshbach transformation which cuts out higher boson energies, a rescaling of the resulting operator so that it acts on the fixed subspace $\HH_{\rm red}$ and a conformal transformation of the spectral parameter.

Let  $0<\xi<1$ and $0 < \rho < 1$.
For $w \in \mathcal{W}_\xi$ we  define the analytic function
$$E_\rho[w](z) := \rho^{-1} E[w](z) := - \rho^{-1} w_{0,0}(z,0) = - \rho^{-1} \langle \Omega , H(w(z)) \Omega \rangle$$
 and the set
$$
U[w] := \{ z \in D_{1/2}  | | E[w](z) | <  \rho  / 2 \} .
$$
\begin{lemma} \label{renorm:thm3} Let $0< \rho\leq 1/2$. Then for all $w \in \mathcal{B}(\cdot, \rho/8, \cdot)$, we have
$$
D_{3\rho/8}\subset U[w]\subset D_{5\rho/8} ,
$$
$| \partial_z E[w](z)- 1 | \leq 4 \rho ( 4 - 5 \rho)^{-2} \leq 8/9$ for all $z \in U[w]$, and $E_{\rho}[w]: U[w] \to D_{1/2}$ is an analytic  bijection.
\end{lemma}
The lemma follows directly from the following lemma by choosing the appropriate values for the corresponding constants
($r = \rho/2$, $\epsilon = \rho / 8$).
\begin{lemma} \label{renorm:thm3A} Let $0 < \epsilon < 1/2$, and
let $E : D_{1/2} \to \C$ be an analytic function which satisfies
\begin{align*}
\sup_{z \in D_{1/2}} |E(z) - z | \leq \epsilon .
\end{align*}
Then for any $r > 0 $ with $r + \epsilon < 1/2$ the following is true.
\begin{itemize}
\item[(a)] For $w \in D_{r}$ there exists a unique $z \in D_{1/2}$ such that
$E(z) = w$.
\item[(b)] The map $E :  U_r := \{ z \in D_{1/2} | | E(z) | < r \} \to D_{r} $ is
biholomorphic.
\item[(c)]  We have $D_{r - \epsilon} \subset U_r \subset D_{r+\epsilon}$.
\item[(d)] If  $z \in  D_{r + \epsilon }$, then $|\partial_z E(z) - 1 | \leq \frac{\epsilon}{2} ( 1/2 - (r + \epsilon))^{-2}$.
\end{itemize}
\end{lemma}

\begin{proof} (a).
  Existence:
For $z \in \partial D_{\epsilon + r}$ and $w \in D_{r}$,
$$
|E(z) - z | \leq \epsilon < |z| - |w|  \leq  |z - w | \; .
$$
By Rouch\'e's theorem, for any $w \in D_{r}$ there exists
a unique $z \in D_{\epsilon + r}$ such that $E(z) = w$. Uniqueness:
If $w \in D_r$, $z \in D_{1/2}$, and $E(z) = w$, then
\begin{align} \label{eq:energyestimate1}
|z| \leq |E(z)| + \epsilon < r + \epsilon \; .
\end{align}
\noindent
(b). This follows from (a) by the inverse function theorem of complex analysis.
\\
\noindent
(c).  The first inclusion follows from $|E(z) | \leq |z| + \epsilon$. The second from
\eqref{eq:energyestimate1}.
\\
\noindent
To obtain the estimate in  (d), we use Cauchy's integral formula
$$
| \partial_z ( E(z) - z)  | \leq \liminf_{\eta \downarrow 0} \left| \frac{1}{2 \pi i }\int_{\partial D_{1/2-\eta}}
\frac{E(w)  - w }{ (w - z )^2} dw \right| \leq  \frac{\pi }{2\pi} \frac{\epsilon}{ ( 1/2 - (r + \epsilon))^{2}} .
$$

\end{proof}

Let $\chi_1$ and  $\chib_1$ be two functions  in $C^\infty([0,\infty);[0,1])$  with $\chi_1^2 + \chib_1^2 = 1$,
$\chi_1 = 1$ on $[0,3/4)$, and $\supp \chi_1 \subset [0,1]$. We set
$$
\chi_\rho(\cdot)  = \chi_1(\cdot /\rho) \quad , \quad \chib_\rho(\cdot) = \chib_1(\cdot /\rho) \; ,
$$
and use the abbreviation
$\chi_\rho = \chi_\rho(H_f)$ and $\chib_\rho  = \chib_\rho(H_f)$.
It should be clear from the context whether $\chi_\rho$ or $\chib_\rho$ denotes a function or an operator.
For an explicit choice of $\chi_1$ and $\chib_1$ see \cite{BCFS03}. The following lemma will be
needed  to be able to define the Feshbach map which will be used later.


\begin{lemma} \label{renorm:thm1} Let $0 < \rho \leq  1/2$.  Then for all $w \in \mathcal{B}^\#(\rho/8,\rho/2,\rho/8)$ we have
\begin{eqnarray}
&& \| \left( H_{0,0}(w) \upharpoonright \ran \chib_\rho \right)^{-1}  \| \leq  \frac{16}{3 \rho} \label{eq:w00estimate}  \\
\label{renorm:thm1:eq31}
&&\| H_{0,0}(w)^{-1} \chib_\rho W[w]  \|   < \frac{2}{3}  \quad ,  \quad   \|   W[w] H_{0,0}(w)^{-1} \chib_\rho \| < \frac{2}{3}   .
\end{eqnarray}
In particular $(H(w),H_{0,0}(w))$ is a Feshbach pair for $\chi_\rho$.
\end{lemma}
\begin{proof}
To prove the lemma we verify the assumptions of Lemma \ref{fesh:thm2}. Clearly $\chi_\rho$ commutes with $H_{0,0}(w)$.
  For $r \in [\frac{3}{4}\rho,1]$, we estimate
\begin{align}
| w_{0,0}(r) | & \geq r -      |r - ( w_{0,0}(r) - w_{0,0}(0) ) | - |w_{0,0}(0)|  \nonumber \\
& \geq r - r \rho/8 - \rho/2  \geq  \frac{3}{4} \rho (1 - \rho/8) - \rho/2 \geq   \frac{3 \rho}{16}  \label{eq:basicwest} \; .
\end{align}
This implies  that $H_{0,0}(w)$ is invertible on the range of $\chib_\rho$ and that \eqref{eq:w00estimate} holds.
By this and $\|W[w]\| \leq \rho/8$, which follows from \eqref{eq:opestimatgeq1},  inequalities  \eqref{renorm:thm1:eq31} follow.
The Feshbach property now follows from Lemma \ref{fesh:thm2}, since  $\| \chib_1 \|_\infty, \|\chi_1 \|_\infty \leq 1$.
\end{proof}

\begin{remark} \label{rem:gartenhaag}  { Note that $w \in \mathcal{B}(\alpha,\beta,\gamma)$ and $z \in U[w]$ imply $w(z) \in \mathcal{B}^\#(\alpha,\rho/2,\gamma)$.}
\end{remark}

In the definition of the renormalization transformation there is a  scaling transformation $S_\rho$ which scales the energy value $\rho$ to the value 1.
It is defined as follows.
For operators  $A \in \mathcal{B}(\FF)$ set
$$
S_\rho(A) =  \rho^{-1}   \Gamma_\rho A \Gamma_\rho^* ,
$$
where  $\Gamma_\rho$ is the unitary dilation on $\FF$ which is  uniquely determined  by $\Gamma_\rho \Omega = \Omega$ and
$ \Gamma_\rho a^\#(k) \Gamma_\rho^* = \rho^{-3/2} a^\#(\rho^{-1} k)$,
for all $k \in \R^3$.
It is easy to check that $\Gamma_\rho H_f \Gamma_\rho^* = \rho H_f$ and hence $\Gamma_\rho \chi_\rho \Gamma_\rho^* = \chi_1$.
We are now ready to precisely define the renormalization transformation, which  in
view of Lemmas \ref{renorm:thm3} and \ref{renorm:thm1}   and
Remark \ref{rem:gartenhaag} is well defined.

\begin{definition} \label{def:defofrenorm} Let $0 < \rho \leq 1/2$. For $w\in \mathcal{B}^\#(\rho/8,\rho/2,\rho/8)$, we define the operator
$$
\mathcal{R}_\rho^\# ( H(w) ) := S_\rho \left( F_{\chi_{\rho}}\left(H(w),H_{0,0}(w) \right) \right) \upharpoonright \HH_{\rm red}  ,
$$
and for $w\in \mathcal{B}(\rho/8,\rho/8,\rho/8)$ we define the renormalization transformation
$$
\left( \mathcal{R}_\rho H(w) \right)(z) :=  \mathcal{R}_\rho^\#(H( w(E_\rho[w]^{-1}(z))) ,
$$
where  $z \in D_{1/2}$.
\end{definition}

In view of the Feshbach property, Theorem \ref{thm:fesh} (ii), and since $\ran \chi_1 \subset \mathcal{H}_{\rm red}$, it will turn out to be sufficient
to study the restriction of the Feshbach map
to $\mathcal{H}_{\rm red}$,

\section{Renormalization Transformation: Kernels}
\label{sec:ren:ker}

We have defined the  renormalization transformation on the level of operators. In this section
we will describe  the induced transformation on the integral kernels. This transformation
is derived the same way as in \cite{BCFS03}. However, we use modified estimates to show
that the renormalized kernel is again an element of $\mathcal{W}_\xi^\#$.

Throughout this section we assume
$w\in \mathcal{B}^\#(\rho/8,\rho/2,\rho/8)$ and $0 < \rho \leq 1/2$. We will show that under suitable  conditions
there exists an integral kernel
$\mathcal{R}_\rho^\#(w) \in \mathcal{W}_\xi^\#$, given in \eqref{nenorm:eq1} below, such that
$$
\mathcal{R}_\rho^\# (H(w)) = H(\mathcal{R}_\rho^\#(w) ) .
$$
Note that the uniqueness of the integral kernel will follow from Theorem \ref{thm:injective}.
Next we show its formal existence. First we expand the Feshbach operator into a Neumann series which is justified by Lemma  \ref{renorm:thm1} and
rearrange  the factorization to arrive at the following identity which holds  on $\HH_{\rm red}$,
\begin{eqnarray}
 F_{\chi_{\rho}} ( H(w),H_{0,0}(w))
&=& T   + \chi_\rho   W \chi_\rho - \chi_\rho W \chib_\rho ( T  + \chib_\rho W \chib_\rho )^{-1} \chib_\rho W \chi_\rho \nonumber \\
&=&  T      +  \sum_{L=1}^\infty  (-1)^{L-1} \chi_\rho W \left[ \frac{\chib_\rho^2}{ T }  W  \right]^{L-1}    \chi_\rho \; ,  \label{eq:feshexpansion}
\end{eqnarray}
where here  we used the abbreviations  $W = W[w]$ and   $T = T[w]$.
Using the commutation relation of the creation and annihilation operators and the pull-through
formula we bring this expression into normal order. To this end  we introduce
\begin{eqnarray*}
\lefteqn{ W_{p,q}^{m,n}[w]( r , K^{(m,n)} ) } \\
&:=& P_{\rm red} \int_{B_1^{p+q}} \frac{d X^{(p,q)}}{|X^{(p,q)}|^{1/2}} a^*(x^{(p)}) w_{p+m,q+n}(H_f + r , x^{(p)},
k^{(m)}, \widetilde{x}^{(q)}, \widetilde{k}^{(n)} )
 a(\widetilde{x}^{(q)}) P_{\rm red}
\end{eqnarray*}
which defines an operator for a.e. $K^{(m,n)} \in B_1^{m+n}$. In the case $m = n = 0$ we set   $W^{0,0}_{m,n}[w](r) := W_{m,n}[w](r)$.
For later use we state an  inequality in the following lemma. The inequality  is obtained the same way as  \eqref{eq:operatornormestimate1}.
\begin{lemma} \label{lem:wwestimate} Let $w \in \WW_\xi^\#$. Then 
\begin{align*} \label{eq:wwestimate}
\| W_{p,q}^{m,n}[w](r,K^{(m,n)}) \|_{\rm op} &\leq \frac{ \| w_{p+m,q+n} \|_\infty}{\sqrt{p! q!}} \\
  \| \partial_r W_{p,q}^{m,n}[w](r,K^{(m,n)}) \|_{\rm op}  &\leq \frac{ \| \partial_r w_{p+m,q+n} \|_\infty}{\sqrt{p! q!}} ,
\end{align*}
where the partial derivative $\partial_r W_{p,q}^{m,n}[w]$ is taken in the weak operator topology.
\end{lemma}
The next theorem,  \cite{BCFS03},  is a variant of Wick's Theorem and will be used to write the $L$-th summand in  \eqref{eq:feshexpansion} in terms of integral kernels.
Its proof can be found in  Appendix B.
\begin{theorem} \label{thm:wicktheorem} Let  $w \in \mathcal{W}^\#_\xi$  and let $F_0,F_1,...,F_L$ be bounded Borel measureable functions on $[0,\infty)$.  Then
$$
F_0(H_f) W[w] F_1(H_f) W[w] \cdots W[w] F_{L-1}(H_f) W[w] F_L(H_f) = H( \widetilde{w}^{({\rm sym})}   ) ,
$$
where
\begin{eqnarray}
\lefteqn{ \widetilde{w}_{M,N}(r,K^{(M,N)}) }  \nonumber \\
& = &
\sum_{\substack{ m_1 + \cdots m_L = M   \\  n_1+...n_L=N     }} \sum_{\substack{ p_1, q_1,...,p_L,q_L: \\  m_l+p_l+n_l+q_l \geq 1 }} \prod_{l=1}^L
\left\{ \binom{ m_l + p_l }{ p_l}   \binom{ n_l + q_l }{ q_l }  \right\}
\nonumber  \\
& & \times F_0(r + \tilde{r}_0)
\langle \Omega   , \prod_{l=1}^{L-1} \left\{ W_{p_l,q_l}^{m_l,n_l}[w]( r + r_l , K_l^{(m_l,n_l)}) F_{l}(H_f + r + \widetilde{r}_{l}) \right\} \nonumber \\
& &
W_{p_L,q_L}^{m_L,n_L}[w]( r + r_L , K_L^{(m_L,n_L)}) \Omega \rangle F_L(r + \widetilde{r}_L) , \label{eq:complicated}
\end{eqnarray}
with
\begin{align*}
& K^{(M,N)} := (K_1^{(m_1,n_1)}, ... , K_L^{(m_L,n_L)}) , \quad K_l^{(m_l,n_l)} := (k_l^{(m_l)},\widetilde{k}_l^{(n_l)}) ,  \\
& r_l := \Sigma[\widetilde{k}_1^{(n_1)}] + \cdots + \Sigma[\widetilde{k}_{l-1}^{(n_{l-1})}] + \Sigma[{k}_{l+1}^{(m_{l+1})}] + \cdots + \Sigma[{k}_L^{(m_L)}] ,  \\
& \widetilde{r}_l := \Sigma[\widetilde{k}_1^{(n_1)}] + \cdots + \Sigma[\widetilde{k}_{l}^{(n_{l})}] + \Sigma[{k}_{l+1}^{(m_{l+1})}] + \cdots + \Sigma[{k}_L^{(m_L)}]
\end{align*}
\end{theorem}
We use the standard convention that $\prod_{j=1}^n a_j = a_1 a_2 \cdots a_n$.

\begin{remark} \label{rem:rem1}  { If $F_0,F_L \in C^\infty[0,\infty)$ have support contained in $[0,1]$, then the summands occurring in \eqref{eq:complicated} satisfy property
(i) of the definition of $\WW_{m,n}^\#$ (see Definition \ref{def:wgartenhaag}). Because of property (i) of the definition of $\WW_{m,n}^\#$,
only the values of $F_1,...,F_{L-1}$ on $[0,1]$ matter in \eqref{eq:complicated}.
The supremum norm can be estimated using
\begin{equation} \label{eq:babyestimate1}
| \langle \Omega , A_1 A_2 \cdots  A_n \Omega \rangle | \leq \| A_1 \|_{\rm op} \| A_2 \|_{\rm op}  \cdots \| A_n \|_{\rm op}
\end{equation}
and Lemma \eqref{lem:wwestimate}. Now suppose $F_1,...,F_{L-1} \in C^1[0,1]$. Then by the defining property of $\mathcal{W}_{m,n}^\#$ we can
calculate the derivative  with respect to $r$ of each summand using the Leibniz rule, where  the interchange of integration and differentiation is justified by
(iii) of  the definition of $\WW_{m,n}^\#$.  Using again \eqref{eq:babyestimate1} and Lemma \ref{lem:wwestimate} it can be shown that each summand of \eqref{eq:complicated}
is in $C^{1}[0,1]$ a.e.. }
\end{remark}

It can be shown that the involved sums converge absolutely in the $\| \cdot \|_\xi^\#$ norm. But for the moment we are only interested in the
combinatorics.
Using  Theorem \ref{thm:wicktheorem}  to write expression \eqref{eq:feshexpansion} in terms of an operator involving integral kernels, it turns
out to useful to  introduce the multi-indices $\umm =(m_1,...,m_L) \in \N_0^L$,  for $L \in \N$. We set
$|\umm|= m_1 + \cdots + m_L$, and  $\underline{0} := (0,0,...,0) \in \N_0^L$.
As a final step we have to scale the operator or equivalently the integral kernels. The integral
kernels scale as follows.
For $(m,n) \in \N_0^2$
$$
s_\rho(w)_{m,n}(r,K^{(m,n)}) := \rho^{m+n - 1 } w_{m,n}(\rho r , \rho K^{(m,n)})  ,
$$
since then
$$
 P_{\rm red} S_\rho ( H(w) ) P_{\rm red}  \upharpoonright {\HH_{\rm red}}               = H (s_\rho ( w )) .
 $$
Following the outlined procedure above, we arrive at the renormalized integral kernels
\begin{equation} \label{nenorm:eq1}
\mathcal{R}_\rho^\#(w) := \widehat{w}^{({\rm sym})} \; ,
\end{equation}
where the  kernels  $\widehat{w}$ are given as follows.
For $M + N \geq 1$,
\begin{align} \label{eq:renormalizedkernels01}
\widehat{w}_{M,N}( r , K^{(M,N)} ) & :=
\sum_{L=1}^\infty (-1)^{L-1} \rho^{M+N - 1}  \sum_{ \substack{ (\umm,\upp,\unn,\uqq) \in \N_0^{4L}:  \\    |\umm|=M, |\unn|=N , \\  m_l+p_l+n_l+q_l \geq 1 }}    \\
& \ \     \prod_{l=1}^L \left\{ \binom{ m_l + p_l }{ p_l} \binom{ n_l + q_l }{ q_l}    \right\} v_{\umm,\upp,\unn,\uqq}[w](r, K^{(M,N)}) , \nonumber
\end{align}
and
\begin{align} \label{eq:renormalizedkernels00}
\widehat{w}_{0,0}( r ) &:= \rho^{-1} w_{0,0}(\rho r)  + \rho^{-1} \sum_{L=2}^\infty (-1)^{L-1}
\sum_{ \substack{ (\upp,\uqq) \in \N_0^{2L}: \\ p_l + q_l \geq 1 } }
v_{\underline{0},\upp, \underline{0},\uqq}[w](r) \; .
\end{align}
Moreover, we have introduced  the expressions
\begin{eqnarray}
\lefteqn{ v_{\umm,\upp,\unn,\uqq}[w]( r, K^{(|\umm|,|\unn|)}) :=  } \label{eq:defofv} \\
&&
\left\langle \Omega  ,  F_0[w] (H_f + \rho (r + \widetilde{r}_0) ) \prod_{l=1}^L \left\{  {W}_{p_l,q_l}^{m_l,n_l}[w](\rho(r+r_l), \rho K_l^{(m_l,n_l)} ) F_l[w]( H_f + \rho( r + \widetilde{r}_l ) ) \right\} \Omega
\right\rangle   , \nonumber
\end{eqnarray}
where $F_0[w](r) := \chi_\rho(r  )$ and  $F_L[w](r) := \chi_\rho(r  )$, and for $l = 1, ... , L - 1$
\begin{eqnarray*}
F_l[w](r)  := F[w](r) := \frac{ \overline{ \chi}_\rho^2(r )}{ w_{0,0}(r) } \; .
\end{eqnarray*}
Above we have used notation introduced in Theorem  \ref{thm:wicktheorem}.
From the previous discussion in this section,  Theorem \ref{renorm:thm2a}, below, follows
 apart from the property that the renormalized kernel is indeed an element of
the Banach space $\mathcal{W}^\#_\xi$ and satisfies a uniform bound.
\begin{theorem} \label{renorm:thm2a} Let $0<\rho \leq 1/2$ and $0< \xi \leq 1/2$ and  assume $w \in \mathcal{B}^\#(\rho/8,\rho/2,\rho/8)$.
Then $\mathcal{R}_\rho^\#(w) \in \mathcal{W}^\#_\xi$ and
\begin{equation*} \label{renorm:thm2:eq1}
 \mathcal{R}_\rho^\# (H(w))   = H(\mathcal{R}_\rho^\#(w)) \; .
\end{equation*}
Moreover, $\sup_{w \in \mathcal{B}^\#(\rho/8,\rho/2,\rho/8)} \| \mathcal{R}_\rho^\#(w) \|^\#_\xi < \infty$.
\end{theorem}
The remaining part of this section  concerns the proof of Theorem \ref{renorm:thm2a}. To prove it we need an
estimate on the kernels  \eqref{eq:defofv}. Note that in
view of Remark  \ref{rem:rem1} the kernels \eqref{eq:defofv} as well as their derivatives are well defined
and can be shown to be bounded.
\begin{lemma}
\label{codim:thm3first} Let  $0 < \rho \leq  1/2$ and   $w \in \mathcal{B}^\#(\rho/8,\rho/2,\rho/8)$.
Then for $(\umm,\upp,\unn,\uqq) \in \N_0^{4L}$ we have
\begin{equation} \label{eq:vestimate1}
 ||  v_{\umm,\upp,\unn,\uqq}[w] ||^\#   \leq   C_L  \left( \frac{1}{t}\right)^{L-1}  \prod_{l=1}^{L}
    \frac{ \|  w_{m_l + p_l, n_l + q_l} \|^\# }  { \sqrt{p_l ! q_l !}}  ,
\end{equation}
where
\begin{align}
t := 3 \rho /16  \quad , \quad
C_L  :=   1 + 2 L  \| \partial_r \chi_1 \|_\infty  + (L-1)8   .  \label{eq:ctheta}
\end{align}
\end{lemma}
\begin{proof}
To arrive at  \eqref{eq:vestimate1} we start with the following estimates.
For $l=0,L$ we have
\begin{eqnarray}
\| F_l[w]( H_f + \rho( r + \widetilde{r}_l ) ) \|_{\rm op} &\leq& 1 \\
\left\| \partial_ r \left( F_l[w]( H_f + \rho( r + \widetilde{r}_l ) ) \right) \right\|_{\rm op} &\leq&  \| \partial_r \chi_1 \|_\infty \ .
\end{eqnarray}
Using  \eqref{eq:basicwest} we find
\begin{equation} \label{est:sss1}
 \| F[w] \|_{\infty} \leq \left[ \inf_{r \in [\rho \frac{3}{4},1] } | w_{0,0}(r) | \right]^{-1} \leq  \frac{1}{t} .
\end{equation}
We also need an estimate on the derivative of $F[w]$,
\begin{equation}\label{eq:derofFFF}
\partial_r F[w](r) =
 \rho^{-1} \frac{ 2 \overline{ \chi}_\rho (r  )  [\partial_r \overline{ \chi}_1] (r/\rho  )}{
  w_{0,0}( r )}
   - \frac{ \overline{ \chi}_\rho^2(r )  [\partial_r w_{0,0}]( r )}{
 (w_{0,0}( r ))^2 } .
\end{equation}
Using  \eqref{est:sss1} and $ \overline{ \chi}_1 \partial  \overline{ \chi}_1  =  - { \chi}_1 \partial { \chi}_1 $
we  estimate \eqref{eq:derofFFF} and  obtain
\begin{equation} \label{est:sss2}
\| \partial_r F[w]  \|_\infty  \leq \rho^{-1} \frac{2 \| \partial_r \chi_1 \|_{\infty}}{t} + \frac{  3/2}{t^2} ,
\end{equation}
noting that $\|  \partial_r w_{0,0} \|_\infty  \leq 1 +  \rho/8 \leq 3/2$.
Next we use \eqref{eq:babyestimate1} and   Lemma \ref{lem:wwestimate}   to obtain the following  estimate
$$
| {v}_{\umm,\upp,\unn,\uqq}[ \underline{w}  ] (r , K^{(|\umm|,|\unn|)}) | \leq \| F[w] \|_\infty^{L-1} \prod_{l=1}^L \frac{ \|  w_{m_l + p_l, n_l + q_l} \|_\infty }  { \sqrt{p_l ! q_l !}} .
$$
Using Leibniz' rule a similar estimate yields,
\begin{eqnarray*} \label{eq:contmultilinear121}
\lefteqn{ | \partial_r {v}_{\umm,\upp,\unn,\uqq}[ \underline{w}  ] (r , K^{(|\umm|,|\unn|)}) |  } \\
 &\leq& \left(  2  \| \partial_r \chi_1 \|_\infty \| F[w] \|_\infty^{L-1}
+ (L-1)  \| F[w] \|_\infty^{L-2}  \| \rho \partial_r F[w] \|_\infty    \right)
  \prod_{l=1}^L \frac{ \|  w_{m_l + p_l, n_l + q_l} \|_\infty }  { \sqrt{p_l ! q_l !}}
  \\
  &&
   + \| F[w] \|_\infty^{L-1} \sum_{l'=1}^L  \frac{ \|  \rho \partial_r w_{m_{l'} + p_{l'}, n_{l'} + q_{l'}} \|_\infty }  { \sqrt{p_{l'} ! q_{l'} !}}
    \prod_{\substack{l=1\\l\neq l'}}^L
  \frac{ \|  w_{m_l + p_l, n_l + q_l} \|_\infty  }  { \sqrt{p_l ! q_l !}} .
\end{eqnarray*}
Collecting estimates yields that claim.
\end{proof}

\vspace{0.5cm}
\noindent
{\it Proof of Theorem \ref{renorm:thm2a}.}
Assume $w \in \mathcal{B}^\#(\rho/8,\rho/2,\rho/8)$. In
view of the discussion in this section it remains to show that  $\mathcal{R}_\rho^\#(w) \in \mathcal{W}^\#_\xi$.
To this end note that
by the definition of $\widehat{w}_{M,N}$, \eqref{eq:renormalizedkernels01},  we find for $M + N \geq 1$,
\begin{eqnarray*}
\lefteqn{ \| \widehat{w}_{M,N}  \|^\# } \\
&\leq & \sum_{L=1}^\infty \sum_{ \substack{ (\umm,\upp,\unn,\uqq) \in \N_0^{4L}:  \\  |\umm| = M  , | \unn| = N  , m_l+p_l+n_l+q_l \geq 1 }}
 \rho^{|\umm|+|\unn|-1}
    \prod_{l=1}^L       \left\{ \binom{ m_l + p_l }{ p_l } \binom{ n_l + q_l }{ q_l }                  \right\}  \| v_{\umm,\upp,\unn,\uqq}[w] \|^\#   . \nonumber \\
 \end{eqnarray*}
  Inserting this below and using the estimate of   Lemma   \ref{codim:thm3first}, we find using  $\frac{1}{{\sqrt{p_l ! q_l !}}} \leq 1$
\begin{eqnarray}
\lefteqn{ \| (\widehat{w}_{M,N})_{M+N  \geq 1} \|_\xi^\# } \nonumber \\
&& = \sum_{M+N \geq 1} \xi^{-(M+N)} \| \widehat{w}_{M,N} \|^\#  \nonumber \\
&&\leq   \sum_{L=1}^\infty      C_L t^{1-L}  \rho^{-1}   \sum_{ \substack{ (\umm,\upp,\unn,\uqq)  \in \N_0^{4L}: \\
   |\umm| + | \unn| \geq 1 , m_l+p_l+n_l+q_l \geq 1} }
 \left(  2 \rho\right)^{|\umm|+|\unn|}   (2 \xi)^{-(|\umm|+|\unn|)}  \nonumber \\
 && \quad \times
      \prod_{l=1}^L    \left\{  \binom{ m_l + p_l }{p_l } \binom{ n_l + q_l}{  q_l }            \frac{\| w_{m_l + p_l, n_l + q_l}  \|^\# }{\sqrt{p_l ! q_l !}} \right\}  \nonumber \\
   &&\leq
    \sum_{L=1}^\infty   C_L t^{-L}
     \sum_{  \substack{  (\umm,\upp,\unn,\uqq) \in \N_0^{4L}:  \\
   m_l+p_l+n_l+q_l \geq 1 }} \nonumber \\
   && \times
      \prod_{l=1}^L \left\{  \binom{ m_l + p_l }{p_l } \binom{ n_l + q_l }{  q_l }      \xi^{p_l + q_l} (1/2)^{m_l+n_l}
                                        \xi^{-(m_l + p_l +n_l + q_l)}   \| w_{m_l + p_l, n_l + q_l} \|^\#
  \right\} \nonumber \\
&& \leq
    \sum_{L=1}^\infty   C_L t^{-L}
    \left[   \sum_{  m+p+n+q \geq 1 }   \binom{ m + p }{ p } \binom{ n + q }{  q }      \xi^{p + q} (1/2)^{m+n}
                                        \xi^{-(m + p +n + q)}   \| w_{m + p, n + q} \|^\# \right]^L \nonumber \\
 &&\leq
    \sum_{L=1}^\infty   C_L t^{-L}
    \left[   \sum_{  l+k \geq 1 }
   \xi^{-(l + k )}   \| w_{l , k } \|^\# \right]^L \nonumber \\
  && \leq
        \sum_{L=1}^\infty  C_L t^{-L}     \left( \| w_{\geq 1} \|^\#_\xi \right)^L ,
   \nonumber
\end{eqnarray}
where in the second last inequality we used the binomial formula
\begin{equation} \label{eq:binomial}
\sum_{   m + p  = l  }   \binom{ m + p }{ p } \xi^p (1/2)^m = ( \xi + 1/2)^l \leq 1
\end{equation}
The term in the last line is bounded since $  \| w_{\geq 1} \|^\#_\xi / t < 1$.
A similar but simpler estimate yields
\begin{eqnarray}
 \|  \widehat{w}_{0,0}  \|^\#
&& \leq \rho^{-1} \| w_{0,0}(\rho \ \cdot ) \|^\#  +  \rho^{-1} \sum_{L=2}^{\infty}
\sum_{\substack{ (\upp,\uqq) \in \N_0^{2L}: \\ p_l + q_l \geq  1 } } \|  v_{\underline{0},\upp,
\uzz,\uqq}[w]  \|^\#  \nonumber \\
&& \leq    \rho^{-1} \| w_{0,0}(\rho \ \cdot ) \|^\#        +  \rho^{-1} \sum_{L=2}^{\infty} C_L t^{1-L}
\sum_{\substack{ (\upp,\uqq) \in \N_0^{2L}: \\ p_l + q_l \geq  1 }     }
\prod_{l=1}^L  \frac{ \| w_{p_l,q_l} \|^\#  }{\sqrt{p_l! q_l!}}  \nonumber  \\
&& \leq   \rho^{-1} \| w_{0,0}(\rho \ \cdot ) \|^\#       +  \sum_{L=2}^{\infty}  C_L (\xi/t)^L
 \left[ \sum_{p+q\geq 1}  \xi^{-(p+q)}\| w_{p,q} \|^\#  \right]^L   \nonumber \\
 && \leq  \rho^{-1} \| w_{0,0}(\rho \ \cdot ) \|^\#   +     \sum_{L=2}^{\infty}  C_L  (\xi/t)^L  \left(\| w_{\geq 1} \|^\#_{\xi} \right)^{L} , \nonumber
\end{eqnarray}
where the last line is bounded since $    \| w_{\geq 1} \|^\#_{\xi} / t <  1$.

\qed

\section{Analyticity and Continuity}
\label{sec:ren:ana}

In this section we show that the renormalization transformation acting on the integral kernels preserves analyticity and c-continuity.
We note that an alternate proof to show that the renormalization transformation preserves analyticity would be to show  that $\mathcal{R}_\rho^\#$
 is a Frechet differentiable map on
the space of integral kernels. Whenever we can treat a statement $A$ concerning analyticity and a statement $C$ concerning continuity in a similar way,
we will write ``$A$ ($C$)'' which stands for ``$A$ respectively $C$''.

\begin{theorem} \label{thm:analytbasic}  Let $0<\rho \leq 1/2$ and $0< \xi \leq 1/2$ ($0<\xi\leq 1/4$). Let $S$ be an open  subset of $\C^\nu$ with $\nu \in \N$ (a topological space). Suppose the map $w(\cdot) : S \to \WW_\xi^\#$ is analytic (c-continuous)
and $w(S) \subset \mathcal{B}^\#(\rho/8,\rho/2,\rho/8)$. Then
$$
\mathcal{R}^\#_\rho(w(\cdot)) : S \to \WW_\xi^\#
$$
is also analytic (c-continuous).
\end{theorem}
Lemma  \ref{renorm:thm3},  Remark \ref{rem:gartenhaag}, Theorem \ref{renorm:thm2a},
and Theorem \ref{thm:analytbasic}    imply  the following theorem.
\begin{theorem} \label{thm:maingenerala}  Let $0<\rho \leq 1/2$ and $0< \xi \leq 1/2$.
For
$w\in \mathcal{B}(\rho/8,\rho/8,\rho/8)$ the  integral kernel
$\mathcal{R}_\rho w : D_{1/2} \to \mathcal{W}_\xi^\#$, defined by $(\mathcal{R}_\rho w)(z) := \mathcal{R}_\rho^\#(w(E_\rho[w]^{-1}(z))$ for  $z \in D_{1/2}$ is in $\mathcal{W}_\xi$ and
$$
(\mathcal{R}_\rho H(w))(z) = H((\mathcal{R}_\rho w)(z)) .
$$
If $w$ is symmetric then also  $\mathcal{R}_\rho w$  is symmetric.
\end{theorem}
The statement about the symmetry follows from \eqref{eq:symmetricsa} and the definition of the renormalization transformation, see Definition   \ref{def:defofrenorm}. The symmetry
property  could also  be ve\-ri\-fied using the explicit expressions \eqref{eq:renormalizedkernels01} and  \eqref{eq:renormalizedkernels00}.
We write F-differentiable for Frechet differentiable. Furthermore, Theorem \ref{thm:analytbasic} has the following theorem as
consequence.
\begin{theorem} \label{thm:analytbasic2}  Let $0<\rho \leq 1/2$ and $0< \xi \leq 1/2$ ($0<\xi\leq 1/4$). Let $S$ be an open  subset of $\C^\nu$ (a topological space). Suppose
\begin{eqnarray*}
w(\cdot, \cdot) : &&S \times D_{1/2} \to \WW_\xi^\# \\
&&(s,z) \mapsto w(s,z)
\end{eqnarray*}
is an analytic (a c-continuous) function such that
 $w(s)(\cdot) := w(s, \cdot)$ is in $\mathcal{B}(\rho/8,\rho/8,\rho/8)$. Then
$$
(s,z) \mapsto (\mathcal{R}_\rho(w(s)))(z)
$$
is also a $\WW_\xi^\#$-valued analytic (c-continuous) function.
\end{theorem}

\begin{remark} {
Note that  by Hartogs' Theorem  joint analyticity is equivalent to individual analyticity. }
\end{remark}

\begin{proof}  First observe that
 $(s,z) \mapsto E_\rho[w(s)](z)$ is analytic (continuous).
It follows that the mapping $(s,z) \mapsto E_\rho[w(s)]^{-1}(z)$ on $S \times D_{1/2}$ is analytic (continuous), which can be seen from
Lemma  \ref{renorm:thm3} and the identity
\begin{equation} \label{eq:implicitE}
E_\rho[w(s)](E_\rho[w(s)]^{-1}(z)) - z = 0 .
\end{equation}
It follows that  the map  $(s,z) \in S \times D_{1/2} \to w(s,E_\rho[w(s)]^{-1}(z)) $  is analytic (c-continuous) and  by  Remark \ref{rem:gartenhaag} its
range is contained in $\mathcal{B}^\#(\rho/8,\rho/2,\rho/8)$. It follows now from Theorem   \ref{thm:analytbasic}
 that $(s,z) \mapsto \mathcal{R}_\rho^\#(w(s,E_\rho[w(s)]^{-1}(z)))$ is analytic (c-continuous).
\end{proof}

The remaining part of this section is devoted to the proof of  Theorem \ref{thm:analytbasic}.
First we show the statement regarding analyticity, then we show the statement regarding c-continuity.

To show the statement about analyticity we first show in Lemma  \ref{lem:frechetdiff1aa}, below, that the map
$ v_{\umm,\upp,\unn,\uqq}[w(\cdot )] : S \to \WW_\xi^\#$ is analytic. It then follows from   \eqref{eq:renormalizedkernels01} and \eqref{eq:renormalizedkernels00}
that the renormalized kernel  $\widehat{w(s)} = \mathcal{R}^\#_\rho(w(s))$ is given as a series of analytic mappings.
Analyticity of  the renormalized kernel  will follow, provided that the series converges uniformly on $S$. Since we are  not
able to show this on the whole set $S$ directly,  we will  show, below, uniform convergence on open subsets of $S$ which constitute a covering of $S$.
  This  is in fact sufficient to  conclude   the analyticity of  $s \mapsto \mathcal{R}^\#_\rho(w(s))$.

\begin{lemma} \label{lem:frechetdiff1aa} Let the assumptions of  Theorem \ref{thm:analytbasic} hold. Then $ v_{\umm,\upp,\unn,\uqq}[w(\cdot )] : S \to \WW_\xi^\#$ is analytic.
\end{lemma}

Lemma \ref{lem:frechetdiff1aa} follows since by part (a) of the following Lemma  and  Estimate \eqref{eq:basicwest}  the function  $v_{\umm,\upp,\unn,\uqq}[w(\cdot )]$ is   a
composition of an analytic map
with a F-differentiable map.

\begin{lemma} \label{lem:frechetdiff1a}  \label{lem:frechetdiff1}   \label{eq:frechetf}
Let $0 < \rho \leq 1/2$.   Then
the following statements hold for  $\epsilon  > 0$.
\begin{itemize}
\item[(a)] On   $\mathcal{O}^{(\epsilon)}  := \{ w \in \WW^\#_{\xi} | {\rm inf}_{r  \in [ \rho \frac{3}{4}, 1]} | w_{0,0}(r) | > \epsilon \}$ the following
map is  F-differentiable
\begin{align*}
 v_{\umm,\upp,\unn,\uqq}[\cdot ] : \mathcal{O}^{(\epsilon)}      \longrightarrow  \WW_\xi^{\#}     \quad , \quad
 w  \longmapsto v_{\umm,\upp,\unn,\uqq}[w] .
\end{align*}
\item[(b)]  On   $\mathcal{O}_{0,0}^{(\epsilon)}  := \{ t \in \WW^\#_{0,0} |  {\rm inf}_{ r  \in [ \rho \frac{3}{4}, 1]} | t (r) | > \epsilon \}$  the following
map is  F-differentiable
\begin{align}
F[\cdot ] : \mathcal{O}_{0,0}^{(\epsilon)}    \longrightarrow  \WW^\#_{0,0}  \label{eq:frechetfa}  \quad , \quad
t  \longmapsto \frac{\chib_{\rho}^2}{t}   .
\end{align}
 \end{itemize}
\end{lemma}

\begin{proof}
First we show part (b).
We will use that for all $f,g \in \WW_{0,0}$ we have $\| f g \|^\# \leq \| f \|^\# \| g\|^\#$ and that   for all $\xi \in \WW_{0,0}$ with
$\| \xi \|^\# < \epsilon / 2$ we have
\begin{align*}
\left\| F[t + \xi ] - F[t] + \frac{\chib_{\rho}^2 \xi}{t^2} \right\|^\# &= \left\| \frac{\chib_{\rho}^2 \xi^2}{t^2 ( t + \xi )} \right\|^\# \\
&\leq \left\|  \frac{\chib_{\rho}}{t^2} \right\|^\#  \left\|  \frac{\chib_{\rho}}{t + \xi } \right\|^\# \| \xi^2 \|^\# \leq C \left( \| \xi \|^\# \right)^2 ,
\end{align*}
where in the last inequality we used the estimate
\begin{eqnarray}
 \left\|  \frac{\chib_{\rho}}{t + \xi } \right\|^\#  &&\leq  \left\|  \frac{\chib_{\rho}}{t + \xi } \right\|_\infty  +  \left\|  \frac{\partial_r \chib_{\rho}}{t + \xi } \right\|_\infty +   \left\|  \frac{\chib_{\rho} \partial_r ( t + \xi)}{(t + \xi)^2 } \right\|_\infty \\
 && \leq \frac{ 1 + \| \partial_r \chib_{\rho} \|_\infty}{ \epsilon - \| \xi \|^\#}   + \frac{\|  t \|^\# + \| \xi \|^\#}{ ( \epsilon - \| \xi \|^\#)^2 } \leq  C .
\end{eqnarray}
This implies that $F[ \cdot ]$ is differentiable with derivative  $-\chib_{\rho}^2 / t^2$.

(b)  The differentiability of  $v_{\umm,\upp,\unn,\uqq}[\cdot ]$ follows from the fact it can be written as a  composition of the F-differentiable  mapping
$\widetilde{v}_{\umm,\upp,\unn,\uqq}[\cdot ]$, defined   below,  and  $F[\cdot ]$.
  \label{lem:genvestimate1a}  \label{lem:genvestimate1}
For  $\underline{w} = (w_1,...,w_L) \in (\WW_\xi^\# )^L$ and $\underline{G} = (G_0,...,G_L) \in (\WW_{0,0}^\#)^{L+1}$ define
the multilinear  expression
\begin{eqnarray}
\lefteqn{ \widetilde{v}_{\umm,\upp,\unn,\uqq}[\underline{w},\underline{G}]( r, K^{(|\umm|,|\unn|)}) :=  } \label{eq:defofvtildea} \\
&
\left\langle \Omega  ,  G_0(H_f + \rho ( r + \widetilde{r}_0 ) ) \prod_{l=1}^L
\left\{  {W}_{p_l,q_l}^{m_l,n_l}[w_l](\rho(r+r_l), \rho K_l^{(m_l,n_l)} ) G_l( H_f + \rho (r  + \widetilde{r}_l ) ) \right\} \Omega
\right\rangle    \nonumber .
\end{eqnarray}
It satisfies  the inequality
\begin{eqnarray}
\lefteqn{ \| \widetilde{v}_{\umm,\upp,\unn,\uqq}[ \underline{w},\underline{G} ] \|^\# } \label{eq:contmultilinear211a} \\&& \leq
 \left( \prod_{l=0}^L
 \| G_l \|_\infty  + \sum_{l'=0}^L  \| \rho \partial_r G_{l'} \|_\infty \prod_{l=0, l\neq l'}^L \| G_l \|_\infty \right)     \prod_{l=1}^{L}
     \frac{ \|  {(w_l)}_{m_l + p_l, n_l + q_l}  \|^\# }  { \sqrt{p_l ! q_l !}}   .  \nonumber
\end{eqnarray}
To obtain \eqref{eq:contmultilinear211a}  we  use \eqref{eq:babyestimate1} and  Lemma \ref{lem:wwestimate}, and
 calculate the derivative with respect to $r$ using Leibniz' rule. 
From  \eqref{eq:contmultilinear211a} it follows that $\widetilde{v}_{\umm,\upp,\unn,\uqq}[\cdot ]$ is continuous, and hence by multilinearity   $\widetilde{v}_{\umm,\upp,\unn,\uqq}[\cdot ]$ is in fact differentiable.
\end{proof}

Next  we show that the defining sequence of $\widehat{w(s)}$, see \eqref{eq:renormalizedkernels01} and \eqref{eq:renormalizedkernels00}, converges uniformly on
open sets which constitute a covering of $S$.
To this end choose $s_0 \in S$
and define the set

$$
U_0 = \{ w \in \mathcal{B}^\#(\rho/8,\rho/2,\rho/8) | \| w - w(s_0) \|_\xi^\# < \epsilon \}
$$
where we set
$$
\epsilon := \frac{\rho/7 - \| w(s_0)_{\geq 1} \|_\xi^\#}{16 e^4} .
$$
The explicit choice of $\epsilon$ is needed for the estimate    \eqref{eq:boundonGa},         
below.
Note that by continuity there exists, $S_0$, an open subset of $S$ containing $s_0$, such that $w(S_0) \subset U_0$.
For $w \in U_0$, we have
\begin{equation*} 
\| w_{m,n} \|^\# \leq E_{m,n} := \| w(s_0)_{m,n} \|^\# + \xi^{m+n} \epsilon .
\end{equation*}
By Lemma \ref{codim:thm3first},
\begin{equation} \label{eq:estonvons}
{\rm sup}_{s \in S_0} \| v_{\umm,\upp,\unn,\uqq}[w( s)] \|^\# \leq C_L t^{-L + 1} \prod_{l=1}^L \frac{E_{m_l + p_l , n_l q_l}}{\sqrt{p_l ! q_l !}} ,
\end{equation}
where we used the notation introduced in that lemma. To establish the uniform convergence on $S_0$  of the series defining  $\widehat{w(s)}$
 it suffices, in view of \eqref{eq:renormalizedkernels01} and  \eqref{eq:renormalizedkernels00},  to show that the following expression is bounded
\begin{align}
\label{eq:rdiffstup1a}
& \sum_{M+N \geq 0} \sum_{L=1}^\infty
\sum_{\substack{ (\umm,\upp,\unn,\uqq) \in \N_0^{4L} \\ | \umm | = M , | \unn | = N \\ m_l + p _l + n_l + q_l \geq 1}} \xi^{-|\umm|-|\unn|} \rho^{|\umm|+|\unn|}
\prod_{l=1}^L \left\{  \binom{m_l + p_l }{ p_l} \binom{ n_l + q_l }{ q_l }      \right\}
 {\rm sup}_{s \in S_0} \| v_{\umm,\upp,\unn,\uqq}[w( s)] \|^\#  \\
&\leq  \sum_{L=1}^\infty C_L t^{1-L}  G^{L}  \nonumber ,
\end{align}
where we used Eq.  \eqref{eq:estonvons} and the definition
$$
G := \sum_{ m+p + n+q \geq 1 }  \binom{m+p }{ p} \binom{n+q }{ q} \xi^{p+q} (1/2)^{m+n} \xi^{-m-p-n-q}  \frac{E_{m+p, n+q}}{\sqrt{p! q!}} .
$$
Below we will show that
\begin{equation} \label{eq:boundonGa}
G \leq \| w(s_0)_{\geq 1} \|_\xi^\# + \epsilon 16 e^4 \leq \rho / 7 .
\end{equation}
Inequalities \eqref{eq:boundonGa}
imply the  convergence of \eqref{eq:rdiffstup1a} , since $t^{-1} G \leq t^{-1} \rho/7  < 1$.
The second inequality in \eqref{eq:boundonGa} follows from the definition of $\epsilon$.
To show the first inequality of \eqref{eq:boundonGa}, we will use the following estimate

\begin{equation}
\sum_{m+p \geq 0} \binom{ m + p }{ p} \xi^p (1/2)^m \frac{1}{\sqrt{p!}}
\leq \sum_{m+p \geq 0}  \binom{ m + p }{ p} (1/4)^p (1/2)^m e^{8 \xi^2}
= 4e^{8 \xi^2}
 \leq 4 e^2 , \label{eq:combestimatedera}
\end{equation}
where in the first inequality we used the trivial estimate  $(16\xi^2)^{p}/{p!} \leq e^{16\xi^2}$.
To show the first inequality in  \eqref{eq:boundonGa}, we   insert the definition of $E_{m,n}$  into the definition of $G$.
This yields two terms, which we have to estimate. The first term, involving $w_{m,n}(s_0)$, is estimated using the binomial formula and the second term, involving $\epsilon$,
is estimated using   \eqref{eq:combestimatedera}.

\vspace{1cm}

It remains to  show the statement regarding c-continuity.
By Lemma  \ref{lem:newcont}, shown next,  the map  $s \mapsto v_{\umm,\upp,\unn,\uqq}[w(s)]$ is  continuous with respect to $\| \cdot \|_2$.
By  \eqref{eq:renormalizedkernels01} and \eqref{eq:renormalizedkernels00} this will imply that the function
$s \mapsto \widehat{w(s)}_{M,N}$ is given as a series involving   expressions which are continuous with respect to  $\| \cdot \|_2$.
The c-continuity of  $s \mapsto \widehat{w(s)}$
will follow provided we show that this series converges uniformly in  $s \in S $ with respect to the $\| \cdot \|_2$ norm. In fact
we will first show uniform convergence with respect to $\| \cdot \|^\#$. In view of  \eqref{eq:infinity2ineq} this will imply the uniform convergence with respect to
the $\| \cdot \|_2$ norm.

\begin{lemma} \label{lem:newcont} Let $w : S \mapsto \WW_\xi^\#$ be   c-continuous, and let $w(S) \subset \mathcal{B}^\#(\rho/8,\rho/2,\rho/8)$. Then for all
$s_0 \in S$
\begin{equation} \label{eq:newcont:eq}
\lim_{s \in S,  s \to s_0} \left\|  v_{\umm,\upp,\unn,\uqq}[w(s_0)] -  v_{\umm,\upp,\unn,\uqq}[w(s)] \right\|_2 \to 0   .
\end{equation}
\end{lemma}
\begin{proof}
The kernel  $v_{\umm,\upp,\unn,\uqq}$ is a multilinear expression of integral kernels.
To show continuity we use the following identity
\begin{eqnarray} \label{eq:telescoping}
\lefteqn{ A_1(s) \cdots A_n(s) - A_1(s_0) \cdots A_n(s_0) } \nonumber \\
&&= \sum_{i=1}^n A_1(s) \cdots A_{i-1}(s) ( A_i(s) - A_i(s_0)) A_{i+1}(s_0) \cdots A_n(s_0) .
\end{eqnarray}
Now \eqref{eq:newcont:eq} follows using estimate \eqref{eq:babyestimate1}, the following inequality, which is shown similarly as the estimate in Lemma \ref{lem:operatornormestimates2},
\begin{equation}
\int \frac{dK^{(m,n)}}{|K^{(m,n)}|^2} \sup_{r \in [0,1] } \| W^{m,n}_{p,q}[w]( r , K^{m,n} ) \|^2_{\rm op} \leq \| w_{m+p,n+q} \|_2^2  ,
\end{equation}
inequality \eqref{eq:infinity2ineq}, and  the limits
\begin{align*}
&  \left\| w(s_0)_{m,n}   -    w(s)_{m,n} \right\|_2  \stackrel{s \to s_0}{\longrightarrow} 0       ,           \\
&\sup_{r} \left| \frac{\chib_1^2(r)}{w(s_0)_{0,0}(r)}    -    \frac{\chib_1^2(r)}{w(s)_{0,0}(r)}  \right|  \stackrel{s \to s_0}{\longrightarrow} 0              ,
\end{align*}
which follow by assumption.
\end{proof}

It remains to  show that the defining series of $\widehat{w(s)}_{M,N}$ converges uniformly in  $s \in S $. In view
of \eqref{eq:renormalizedkernels01}
and  \eqref{eq:renormalizedkernels00} this will be established
if we can show that \eqref{eq:uniformpcontestmn} and \eqref{eq:uniformpcontest00} are finite.
To this end, first observe that  
it follows that for all $m+n \geq 1$
\begin{equation} \label{eq:initialestimate3455}
\sup_{s \in S} \| w(s)_{m,n} \|^\# \leq \xi^{m+n} \sup_{s \in S} \|(w(s)_{m,n})_{m+n \geq 1} \|_{\xi}^\# \leq  \xi^{m+n} \frac{\rho}{8} .
\end{equation}
Inserting  \eqref{eq:vestimate1} and the above estimate  into the following expression,
for  $M + N \geq 1$,  we find
\begin{eqnarray}
\lefteqn{
\sum_{L=1}^\infty \sum_{ \substack{ (\umm,\unn) \in \N_0^{2L}:  \\    |\umm|=M, |\unn|=N }}  \sum_{  \substack{ (\upp,\uqq) \in \N_0^{2L} \\ m_l+p_l+n_l+q_l \geq 1       } }
 \prod_{l=1}^L \left\{ \binom{ m_l + p_l }{ p_l} \binom{ n_l + q_l }{  q_l}   \right\}  {\rm sup}_{s \in S} \| v_{\umm,\upp,\unn,\uqq}[w( s)] \|^\# } \nonumber   \\
&& \leq
\sum_{L=1}^\infty (-1)^{L-1} \rho^{M+N - 1}  C_L t^{1-L} \sum_{ \substack{ (\umm,\unn) \in \N_0^{2L}:  \\    |\umm|=M, |\unn|=N }}   \nonumber   \\
&& \ \     \prod_{l=1}^L \left\{  {\sum_{p_l,q_l }}'  \binom{ m_l + p_l}{  p_l} \binom{ n_l + q_l }{ q_l}
 \xi^{m_l+p_l+n_l+q_l} \frac{\rho}{8}  \right\}   \nonumber  \\
&& \ \ =: F , \label{eq:uniformpcontestmn}
\end{eqnarray}
where $\sum'_{ p_l,q_l }$ denotes the sum over all $(p_l,q_l) \in \N_0^{2}$ such that $m_l+p_l+n_l+q_l \geq 1 $.
If $m_l \neq 0$ or $n_l \neq 0$  we estimate using $\binom{ n }{ k } \leq 2^n$
\begin{align}
 {\sum_{ p_l,q_l}}'  \binom{ m_l + p_l }{ p_l} \binom{ n_l + q_l }{ q_l}  \xi^{m_l+p_l+n_l+q_l}
 \leq (2\xi)^{m_l+n_l} \sum_{p_l \geq 0} (2\xi)^{p_l} \sum_{q_l \geq 0} (2 \xi)^{q_l} \leq   4  ,
\end{align}
where we used   that $0 < \xi \leq 1/4$. If both $m_l=0$ and $n_l=0$ then either $p_l \geq 1$ or $q_l \geq 1$ and we estimate
\begin{align}
 {\sum_{  p_l,q_l}}'  \binom{ m_l + p_l}{ p_l} \binom{ n_l + q_l }{  q_l}  \xi^{m_l+p_l+n_l+q_l}
 \leq \sum_{p_l + q_l \geq 1}   \xi^{p_l} \xi^{q_l}  \leq \frac{7}{9}  \leq 1 ,
\end{align}
where in the second last inequality we used again $0< \xi \leq 1/4$.
Inserting these estimates into $F$ and using that  there are at most $(M+1)(N+1)$ factors for which   $m_l \neq 0$ or $n_l \neq 0$  we find
$$
F \leq \rho^{M+N -1} \sum_{L=1}^\infty C_L t^{1-L} \sum_{ \substack{ (\umm,\unn) \in \N_0^{2L}:  \\    |\umm|=M, |\unn|=N }}   4^{(M+1)( N+1)} \left( \frac{\rho}{8} \right)^{L}
$$
Now the estimates
$$
\sum_{ \substack{ (\umm,\unn) \in \N_0^{2L}:  \\    |\umm|=M, |\unn|=N }}1 \leq (L+1)^{M+N}
$$
and $t^{-1} \rho/8 < 1$ imply $F < \infty$.  Now we consider \eqref{eq:renormalizedkernels00}.
 Using  \eqref{eq:vestimate1}  and \eqref{eq:initialestimate3455} we find
\begin{align}
\sum_{L=2}^\infty
\sum_{ \substack{ (\upp,\uqq) \in \N_0^{2L}: \\ p_l + q_l \geq 1 } }
\left\{ {\rm sup}_{s \in S} \| v_{\uzz,\upp,\uzz,\uqq}[w( s)] \|^\#  \right\}
\leq   \sum_{L=2}^\infty  C_L t^{L-1}
\left\{ \sum_{ p + q \geq 1  }\xi^{p+q}  \frac{\rho}{8}  \right\}^L . \label{eq:uniformpcontest00}
\end{align}
This converges since $\sum_{ p + q \geq 1  }\xi^{p+q} \leq 7/9 \leq 1$ and $t^{-1}\rho/8 < 1$.

\section{Codimension-1 Contractivity}
\label{sec:contract}

In this section we prove that the renormalization transformation is in certain directions a contraction in $\WW_\xi$. We recall Definition~\ref{def:defofrenorm} and the definition introduced in 
 Theorem~\ref{thm:maingenerala}.
In contrast to \cite{BCFS03},  the contraction originates from the fact the we restrict the renormalization transformation
to integral kernels for which the sum of the number of creation and annihilation operators is even, rather than  an infrared condition.
\label{sec:cod}

\begin{theorem} \label{codim:thm1} For any positive  numbers $\rho_0 \leq 1/2$ and $\xi_0 \leq 1/2$ there exist  numbers $\rho, \xi, \epsilon_0$ satisfying
$\rho \in (0, \rho_0]$, $\xi \in (0, \xi_0]$, and $0 < \epsilon_0 \leq \rho/8$
such that the following property holds,
\begin{equation}  \label{codim:thm1:eq}
\mathcal{R}_\rho : \mathcal{B}_0(\epsilon, \delta_1, \delta_2 ) \to
\mathcal{B}_0(   \epsilon + \delta_2/2  ,   \delta_2/2     ,   \delta_2/2   ) \quad , \quad \forall  \ \epsilon, \delta_1, \delta_2 \in [0, \epsilon_0)  .
\end{equation}
\end{theorem}

In fact we will prove the following remark which is a
slightly stronger statement than Theorem \ref{codim:thm1}.

\begin{remark} \label{rem:codim:thm1}  Define the constant $C_\theta := 3 + 2 \| \partial_r \chi_1 \|_\infty$.
The contraction property \eqref{codim:thm1:eq} holds
whenever $0 < \rho \leq  \frac{1}{16 C_\theta}$,  $0 < \xi \leq  \left[ {\rho}/ ({2 C_\theta})\right]^{1/4}$, and $0 < \epsilon_0 \leq  \frac{\rho}{32}$.
\end{remark}

\begin{proof}  We will prove Remark \ref{rem:codim:thm1}. Theorem  \ref{codim:thm1} will then follow.
First observe that if $w \in \mathcal{B}_0(\epsilon,\delta_1,\delta_2)$, then  $(\mathcal{R}_\rho w )_{m,n} = 0$, if $m+n$ is odd.
Since $C_\theta \geq 1$, we can assume  that  $\xi \leq 1/2$ and $\rho \leq 1/2$.
To show the contraction property, we will use the following estimate for $w \in \mathcal{B}^\#(\rho/8,\rho/2,\rho/8)$
\begin{equation} \label{eq:vestimate111}
 ||  v_{\umm,\upp,\unn,\uqq}[w] ||^\#   \leq  C_\theta \left(\frac{16}{\rho} \right)^{L-1}  \prod_{l=1}^{L}
    \frac{ \|  w_{m_l + p_l, n_l + q_l} \|^\# }  { \sqrt{p_l ! q_l !}}  ,
\end{equation}
which follows directly from Lemma  \ref{codim:thm3first}.
We shall use the notation $z = E_\rho[w]^{-1}(\zeta)$ where $\zeta \in D_{1/2}$, see Lemma \ref{renorm:thm3}.

\vspace{0.5cm}

\noindent
\underline{Step 1:}  We have
$$
\| (\mathcal{R}_\rho w)_{ \geq 2} \|_\xi \leq    \frac{1}{2}   \| w_{\geq 2} \|_\xi .
$$

\vspace{0.5cm}

By the definition of $\widehat{w}_{M,N}$, \eqref{eq:renormalizedkernels01},   we find for $M + N \geq 2$,
\begin{eqnarray*}
\lefteqn{ \|    \widehat{w(z)}_{M,N}             \|^\# } \\
&\leq & \sum_{L=1}^\infty \sum_{ \substack{ (\umm,\upp,\unn,\uqq) \in \N_0^{4L}:  \\  |\umm| = M  , | \unn| = N  , m_l+p_l+n_l+q_l \geq 1 }}
 \rho^{|\umm|+|\unn|-1}
      \prod_{l=1}^L     \binom{ m_l + p_l }{ p_l } \binom{ n_l + q_l }{  q_l }    \| v_{\umm,\upp,\unn,\uqq}[w(z)] \|^\#  \nonumber \\
 \end{eqnarray*}
  Inserting this below and using the Estimate \eqref{eq:vestimate111} we find with $\tau := 16/\rho$,
\begin{eqnarray}
\lefteqn{ \| ((\mathcal{R}_\rho w) (\zeta))_{M+N  \geq 2} \|_\xi^\# } \nonumber \\
&& = \sum_{M+N \geq 2} \xi^{-(M+N)} \| \widehat{w(z)}_{M,N} \|^\#  \nonumber \\
&&\leq   \sum_{L=1}^\infty          \sum_{ \substack{ (\umm,\upp,\unn,\uqq)  \in \N_0^{4L}: \\
   |\umm| + | \unn| \geq 2 , m_l+p_l+n_l+q_l \geq 1} }
\rho^{-1}  \left(  2 \rho\right)^{|\umm|+|\unn|}   (2 \xi)^{-(|\umm|+|\unn|)}  C_\theta \tau^{L-1} \nonumber \\
 && \quad \times
      \prod_{l=1}^L  \left\{   \binom{ m_l + p_l }{p_l } \binom{ n_l + q_l}{  q_l }            \frac{\| w(z)_{m_l + p_l, n_l + q_l}  \|^\# }{\sqrt{p_l ! q_l !}}  \right\}  \nonumber \\
   &&\leq
   \frac{C_\theta}{16} [2 \rho]^2   \sum_{L=1}^\infty  \tau^{L}
     \sum_{  \substack{  (\umm,\upp,\unn,\uqq) \in \N_0^{4L}:  \\
   m_l+p_l+n_l+q_l \geq 1 }} \nonumber \\
   && \times
     \prod_{l=1}^L \left\{   \binom{ m_l + p_l}{ p_l } \binom{ n_l + q_l }{ q_l }      \xi^{p_l + q_l} 2^{-(m_l+n_l)}
                                        \xi^{-(m_l + p_l +n_l + q_l)}   \| w(z)_{m_l + p_l, n_l + q_l} \|^\#
  \right\} \nonumber \\
&& \leq \frac{C_\theta}{4} \rho^2
    \sum_{L=1}^\infty   \tau^L
    \left[   \sum_{  m+p+n+q \geq 1 }    \binom{ m + p }{ p } \binom{ n + q}{  q }      \xi^{p + q} 2^{-(m+n)}
                                        \xi^{-(m + p +n + q)}   \| w(z)_{m + p, n + q} \|^\# \right]^L \nonumber \\
 &&\leq \frac{C_\theta}{4} \rho^2
    \sum_{L=1}^\infty  \tau^L
    \left[   \sum_{  l+k \geq 1 }
   \xi^{-(l + k )}   \| w(z)_{l , k } \|^\# \right]^L \nonumber \\
  && \leq
    \frac{C_\theta}{4} \rho^2     \sum_{L=1}^\infty  \tau^{L}     \left( \| w(z)_{\geq 2} \|^\#_\xi \right)^L
   \nonumber \\
   && \leq
   8 C_\theta \rho \| w(z)_{\geq 2} \|_\xi^\# , \nonumber,
\end{eqnarray}
where in the third last inequality we used the binomial formula, \eqref{eq:binomial},  and we used
$\tau \| w_{\geq 2} \|_\xi \leq 1/2$ in the last inequality.

\vspace{0.5cm}

\noindent
\underline{Step 2:}
\begin{align*}
\sup_{\zeta \in D_{1/2}}\| \partial_r (\mathcal{R}_\rho w)(\zeta)_{0,0}  - 1 \|_\infty \leq \sup_{z \in D_{1/2} } \| \partial_r w(z)_{0,0} - 1 \|_\infty + \frac{1}{2}  \| w_{\geq 1} \|_\xi .
\end{align*}

Using the definition of $\widehat{w}_{0,0}$, \eqref{eq:renormalizedkernels00}, and
we find,
\begin{eqnarray}
 \|  \partial_r (\mathcal{R}_\rho w)(\zeta)_{0,0}   - 1 \|_\infty
&& \leq \|  \partial_r w(z)_{0,0}  - 1 \|_\infty + \rho^{-1} \sum_{L=2}^{\infty}
\sum_{\substack{ (\upp,\uqq) \in \N_0^{2L}: \\ p_l + q_l \geq  1 } } \|  v_{\underline{0},\upp,
\uzz,\uqq}[w(z)]  \|^\#  \nonumber \\
&& \leq \| \partial_r w(z)_{0,0} - 1 \|_\infty +  \rho^{-1} \sum_{L=2}^{\infty} C_\theta \tau^{L-1}
\sum_{\substack{ (\upp,\uqq) \in \N_0^{2L}: \\ p_l + q_l \geq  2 }     }
\prod_{l=1}^L  \frac{ \| w(z)_{p_l,q_l} \|^\#  }{\sqrt{p_l! q_l!}}  \nonumber  \\
&& \leq \| \partial_r w(z)_{0,0} - 1 \|_\infty  + \frac{C_\theta}{16} \sum_{L=2}^{\infty}  \left[\tau \xi^2 \right]^{L}
 \left[ \sum_{p+q \geq 2}  \xi^{-(p+q)}\| w(z)_{p,q} \|^\#  \right]^L   \nonumber \\
 && \leq \| \partial_r w(z)_{0,0} - 1 \|_\infty +  \frac{C_\theta}{16}  \xi^4 \sum_{L=2}^{\infty}    \left[ \tau  \| w(z)_{\geq 2} \|^\#_{\xi} \right]^{L} \nonumber \\
 && \leq \| \partial_r w(z)_{0,0} - 1 \|_\infty  +   \frac{C_\theta}{16}  \xi^4  \tau
\| w(z)_{\geq 2} \|^\#_{\xi}   \label{eq:kernelestimate2}
\end{eqnarray}
where in  the last estimate we used $    \tau \| w_{\geq 1} \|_{\xi}  \leq 1/2$.

\vspace{0.5cm}

\noindent
\underline{Step 3:}
\begin{align*}
\sup_{\zeta \in D_{1/2}} |( \mathcal{R}_\rho w)(\zeta)_{0,0}(0) + \zeta | &  \leq  \frac{1}{2}  \| w_{\geq 1} \|_\xi  .
\end{align*}

 \vspace{0.5cm}

 We estimate
\begin{align*}
|  ( \mathcal{R}_\rho w)(\zeta)_{0,0}(0) + \zeta | & \leq \rho^{-1} \sum_{L=2}^\infty
\sum_{\substack{ (\upp,\uqq) \in \N_0^{2L}: \\ p_l + q_l \geq  1 } } \|  v_{\underline{0},\upp,
\uzz,\uqq}[w(z)]  \|^\#    \nonumber     \\
&  \leq   \frac{C_\theta}{16}  \xi^4 \tau  \| w(z)_{\geq 1} \|^\#_{\xi}   \; , \label{eq:kernelestimate2}
\end{align*}
where in the last step we used  an estimate from Step 2.

\end{proof}

\section{Construction of Eigenvectors and Eigenvalues}
\label{sec:con}

In this section we show how the contraction property of Theorem \ref{codim:thm1} and the Feshbach property allows us to
recover the eigenvectors and eigenvalues of the initial operator. The main theorems of this section are Theorems \ref{thm:bcfsmain}
and \ref{thm:continsigma}. Theorem  \ref{thm:bcfsmain}, apart from the last sentence, is from \cite{BCFS03}. We follow the
proof given there, and isolate a few estimates which will be needed to prove the analyticity and continuity results of
Theorem   \ref{thm:main2}. The last sentence in Theorem  \ref{thm:bcfsmain}, has been shown
in \cite{GH09} but in a different way, due to the different representation of the spectral parameter in \cite{GH09}.

Throughout this section we assume the following hypothesis.

\vspace{0.5cm}

\noindent
(R) \quad  Let $\rho, \xi, \epsilon_0$ be positive numbers such that the contraction property \eqref{codim:thm1:eq}
holds and $\rho \leq  1/4$, $\xi \leq 1/4$ and $\epsilon_0 \leq \rho/8$.

\vspace{0.5cm}

We note that many statements only require $0 < \rho \leq 1/2$ and $0 < \xi \leq 1/2$. But we will need $0 < \rho \leq 1/4$ in Lemma  \ref{lem:limitofe}, below,
and we will need $0 < \xi \leq 1/4$ for the statement about c-continuity in Theorem  \ref{thm:analytbasic2}.
Hypothesis (R)
allows us to iterate the renormalization transformation as follows,
\begin{equation} \nonumber
\mathcal{B}_0(\sfrac{1}{2}\epsilon_0 , \sfrac{1}{2} \epsilon_0 , \sfrac{1}{2} \epsilon_0 )  \stackrel{\mathcal{R}_\rho}{\longrightarrow}
\mathcal{B}_0( [ \sfrac{1}{2}+ \sfrac{1}{4} ] \epsilon_0 , \sfrac{1}{4} \epsilon_0 , \sfrac{1}{4} \epsilon_0 )    \stackrel{\mathcal{R}_\rho}{\longrightarrow}    \cdots  \stackrel{\mathcal{R}_\rho}{\longrightarrow}
\mathcal{B}_0( \Sigma_{l=1}^n \sfrac{1}{2^l} \epsilon_0 , \sfrac{1}{2^n} \epsilon_0 , \sfrac{1}{2^n} \epsilon_0 )  .
\end{equation}
For  $w \in \mathcal{B}_0(\epsilon_0/2,\epsilon_0/2,\epsilon_0/2)$ and $n \in \N_0$, we define
\beqn \label{construct:eq4}
w^{(n)} := \mathcal{R}_\rho^n(w) \in \mathcal{B}_0(\epsilon_0, 2^{-n-1}\epsilon_0, 2^{-n-1}\epsilon_0) \; .
\eeqn
We introduce the definitions
\begin{align*}
E_{n,\rho}[w](z) &:= E_\rho[w^{(n)}](z) = -  \rho^{-1}  \langle \Omega, H( w^{(n)}(z)) \Omega \rangle  \\
U_n[w] &:= U[w^{(n)}] = \{ z \in D_{1/2} | | E_n(z) | < \rho/2 \} .
\end{align*}
By Lemma  \ref{renorm:thm3} the map
$$
J_n[w] := E_{n,\rho}[w] : U_n[w] \to D_{1/2}  \quad , \quad z \mapsto E_{n,\rho}[w](z)
$$ is an analytic bijection and
  $J_n[w]^{-1} : D_{1/2} \to U_n[w] \subset D_{1/2}$. For
 $0 \leq n \leq m$, we define
\ben
e_{(n,m)}[w] := J_n[w]^{-1} \circ \cdots \circ J_m[w]^{-1}(0) \; .
\een
Lemma \ref{lem:limitofe} stated below  immediately implies that
the limit $e_{(n,\infty)}[w] := \lim_{m \to \infty} e_{(n,m)}[w]$ exists for all $n \in \N_0$.
\begin{lemma}  \label{lem:limitofe} Assume (R) and let $w \in \mathcal{B}_0(\epsilon_0/2, \epsilon_0/2, \epsilon_0/2)$. Then
\begin{equation} \label{eq:defofe}
| e_{(n,m)}[w] - e_{(n,m + k)}[w] | \leq \left(\frac{4\rho}{3} \right)^{m-n} .
\end{equation}
\end{lemma}
\begin{proof} For notational simplicity we drop the $w$ dependence in the proof.
By Lemma \ref{renorm:thm3},
\begin{equation} \label{eq:diffJ}
|\rho \partial_z J_n(z) - 1 | \leq 1/4  \quad , \quad \forall z \in U_n \; .
\end{equation}
This implies by the inverse function theorem that
\beqn \label{construct:eq5}
|  \partial_\zeta J^{-1}_n(\zeta) | \leq \frac{4 \rho}{3}   \quad  , \quad \forall \zeta \in D_{1/2} \; .
\eeqn
An iterated application of  \eqref{construct:eq5}, the convexity of $D_{1/2}$, and  the chain rule yields
\begin{align*}
|e_{(n,m)} - e_{(n,m+k)} | &= |  J^{-1}_n \circ \cdots \circ J^{-1}_m(0) -   J^{-1}_n \circ \cdots \circ J^{-1}_m  ( J^{-1}_{m+1}  \circ \cdots \circ J^{-1}_k(0)) | \\
\leq \left(\frac{4 \rho}{3}\right)^{m-n} \frac{1}{2} \; .
\end{align*}
\end{proof}
Next we introduce some notation. Let
\begin{align*}
H_n[w] &:= H(w^{(n)}\left(e_{(n,\infty)}[w])\right) \\
T_n[w] &:=   w_{0,0}^{(n)}(e_{(n,\infty)}[w])(H_f)  \\
Q_n [w] &:= Q_{\chi_\rho}(H_{n}[w], T_{n}[w])
\end{align*}
For $n,m \in \N_0$ with $n\leq m$ we define vectors $\psi_{(n,m)}[w] \in \HH_{\rm red}$ by setting $\psi_{(n,n)}[w] = \Omega$ and
$$
\psi_{(n,m)}[w] = Q_{n}[w] \Gamma_\rho^* Q_{n+1}[w] \Gamma_\rho^* \cdots Q_{m-1}[w] \Omega .
$$
Lemma \ref{lem:limitofe} stated below  immediately implies that this sequence converges as $m \to \infty$, i.e.,
the limit
\begin{equation} \label{eq:defofpsi}
\psi_{(n,\infty)}[w] := \lim_{m \to \infty} \psi_{(n,m)}[w]
\end{equation}
exists for all $n \in \N_0$ .

\begin{lemma} \label{lem:limitofpsi}  Assume (R) and  let   $w \in \mathcal{B}_0(\epsilon_0/2,\epsilon_0/2,\epsilon_0/2)$. Then
\beqn \label{construct:eq7}
\| \psi_{(n,m+1)}[w] - \psi_{(n,m)}[w] \| \leq  2^{-m} \frac{16 \epsilon_0}{\rho} \exp[2^{-n}32 \epsilon_0 \rho^{-1}] \; .
\eeqn
\end{lemma}
\begin{proof} For notational compactness we drop the $w$ dependence in the proof.
Note that
\beqn \label{construct:eq5b}
\psi_{(n,m+1)} - \psi_{(n,m)} = Q_{n} \Gamma_\rho^* Q_{n+1} \Gamma_\rho^* \cdots Q_{m-1} \Gamma_\rho^*( Q_m - \chi_\rho) \Omega,
\eeqn
where we used $\Gamma_\rho^* \chi_\rho \Omega = \Omega$.
Next we set $W_n := H_n - T_n$  and  estimate $Q_n - \chi_\rho$,
\begin{align}
\| Q_n  - \chi_\rho \| &\leq \|  \chib_\rho ( T_n + \chib_\rho W_n \chib_\rho )^{-1} \chib_\rho  W_n \chi_\rho \|  \nonumber \\
& \leq \left( \rho/8 - \| W_n \| \right)^{-1} \| W_n \|  \nonumber \\
&\leq \frac{ 16\epsilon_0}{\rho} 2^{-n} \; ,  \label{construct:eq6}
\end{align}
where in the second inequality we used that $|T_n(r)| \geq \rho/8$ if $r \in [\frac{3}{4} \rho,1 ]$, see \eqref{eq:basicwest}
and in the last inequality we used $\|W_n \| \leq 2^{-n-1} \epsilon_0 \leq \rho/16$, see \eqref{eq:opestimatgeq1}.
Eq. \eqref{construct:eq6} implies
$$
\| Q_n \| \leq 1 +  \frac{ 16\epsilon_0}{\rho} 2^{-n} \; .
$$
Using this and \eqref{construct:eq6} to estimate the difference  \eqref{construct:eq5b}, we find
\beqn \label{construct:eq7b}
\| \psi_{(n,m+1)} - \psi_{(n,m)} \| \leq 2^{-m} \frac{16 \epsilon_0}{\rho}
 \prod_{j=n}^{m-1} \left[ 1 + 2^{-j }{16 \epsilon_0}/{\rho} \right]  .
\eeqn
The estimate of the lemma follows from  $\prod_{j=0}^\infty (1 + \lambda_j) \leq \exp[\sum_{j=0}^\infty \lambda_j ]$, which holds for $\lambda_j \geq 0$.
\end{proof}
We are now ready to state the main theorem of this section.
\begin{theorem} \label{thm:bcfsmain} Assume Hypothesis (R).
Let $w \in \mathcal{B}_0(\epsilon_0/2,\epsilon_0/2,\epsilon_0/2)$. Then the complex number $e_{(0,\infty)}[w] \in D_{1/2}$ defined
in \eqref{eq:defofe} is an eigenvalue of $H(w)$, in the sense that
$$
{\rm dim} \ker \{ H(w(e_{(0,\infty)})[w]) \} \geq 1 \; .
$$
Moreover, the vector $\psi_{(0,\infty)}[w]$ defined in \eqref{eq:defofpsi} is a corresponding eigenvector, i.e., is a non-zero element of
$\ker \{ H(w(e_{(0,\infty)}[w]) \}$. We have the bound $\| \psi_{(0,\infty)}[w]\| \leq 4 e^4$.
If $w$ is symmetric and  $-1/2 <  z < e_{(0,\infty)}[w]$, then $H(w(z))$ is bounded invertible.
\end{theorem}

\begin{proof} For compactness we suppress  the $w$ dependence in the proof.
We show that $\psi_{(0,\infty)}$ is a nonzero vector, which is in the kernel of  $H(w(e_{(0,\infty)})$.
By \eqref{construct:eq7} we have the norm estimate
\begin{equation} \label{eq:uniformgsestimate}
\| \psi_{(n,\infty)} - \Omega \| = \| \psi_{(n,\infty)} - \psi_{(n,n)} \| \leq  2^{-n} \frac{32 \epsilon_0}{\rho} \exp[2^{-n} 32 \epsilon_0 \rho^{-1} ] \; ,
\end{equation}
This  implies that $\psi_{(n_0,\infty)} \neq 0$ provided $n_0$ is sufficiently large.
Next we show that $\psi_{(n,\infty)}$ is in the kernel of $H_{n}$. To this end, we shall iterate  the following identity
\beqn  \nonumber
H_{n-1} Q_{n-1} \Gamma_\rho^* = \rho \Gamma_\rho^* \chi_1(H_f) H_n \; ,
\eeqn
which is a consequence of identities involving the Feshbach operator. For $n \leq m$,
\begin{align}
H_n \psi_{(n,m)} &:= ( H_n Q_n \Gamma_\rho^* ) ( Q_{n+1} \Gamma_\rho^* \cdots Q_{m-1} \Omega )  \nonumber \\
&\ =\rho \Gamma_\rho^*\chi_1 (H_{n+1} Q_{n+1} \Gamma_\rho^*) (Q_{n+2} \Gamma_\rho^* \cdots Q_{m-1} \Omega ) \nonumber  \\
&\ \  \vdots \nonumber  \\
&\ =  \rho^{m-n} (\Gamma_\rho^* \chi_1)^{m-n} H_{m} \Omega \; . \label{construct:eq8}
\end{align}
Since $H_n$ is a bounded operator on $\HH_{\rm red}$ the left hand side converges  to $H_n \psi_{(n,\infty)}$ as $m\to \infty$.
Also  the right hand side converges to 0 as $m \to \infty$, since by  \eqref{construct:eq4}
$$
\| H_m \Omega  \|  \leq  {\rm const.}
$$
and there is an overall factor $\rho^{m-n}$. Thus taking the limit as $m$ tends to infinity  in  \eqref{construct:eq8}
yields  for all $n \in \N_0$
$$
H_n \psi_{(n,\infty)} = 0 \; .
$$
In particular we have shown,  $H_{n_0} \psi_{(n_0,\infty)} =0$ and $\psi_{(n_0,\infty)} \neq 0$.
A repeated application of  the Feshbach property implies that $\psi_{(0,\infty)} \neq 0$ and
$H_0 \psi_{(0,\infty)} = 0$. The bound on  $\psi_{(0,\infty)}$ follows from Lemma \ref{lem:limitofpsi} and $\epsilon_0 \leq \rho/8$.

Now we show  the statement about symmetric kernels $w$. Thus let  $w$ be symmetric. Then all  $w^{(n)}$ are also symmetric, by Theorem \ref{thm:maingenerala}.
Let $- \frac{1}{2} <  \zeta \leq  - \frac{3}{16} \rho$. Then we estimate,  with $E_{n}(\zeta) := - \langle \Omega, H(w^{(n)}(\zeta) ) \Omega \rangle$,
\begin{eqnarray*}
\lefteqn{ \langle \varphi , H(w^{(n)}(\zeta)) \varphi \rangle } \\
& = &  \left\langle \varphi , \left( T[w^{(n)}(\zeta)] + E_{n}(\zeta) -\zeta  + \zeta - E_{n}(\zeta)  + W[w^{(n)}(\zeta)] \right) \varphi \right\rangle  \\
& \geq & \langle \varphi , ( T[w^{(n)}(\zeta)] + E_{n}(\zeta) )  \varphi \rangle - \zeta \| \varphi \|^2 - |\zeta - E_{n}(\zeta) | \| \varphi \|^2  - |  \langle \varphi ,  W[w^{(n)}(\zeta)]  ) \varphi \rangle |  \\
& \geq & \left(  \frac{3}{16}  - \frac{1}{16} - \frac{1}{16}  \right) \rho \| \varphi \|^2  = \frac{1}{16}\rho \| \varphi \|^2 ,
\end{eqnarray*}
where the first term in the second line is non-negative  since $\| \partial_r w^{(n)}_{0,0} - 1 \| \leq 1/2$ and $w^{(n)}$ is symmetric, and the last term in the second line is estimated using \eqref{eq:opestimatgeq1}.
Applying Theorem \ref{thm:fesh} iteratively, we find that $H(w(z))$ is bounded invertible if
$ z \in K_n  ( (-\frac{1}{2}, - \frac{3}{16} \rho] ) $,
where we have set $K_n := J^{-1}_0 \circ \cdots \circ J^{-1}_{n-1}$ if $n \geq 1$ and $K_0 := {\rm id}$. It follows that
 $H(w(z))$ is bounded invertible if
 $$
z \in  I_N := \bigcup_{n=0}^N K_n \left( (-\frac{1}{2}, -\frac{3}{16}\rho ] \right) ,
 $$
 for some $N \in \N$.
Below we will show that
\begin{equation} \label{eq:tedious1}
I_N \supset (- \frac{1}{2}, K_N( -\frac{3}{16}\rho ) ] .
\end{equation}
In view of estimate \eqref{construct:eq5} and the definition of $e_{(0,\infty)}$ we have  $\lim_{N \to \infty} K_N(  -\frac{3}{16}\rho ) = e_{(0,\infty)}$.
Thus \eqref{eq:tedious1}  will imply that $H(w(z))$ is bounded invertible for all $z \in (-\frac{1}{2}, e_{(0,\infty)})$.
To show \eqref{eq:tedious1} we first note that $J_n^{-1}: D_{1/2} \to U_n$ is bijective, differentiable, maps real numbers to real numbers, since $w^{(n)}$ is symmetric, and is increasing
because of \eqref{eq:diffJ}. Note that by  \eqref{construct:eq5} $J_n^{-1}$ extends continuously to the boundary of $D_{1/2}$. It follows that $U_n \cap \R = (J_n^{-1}(-1/2), J_n^{-1}(1/2) )$. Since $D_{\frac{3}{8} \rho} \subset U_n \subset D_{1/2}$ (see Lemma \ref{renorm:thm3})
we conclude that
$$
- \frac{1}{2} \leq J_n^{-1}(-1/2) \leq -\frac{3}{16} \rho .
$$
This implies that for any $b \in (-1/2,1/2)$, we have
$$
 \left( (-\frac{1}{2}, - \frac{3}{16} \rho] \cup J_n^{-1} ( (-\frac{1}{2}, b ]) \right) \supset   (-\frac{1}{2}, J_n^{-1}(b)] .
$$
Iterating this relation  one easily shows \eqref{eq:tedious1}.
\end{proof}

The next theorem states in what sense  analytic kernels lead to analytic  eigenvalues and eigenvectors. It also relates
 c-continuous kernels to continuous eigenvalues and eigenvectors.

\begin{theorem}  \label{thm:continsigma} Assume Hypothesis (R).
Let $S$ be an open  subset of $\C^\nu$ (a topological space). Suppose
\begin{eqnarray*}
w(\cdot, \cdot) : &&S \times D_{1/2} \to \WW_\xi^\# \\
&&(s,z) \mapsto w(s,z)
\end{eqnarray*}
is an analytic (a c-continuous) function such that
 $w(s)(\cdot) := w(s, \cdot)$ is in $\mathcal{B}_0(\epsilon_0/2,\epsilon_0/2,\epsilon_0/2)$.
Then $s \mapsto  e_{(0,\infty)}[w(s)]$ and  $s \mapsto \psi_{(0,\infty)}[w(s)]$ are analytic (continuous) functions.
\end{theorem}

\begin{proof} From  Theorems   \ref{codim:thm1} and  \ref{thm:analytbasic2}  it  follows that all integral kernels  $(s,z) \mapsto w^{(n)}(s,z)$  are analytic  (c-continuous).
In particular $(s,z) \mapsto E_\rho[w^{(n)}(s)](z)$ is analytic (continuous).
It follows that the mapping $(s,z) \mapsto E_\rho[w^{(n)}(s)]^{-1}(z)$ on $S \times D_{1/2}$ is analytic (continuous), which can be seen from
Lemma  \ref{renorm:thm3} and  identity \eqref{eq:implicitE}.
Now it follows from the definition that  $e_{(n,m)}[w(s)]$ is an analytic (continuous) function of $s$.
By Lemma \ref{lem:limitofe}  the limits of  $e_{(n,m)}[w(s)]$ as $m$ tends to infinity are uniform in $s$. Thus $s \mapsto e_{(n,\infty)}[w(s)]$ is analytic   (continuous).
It follows that $H_n[w(s)], T_n[w(s)]$ depend analytically on $s$ by the inequality   \eqref{eq:opestimatgeq12}    (continuously on $s$ by  Lemma \ref{lem:pcontop}). This implies
that  $Q_n[w(s)]$ is an analytic (continuous) function of $s$.
By definition now also  $\psi_{(n,m)}[w(s)]$ is an analytic (continuous) function of $s$.
Since by Lemma  \ref{lem:limitofpsi} the limit  of   $\psi_{(n,m)}[w(s)]$ as $m$ tends to infinity is uniform in $s$,
it follows that  $\psi_{(0,\infty)}[w(s)]$ is also an analytic (continuous) function of $s$.
\end{proof}

\section{Initial Feshbach Transformations}
\label{sec:ini}

We perform two initial Feshbach transformations before we
start the renormalization procedure. The main theorem of this section is Theorem \ref{initial:thmmain}.
Throughout  this section we will assume that Hypothesis (H) holds.
First we set
$$
\chi^{(I)} = P_1 \otimes 1  \quad , \quad \chib^{(I)} = P_2 \otimes 1
$$
where
$$
 P_1 =  \left( \begin{array}{cc} 0 & 0 \\ 0 & 1
\end{array} \right) \quad , \quad
P_2 =  \left( \begin{array}{cc} 1 & 0 \\ 0 & 0 \end{array} \right)
\; .
$$
We do not choose to  include a boson momentum cutoff in $\chi^{(I)}$, since the associated Feshbach map
 would otherwise contain terms which are linear in creation and
annihilation operators.
\begin{theorem}  \label{initial:thmA}  $(H_{\lambda,\sigma}  - z , H_f + \tau - z)$ is a Feshbach pair for $\chi^{(I)}$ provided $|z|<2$.
\end{theorem}
\begin{proof}   We will identify the ranges of $P_1$ and $P_2$ with $\FF$.
If is sufficient to verify the assumptions of  Lemma \ref{fesh:thm2}. These are easily verified noting that:
 $ ( H_f + \tau - z)|_{\ran  \chib^{(I)}}  \cong H_f + 2 - z $ is invertible if  $|z|<2$,  $P_2\sigma_x \phi(f_\sigma)P_2=0$, and  by elementary estimates given in
  Appendix A we have that $ (H_f + 2 - z )^{-1} \phi(f_\sigma) $ is bounded if $|z|<2$
\end{proof}
By Theorem \ref{initial:thmA}
the  following definition makes sense for  $|z|<2$
\begin{align*}
H^{(I)}_{\lambda,\sigma}(z)  & := F_{\chi^{(I)}}(H_{\lambda,\sigma}, H_f + \tau - z)  \upharpoonright \ran \chi^{(I)}
\cong  H_f  - z - \lambda^2 \phi(f_\sigma)(  H_f + 2 - z )^{-1}  \phi(f_\sigma)  ,
\end{align*}
where we identified the range of $\chi^{(I)}$ with the Fock space. Next
we use   the pull-through formula, Lemma  \eqref{lem:pullthrough},  to express $H_{\lambda,\sigma}^{(I)}$ in terms of integral kernels.
To this end we introduce the notation
\beqn \label{eq:defhlinemn}
\overline{H}_{m,n}(w_{m,n}) = \int_{\R^{m+n}} \frac{ dK^{(m,n)}}{|K^{(m,n)}|^{1/2}} a^*(k^{(m)}) w_{m,n}(H_f , K^{(m,n)}) a(\tilde{k}^{(n)})  \; .
\eeqn
Definition \eqref{eq:defhlinemn} is understood in the sense of forms.
 We have
\begin{align*}
H_{\lambda,\sigma}^{(I)}  &=  T_{\lambda,\sigma}^{(I)}  + W_{\lambda,\sigma}^{(I)}
\end{align*}
with $T_{\lambda,\sigma}^{(I)}(z) = \overline{H}_{0,0}(w_{0,0}^{(I)}({\lambda,\sigma,z}))$ and   $W_{\lambda,\sigma}^{(I)}(z) :=  \sum_{m+n=2} \overline{H}_{m,n}(w^{(I)}(\lambda,\sigma,z))$
where
\begin{align*}
t^{(I)}(\lambda,\sigma, z ) (r) := w_{0,0}^{(I)}(\lambda,\sigma,z)(r)  &: = r  - z - \lambda^2 \int  \frac{d^3k}{(4 \pi)^2 \omega(k)} \frac{ |f_\sigma(k)|^2 }{r + |k| + 2 - z} \; ,
\end{align*}
and  $w_{m,n}^{(I)} = (\widehat{w}_{m,n}^{(I)})^{\rm sym}$, with
\begin{align*}
\widehat{w}^{(I)}_{2,0}(\lambda,\sigma,z)(r , k_1,k_2)  &:=  - \lambda^2     f_\sigma(k_1)f_\sigma (k_2)       \frac{1}{r + |k_1| +  2 - z}  , \\
\widehat{w}^{(I)}_{0,2}(\lambda,\sigma, z )(r ,\widetilde{k}_1,\widetilde{k}_2 ) &:=  - \lambda^2   \frac{1}{r + |\widetilde{k}_1| +  2 - z}  \overline{f_\sigma}(\widetilde{k}_1)\overline{f}_\sigma(\widetilde{k}_2)  , \\
\widehat{w}^{(I)}_{1,1}(\lambda,\sigma, z )(r, {k}_1,\widetilde{k}_1)  &:=  - \lambda^2  f_\sigma(k_1) \left[  \frac{1}{r  +  2 - z}   +
 \frac{1}{r  + |k_1| + |\widetilde{k}_1|  +   2 - z}     \right]   \overline{f}_\sigma(\widetilde{k}_1)        \; .
\end{align*}
By $w^{(I)}$ we denote the tuple consisting of the 4 components $w^{(I)}_{m,n}$ with $m+n=0,2$.
We will now  apply the Feshbach transformation  one more time.
The next theorem states that  for sufficiently small values of the coupling
constant $(H_{\lambda,\sigma}^{(I)},T_{\lambda,\sigma}^{(I)})$ is a Feshbach pair for $\chi_1$. To formulate the
theorem we introduce the following constant
$$
\mu_0 :=  \left( 8 \max( \| f/ ( 4 \pi \sqrt{ \omega} ) \|, \| f /  (4 \pi \omega) \| ) \right)^{-1}  .
$$

\begin{theorem} \label{initial:thmC} Let $|\lambda| < \mu_0$, $\sigma \geq 0$, and   $|z| \leq 1/2$. Then
the pair of operators $(H_{\lambda,\sigma}^{(I)}(z),T_{\lambda,\sigma}^{(I)}(z))$ is a Feshbach pair for $\chi_1$ and on $D_{1/2}$ we have
\begin{equation}  \label{initial:thmC:eq1}
\|  {T_{\lambda,\sigma}^{(I)}}^{-1} |_{\rm \ran \chib_1}\| \leq \frac{64}{15} \quad , \quad  \| {T_{\lambda,\sigma}^{(I)}}^{-1} \chib_1 W_{\lambda,\sigma}^{(I)} \| < \frac{7}{15} \quad , \quad \|W_{\lambda,\sigma}^{(I)} { T_{\lambda,\sigma}^{(I)}}^{-1} \chib_1 \| < \frac{7}{15} .
\end{equation}
\end{theorem}

\begin{proof} Let  $\delta_0 := \max( \| f/ (4 \pi \sqrt{ \omega}) \|, \| f / (4 \pi \omega) \| ) $.
First we show that   on $D_{1/2}$  we have
\begin{equation} \label{eq:sttpd1}
\| W_{\lambda,\sigma}^{(I)} \| \leq |\lambda|^2  7 \delta_0^2 .
\end{equation}
To this end note  that
$$
W_{\lambda,\sigma}^{(I)}(z) = - \lambda^2 \phi(f_\sigma) ( H_f + 2 - z )^{-1} \phi(f_\sigma) + \lambda^2 \int \frac{d^3k}{(4 \pi)^2 \omega(k)} \frac{|f_\sigma(k)|^2}{H_f + |k| + 2 - z } .
$$
For  $|z| \leq 1/2$, this yields the estimate
\begin{equation} \label{eq:sttpd3}
\| W_{\lambda,\sigma}^{(I)} \| \leq |\lambda|^2 \| \phi(f_\sigma) (H_f + 1 )^{-1/2} \|   \| (H_f + 1 )^{-1/2} \phi(f_\sigma) \| + |\lambda|^2  \| f /  ( 4 \pi  \sqrt{\omega} ) \|^2 .
\end{equation}
Now using elementary estimates collected in Lemma \ref{thm:estimates1} to estimate the first term  in \eqref{eq:sttpd3} one obtains  \eqref{eq:sttpd1}.
Next  observe that $T_{\lambda,\sigma}^{(I)}$ commutes with $\chi_1$ and $\chib_1$. From the following
estimate it follows that $T_{\lambda,\sigma}^{(I)}$ is bounded invertible on the range of $\chib_1$. For $r \geq 3/4$,
\begin{align*}
|t^{(I)}(\lambda,\sigma,z)(r)| \geq 3/4 - 1/2 - |\lambda|^2  \int \frac{d^3k}{(4 \pi )^2 \omega(k)} \frac{|f(k)|^2}{|k| + 1}
\geq 15/64 \; .
\end{align*}
This and \eqref{eq:sttpd1} imply \eqref{initial:thmC:eq1}.  In view of
Lemma \ref{fesh:thm2} it follows that  $(H_{\lambda,\sigma}^{(I)},T_{\lambda,\sigma}^{(I)})$ is a Feshbach pair for $\chi_1$.
\end{proof}

Let   $|\lambda| < \mu_0$ and  $|z| \leq 1/2$. Then  by Theorem \eqref{initial:thmC} the second Feshbach map,
$$
H_{\lambda,\sigma}^{(0)}(z) := F_{\chi_1}(H_{\lambda,\sigma}^{(I)}(z),T_{\lambda,\sigma}^{(I)}(z))  \upharpoonright \mathcal{H}_{\rm red},
$$
is well defined and we are allowed to expand the operator $H_{\lambda,\sigma}^{(0)}(z)$ in a Neumann series. We obtain
on $\HH_{\rm red}$
\begin{align*}
H^{(0)} = T^{(I)} + \chi_1 W^{(I)} \chi_1 - \chi_1 W^{(I)} \chib_1 \sum_{n=0}^\infty \left( - {T^{(I)}}^{-1} \chib_1 W^{(I)} \chib_1 \right)^n
{T^{(I)}}^{-1} \chib_1 W^{(I)} \chi_1 \; ,
\end{align*}
where we dropped the $\lambda,\sigma,z$ dependence and assumed that $z \in D_{1/2}$.
Again we normal order the above expression, using the pull-through formula. To this end we use the identity of Theorem \ref{thm:wicktheorem},
which also holds for the integral kernels considered here since its proof is based on algebraic identities.
This yields a sequence of integral kernels
$\widetilde{w}^{(0)}$, which  are given as follows.
For $M+N \geq 1$,
\begin{eqnarray}
\lefteqn{ \widetilde{w}^{(0)}_{M,N}(\lambda,\sigma,z)(r , K^{(M,N)})} \label{initial:eq7}  \\ &&= \sum_{L=1}^\infty (-1)^{L+1}
\sum_{\substack{ (\umm,\upp,\unn,\uqq)  \in \N_0^{4L}: \\ |\umm|=M,  |\unn|=N, \\  m_l+p_l+q_l+n_l = 2 } }   \prod_{l=1}^L \left\{
\binom{ m_l + p_l }{ p_l}  \binom{ n_l + q_l }{ q_l }
\right\}
V_{(\umm,\upp,\unn,\uqq)}[w^{(I)}(\lambda,\sigma,z)](r,K^{(M,N)})  .\nonumber
\end{eqnarray}
 Furthermore,
\begin{align*}
\widetilde{w}^{(0)}_{0,0}(\lambda,\sigma,z)(r) = t^{(I)}(\lambda,\sigma,z)(r) + \sum_{L=2}^\infty (-1)^{L+1}
\sum_{(\upp,\uqq)\in \N_0^{2L}: p_l+q_l = 2}
V_{(\uzz,\upp,\uzz,\uqq)}[w^{(I)}(\lambda,\sigma,z)](r) \; ,
\end{align*}
where
we have used the definition
\begin{eqnarray} \label{eq:defofV}
V_{\umm,\upp,\unn,\uqq}[w](r, K^{(|\umm|,|\unn|)}) :=
\left\langle \Omega  ,  \bar{F}_0(H_f + r) \prod_{l=1}^L \left\{
\overline{W}_{p_l,q_l}^{m_l,n_l}[w](r + r_l,K^{(m_l,n_l)})
 \bar{F}_l(H_f + r +  \widetilde{r}_l ) \right\} \Omega
\right\rangle
\end{eqnarray}
with $\bar{F}_0[w](r) := \chi_1(r )$, $\bar{F}_L[w](r) := \chi_1(r  )$, and for $l = 1, ... , L - 1$ we set
\begin{eqnarray*}
\bar{F}_l[w](r) := \bar{F}[w](r ) := \frac{ \overline{ \chi}_1(r  )^2}{ w_{0,0}(r )} .
\end{eqnarray*}
Here, we used the definition
\begin{eqnarray} \label{eq:defofWW}
\lefteqn{
\overline{W}_{p,q}^{m,n}[w](r,K^{(m,n)}]} \\
&:=& \int \frac{d X^{(p,q)}}{|X^{(p,q)}|^{1/2}} a^*(x^{(p)}) w_{m+p,n+q}(H_f + r , k^{(m)}, x^{(p)}, \widetilde{k}^{(n)},
\widetilde{x}^{(q)}) a(\widetilde{x}^{(q)}) . \nonumber
\end{eqnarray}
Recall also the notation    introduced  in     Theorem \ref{thm:wicktheorem}.
Since we want to consider symmetric kernels we set $w^{(0)} := \left( \widetilde{w}^{(0)} \right)^{({\rm sym})}$.
We are now ready to state the main theorem of this section.

\begin{theorem} \label{initial:thmmain} Let $0< \xi < 1$ and  $\delta_1, \delta_2, \delta_3 > 0$. Then there exists a positive
$\lambda_0 \leq \mu_0$
such that for all $\lambda \in B_{\lambda_0}$ and $\sigma \geq 0$
we have
\begin{align} \label{eq:initialthmmain}
&w^{(0)}(\lambda,\sigma, \cdot) \in \mathcal{B}_0(\delta_1,\delta_2,\delta_2) \\
&H^{(0)}_{\lambda,\sigma}(z) = H(w^{(0)}(\lambda,\sigma,z))  , \quad \forall z \in D_{1/2} . \label{eq:initialthmmain2}
\end{align}
Moreover the following is true.
\begin{itemize}
\item[(i)] For $\sigma \geq 0$, the map $(\lambda,z) \mapsto w^{(0)}(\lambda,\sigma,z)$ is a $ \WW_{\xi}^\#$-valued analytic function  on $B_{\lambda_0} \times D_{1/2}$.
\item[(ii)] For each $\lambda \in B_{\lambda_0}$, the map $(\sigma,z) \mapsto w^{(0)}(\lambda,\sigma,z) \in \WW_{\xi}^\#$ is a c-continuous function on  $ [0,\infty) \times  D_{1/2} $.
\item[(iii)] For real $\lambda \in B_{\lambda_0}$ and $\sigma \geq 0$, the kernel $w^{(0)}(\lambda,\sigma)$ is symmetric.
\end{itemize}
\end{theorem}

The remaining part of this section is devoted to the proof of Theorem \ref{initial:thmmain}. Let us first outline
the proof. From the previous discussion we know that once  \eqref{eq:initialthmmain} has been established then
 \eqref{eq:initialthmmain2} will follow.
Thus first we will show \eqref{eq:initialthmmain}. The fact that the number of creation and annihilation operators of $w^{(0)}$ is  even, follows
 directly from the definition.  Showing    \eqref{eq:initialthmmain}  also requires an  estimate of the kernel.
To this end we use an estimate on $V_{\umm,\upp,\unn,\uqq}[w^{(I)}]$ which
is  given in Lemma \ref{initial:thmE2222}, below.
Using that estimate for $V_{\umm,\upp,\unn,\uqq}[w^{(I)}]$, we will then obtain estimates \eqref{eq:feb2:1}, \eqref{eq:feb2:2}, and \eqref{eq:feb2:3},
which imply    \eqref{eq:initialthmmain}. Those estimates establish uniform convergence which will then be  used
to show (i) and (ii) using the corresponding statement for $V_{\umm,\upp,\unn,\uqq}[w^{(I)}]$.
(iii) follows from the definition and
\eqref{eq:symmetricsa}.

First we show the following Lemma.
\begin{lemma}   \label{initial:thmE2222} Let $|z| \leq  1/2 $,  $|\lambda| \leq \mu_0 $, and $\sigma \geq 0$. Then
\begin{align} \label{initial:thmmain:eq2}
 \| V_{\umm,\upp,\unn,\uqq}[w^{(I)}(\lambda,\sigma,z)] \|^\# \leq   (L + 1) C_{\bar{F}}^{L+1} C_W(\lambda)^L   .
\end{align}
where
\begin{eqnarray*}
C_{\bar{F}} &:=&  10 \| \partial_r \chib_1 \|_\infty + 31  , \\
C_W(\lambda) &:=&  \sup_{m+n+p+q=2} |\lambda|^{2} 6  ( \| f /{\omega} \|^2 + 2 \| f \| \| f / \omega \| )^{\frac{p+q}{2}} \| f \|_\infty^{m+n} .
 \end{eqnarray*}
\end{lemma}
Lemma \ref{initial:thmE2222} will essentially follow from Lemma  \ref{initial:thmE}, shown below.
To this end we introduce the following norms.
Again we will use the canonical identification of  the space
 $L^\infty(\R^{3 m + 3 n}; C[0,\infty))$ with a subset of $L^\infty([0,\infty) \times \R^{3 m + 3 n})$.
For functions
$w_{m,n} \in L^\infty(\R^{3 m + 3 n}; C[0,\infty))$ and  any $(s,p,t,q)$ with $m=s+p$ and $n=t+q$ we define the norm
\begin{eqnarray*}
\| w \|_{s,p,t,q}^\flat
&& :=
\esssup_{ (k^{(s)},\widetilde{k}^{(t)}) \in   \R^{3s + 3t}} \Big[ \int
\frac{d X^{(p,q)}}{|X^{(p,q)}|^2}  \\
&& \times  \sup_{r \geq 0} \left| w_{m,n}(r  , k^{(s)},x^{(p)}, \widetilde{k}^{(t)} , \widetilde{x}^{(q)}) \right|^2 ( r + \Sigma[x^{(p)}])^p  ( r + \Sigma[\widetilde{x}^{(q)}])^q \Big]^{1/2}. \nonumber
\end{eqnarray*}

Using Lemma  \ref{kernelopestimate} in  Appendix A, we see that
\begin{align}
&\| \overline{W}^{m,n}_{p,q}[w](r,K^{(m,n)})  \|_{\rm op} \leq \| w_{m+p,n+q} \|^\flat_{m,p,n,q} \label{eq:estonWintermofw}
\end{align}
and if $w_{m,n} \in L^\infty(B_1^{m+n} ; C^1[0, \infty) )$ we have
\begin{align}
&\| \partial_r \overline{W}^{m,n}_{p,q}[w](r,K^{(m,n)})  \|_{\rm op} \leq \| \partial_r w_{m+p,n+q} \|^\flat_{m,p,n,q} , \label{eq:estonWrintermofw}
\end{align}
where the partial derivative $\partial_r \overline{W}^{m,n}_{p,q}[w]$ is understood with respect to the weak operator topology.
\begin{lemma}  \label{initial:thmE} Let $|z| \leq  1/2 $ and $\sigma \geq 0$. Then the following statements are true.
\begin{itemize}
\item[(a)]  If   $m+n+p+q=2$ and $s=0,1$, then  we have
\begin{eqnarray*}
\|  \partial_r^s W_{p,q}^{m,n}[w^{(I)}(\lambda,\sigma,z)](r,K^{(m,n)} ) \|_{\rm op}
\leq |\lambda|^{2} 6  ( \| f /{\omega} \|^2 + 2 \| f \| \| f / \omega \| )^{\frac{p+q}{2}} \| f \|_\infty^{m+n} .
\end{eqnarray*}
\item[(b)]  Let  $|\lambda| \leq \mu_0 $. Then
\begin{equation} \label{eq:thmE22-1}
\left\| ( \partial_r  \bar{F}[w^{(I)}(\lambda,\sigma,z)])(H_f + r) \right\| + \left\| \bar{F}[w^{(I)}(\lambda,\sigma,z)](H_f + r)   \right\|  \leq  10 \| \partial_r \chib_1 \|_\infty + 31 .
\end{equation}
\end{itemize}
\end{lemma}
\begin{proof}
(a).
In view of  inequalities \eqref{eq:estonWintermofw} and \eqref{eq:estonWrintermofw} we need to estimate the
following, where we use  the abbreviation $f(x^{(p)}) = f(x_1) \cdots f(x_p)$,
\begin{eqnarray}
\lefteqn{ \left[ \| \partial_r^s w^{(I)}_{m+p,n+q}(\lambda,\sigma,z) \|^\flat_{m,p,n,q}  \right]^2  }  \nonumber \\
&& \leq |\lambda|^{2} \| f \|_\infty^{m+n} 4 \int_{} \frac{d X^{(p,q)}}{|X^{(p,q)}|^2} \sup_{r \geq 0}  \frac{ |f(x^{(p)})|^2 |f(\widetilde{x}^{(q)})|^2}{|r  + 1 |^{2+s}}
( r + \Sigma[x^{(p)}])^p  ( r + \Sigma[\widetilde{x}^{(q)}])^q  \nonumber  \\
&& \leq  |\lambda|^2 4 \| f \|_\infty^{m+n} E_p(f) E_q(f) ,  \label{eq:basicIestimate}
\end{eqnarray}
where
$$
E_p(f) := \int_{\R^{p}} \frac{d x^{(p)}}{|x^{(p)}|^2} | |f(x^{(p)})|^2
( 1 + \Sigma[x^{(p)}])^p .
$$
we have $E_0(f) =1$ and for $p=1,2$ it is an elementary estimate to obtain $E_{p}(f) \leq 3 ( \| f /\omega \|^2 + 2 \| f \| \| f / \omega \| )^p$.

(b).
We have using    Lemma \ref{initial:thmC}
$$
\| \bar{F}[w^{(I)}(\lambda,\sigma,z)](H_f + r ) \| \leq \left[ \inf_{r \geq 3/4} t^{(I)}(\lambda,\sigma,z)(r) \right]^{-1} \leq 64/15  .
$$
Similarly we estimate the derivative. First we have
$$
\partial_r \bar{F}[w^{(I)}]  = \frac{2 \chib_1 \partial_r \chib_1}{ t^{(I)}} - \frac{\chib_1^2 \partial_r t^{(I)}}{(t^{(I)})^2} .
$$
Now
\begin{align*}
\partial_r t^{(I)}(\lambda,\sigma,z)(r) &=  1  +  \lambda^2 \int  \frac{d^3k}{(4 \pi)^2 \omega(k)} \frac{ |f_\sigma(k)|^2 }{(r + |k| + 2 - z)^2} .
\end{align*}
This yields the estimate  $| \partial_r t^{(I)}(\lambda,\sigma,z)(r)|  \leq   1  +  |\lambda|^2 \| f / ( 4 \pi\sqrt{ \omega} ) \|^2  \leq 1 + 1/64$.
Thus we find using again Lemma \ref{initial:thmC},
$\| \partial_r \bar{F}(H_f + r) \| \leq  10 \| \partial_r \chib_1 \|_\infty + 26 $,
and hence (b) follows.
\end{proof}

\vspace{0.5cm}

\noindent {\it Proof of Lemma \ref{initial:thmE2222}.}
First observe that $ V_{\umm,\upp,\unn,\uqq}[w^{(I)}]$ satisfies the property (i) of the definition of $\WW^\#_{|\umm|,|\unn|}$.
To estimate the norm $\| V_{\umm,\upp,\unn,\uqq}[w^{(I)}] \|^\#$ we use \eqref{eq:babyestimate1}, the estimates of  Lemma  \ref{initial:thmE},
and we calculate the derivative with respect to $r$ using Leibniz rule.
\qed

\vspace{0.5cm}

Using Lemma  \ref{initial:thmE2222}  we are now ready to show    \eqref{eq:initialthmmain}.
To this end we let $S^L_{M,N}$ denote the set of tuples $(\umm,\upp,\unn,\uqq) \in \N_0^{4L}$ with
$|\umm|=M$, $|\unn|=N$, and  $m_l+p_l+q_l+n_l = 2$.
We estimate the  combinatorial factor in \eqref{initial:eq7} by $2^L$ and obtain
 for $z \in D_{1/2}$
\begin{align}
\| w^{(0)}_{\geq 1}(\lambda,\sigma,z) \|^\# &= \sum_{M+N \geq 1} \xi^{-(M+N)} \| \widetilde{w}_{M,N}(\lambda,\sigma,z) \|^\# \nonumber \\
&\leq \sum_{M+N\geq 1} \sum_{L=1}^\infty \sum_{(\umm,\upp,\unn,\uqq) \in S^L_{M,N}}
\xi^{-(M+N)} 2^L \| V_{\umm,\upp,\unn,\uqq}[w^{(I)}(\lambda,\sigma,z)] \|^\#  \nonumber \\
&\leq \sum_{L=1}^\infty \sum_{M+N\geq 1}  \sum_{(\umm,\upp,\unn,\uqq) \in S^L_{M,N} }
\xi^{-|\umm|-|\unn|} (L+1) C_{\bar{F}} \left[ 2  C_W(\lambda) C_{\bar{F}} \right]^L  \nonumber \\
&\leq\sum_{L=1}^\infty (L+1) 10^L \xi^{-2L}  C_{\bar{F}}  \left[ 2 C_W(\lambda) C_{\bar{F}}  \right]^L \; , \label{eq:feb2:1}
\end{align}
where we used \eqref{initial:thmmain:eq2} in the first inequality, and in the second
inequality we used  $|\umm|+|\unn| \leq 2L$ and that the number of elements $(\umm,\upp,\unn,\uqq) \in \N_0^L$ with
$m_l+n_l+p_l+q_l=2$ is bounded by $10^{L}$.
A similar but simpler estimate, yields
\begin{align}
\| \partial_r w^{(0)}_{0,0}(\lambda,\sigma,z) - 1  \|
&\leq  \| \partial_r t^{(I)}(\lambda,\sigma,z) - 1 \| + \sum_{L=2}^\infty  \sum_{(\upp,\uqq) \in \N_0^{2L}: p_l+q_l=2}
 \| V_{\uzz,\upp,\uzz,\uqq}[w^{(I)}(\lambda,\sigma,z)] \|^\# \nonumber \\
&\leq       |\lambda|^2  \| f/(4 \pi \sqrt{\omega}) \|^2  +  \sum_{L=2}^\infty  3^L (L+1) C_{\bar{F}}   \left[ C_W(\lambda) C_{\bar{F}} \right]^L \; . \label{eq:feb2:2}
\end{align}
 Analogously we have
\begin{align}
\|  w^{(0)}_{0,0}(\lambda, \sigma,z)(0) - z  \|
&\leq  \|  t^{(I)}(\lambda,\sigma,z)(0) - z \| + \sum_{L=2}^\infty  \sum_{(\upp,\uqq) \in \N_0^{2L}: p_l+q_l=2}
 \| V_{\uzz,\upp,\uzz,\uqq}[w^{(I)}(\lambda,\sigma,z)]  \|^\#  \nonumber \\
&\leq       |\lambda|^2  \| f/ ( 4 \pi \sqrt{\omega}) \|^2  +  \sum_{L=2}^\infty  3^L (L+1) C_{\bar{F}}   \left[  C_W(\lambda) C_{\bar{F}} \right]^L \; . \label{eq:feb2:3}
\end{align}
In view of the definition of $C_W(\lambda)$  the right hand side in \eqref{eq:feb2:1}--\eqref{eq:feb2:3}  can
be made arbitrarily small for sufficiently small $|\lambda|$. It now remains to show (i) and (ii)  of Theorem \ref{initial:thmmain}.

Part (i) follows from the convergence established in estimates  \eqref{eq:feb2:1}--\eqref{eq:feb2:3},  which is
uniform in $(\lambda,z) \in B_{\lambda_0}  \times D_{1/2}$ for some ball $B_{\lambda_0}$ of nonzero radius,
 and the following lemma.

  \begin{lemma}   \label{lem:analytwI0}    For $(\umm,\upp,\unn,\uqq) \in \N_0^{4L}$ and $\sigma \geq 0$,
the  function
\begin{equation} \label{eq:Vanalyt}
(\lambda,z) \mapsto  V_{\umm,\upp,\unn,\uqq}[w^{(I)}(\lambda,\sigma,z)]
\end{equation}
 is an analytic
$\overline{\WW}_{|\umm|,|\unn|}^\#$-valued function on $B_{\mu_0} \times D_{1/2}$.
\end{lemma}

\begin{proof} The idea of the proof is to show that \eqref{eq:Vanalyt}  is a composition of an analytic
map with an F-differentiable mapping between suitable Banach spaces. The lemma will follow from Steps 1 and 2, below.
First  we introduce the following Banach spaces.
Let
$$
\overline{\WW}_{0,0}^\# := \{ w \in C^1[0,\infty) | \| w \|^\# := \| w \|_\infty + \| \partial_r w \|_\infty \} \ .
$$
Let $\overline{\WW}_{m,n}^\#$ be the Banach space consisting of functions $w_{m,n} \in L^\infty( \R^{3m + 3n} ; C^1[0,\infty))$
satisfying the following properties:
\begin{itemize}
\item[(a)] $w_{m,n} = w_{m,n}^{({\rm sym})}$,
\item[(b)] the following norm is finite
$$
\| w_{m,n} \|^\# := \sup_{\substack{ s + p = m \\ t + q = n}} \| w_{m,n} \|_{s,p,t,q}^\flat
+ \sup_{\substack{ s + p = m \\ t + q = n}} \| \partial_r w_{m,n} \|_{s,p,t,q}^\flat  .
$$
\end{itemize}
Hence for almost all  $K^{(m,n)} \in \R^{3m+3n}$ we have $w_{m,n}(\cdot,K^{(m,n)}) \in C^1[0,\infty)$, where
the derivative is denoted by $\partial_r w_{m,n}$.
Furthermore we introduce the Banach space
$
\overline{\WW}^\# := \overline{\WW}_{0,0}^\# \oplus \overline{\WW}_{0,2}^\# \oplus  \overline{\WW}_{2,0}^\# \oplus \overline{\WW}_{1,1}^\#
$
with norm $\| w \| := \sum_{m+n=0,2} \| w_{m,n} \|^\#$.

\vspace{0.5cm}

\noindent
\underline{Step 1:}  \ \ Let $\epsilon > 0$. Then the  map on $\bar{\mathcal{O}}^{(\epsilon)} := \{ w \in \overline{\WW}^\#  | \inf_{r \geq  \frac{3}{4}} | w_{0,0}(r) | > \epsilon \}$
\begin{eqnarray*}
V_{ \umm,\upp,\unn,\uqq}[\cdot ] \ : \ \bar{\mathcal{O}}^{(\epsilon)}    &\to&  \WW^\#_{|\umm|,|\unn|} \\
w &\mapsto& V_{ \umm,\upp,\unn,\uqq}[w]
\end{eqnarray*}
is F-differentiable.

\vspace{0.5cm}

Step 1 is shown by writing $V_{ \umm,\upp,\unn,\uqq}[\cdot ]$  as a composition of the F-differentiable maps defined in \eqref{eq:defofVtildefunction} and \eqref{eq:defofFfunction}.
Let  $\underline{w} = (w_1,...,w_L)$ with $w_l \in \overline{\WW}_{n_l,m_l}$
and $\underline{G} = (G_1,...,G_L)$ with $G_l  \in \overline{\WW}_{0,0}$.
Then the expression
\begin{eqnarray}
\lefteqn{ \widetilde{V}_{\umm,\upp,\unn,\uqq}[\underline{w},\underline{G}]( r, K^{(|\umm|,|\unn|)}) :=  } \label{eq:defofVtildefunction} \\
&
\left\langle \Omega  ,  G_0( r + \widetilde{r}_0 ) \prod_{l=1}^L
\left\{  \overline{W}_{p_l,q_l}^{m_l,n_l}[w_l](r+r_l,  K_l^{(m_l,n_l)} ) G_l(  r  + \widetilde{r}_l ) \right\} \Omega
\right\rangle  ,  \nonumber
\end{eqnarray}
satisfies the property (a) of the definition of $\WW_{\umm,\unn}^\#$ if $G_0$ and $G_L$ have support contained in $[0,1)$.
Moreover one easily shows the bound
\begin{equation} \label{eq:contmultilinear2}
 \| \widetilde{V}_{\umm,\upp,\unn,\uqq}[\underline{w},\underline{G}] \|^\# \leq \prod_{l=0}^L \left\{ \| G_l \|^\# \right\}   \prod_{l=1}^{L}
  \left\{ \| w_l \|^\#  \right\}
\end{equation}
calculating  the derivative  $\partial_r \widetilde{V}_{\umm,\upp,\unn,\uqq}[\underline{w},\underline{G}]$ using  Leibniz' rule and
estimating the  resulting expression using inequality \eqref{eq:babyestimate1} and estimates
\eqref{eq:estonWintermofw}  and  \eqref{eq:estonWrintermofw}.
To show Step 1  it remains to observe  that  the map on $\bar{\mathcal{O}}_{0,0}^{(\epsilon)}
 :=\{ t \in \overline{\WW}_{0,0}^\#  | {\rm inf}_{r  \geq \frac{3}{4}} | w_{0,0}(r) | > \epsilon \}$
\begin{align} \label{eq:defofFfunction}
 \bar{F} [\cdot ]  : \bar{\mathcal{O}}_{0,0}^{(\epsilon)} & \longrightarrow  \overline{\WW}_{0,0}^\#   \quad , \quad
 t \mapsto \frac{\chib_1^2}{t}
\end{align}
 is F-differentiable. Now this is proved essentially the same way as the first part of Lemma \ref{eq:frechetf}.

\vspace{0.5cm}

\noindent
\underline{Step 2:} \ \
The  function $(\lambda,z) \mapsto w^{(I)}(\lambda,\sigma,z)$ is an analytic $\overline{\WW}^\#$-valued function on $B_{\mu_0} \times D_{1/2}$.

\vspace{0.5cm}

 Let $(\lambda,z) \in B_{\mu_0} \times D_{1/2}$. First observe that $w^{(I)}(\lambda,\sigma,z) \in \overline{\WW}^\#$,
which follows from inequalities  \eqref{eq:thmE22-1} and \eqref{eq:basicIestimate}.
Analyticity in $\lambda$ is trivial since $w^{(I)}$ is a polynomial of second degree in $\lambda$,
with coefficients which are elements of $\overline{\WW}^\#$, which again follows from inequality \ref{eq:basicIestimate} and estimates used
to show  \eqref{eq:thmE22-1}.
To show analyticity in $z$ we
show the following estimate for $\lambda \in B_{\lambda_0}$ and $z,z+h \in D_{1/2}$,
\begin{equation} \label{eq:analyticityinzwovereline}
\|  \left( w^{(I)}(\lambda,\sigma,z+h) - w^{(I)}(\lambda,\sigma,z) \right) -  \partial_z w^{(I)}(\lambda,\sigma, z) \|^\# = o(1) , \quad h \to 0  .
\end{equation}
First we need to show that  $\partial_z w^{(I)}(\lambda,\sigma,z)  \in  \overline{\WW}^\#$. To show this, consider
the denominator occurring in in the expressions for $w_{m,n}^{(I)}$,
$$
Q(r,a,z) := \frac{1}{r + 2 + a - z} ,
$$
where $a$ stands for $0$, $|k_1|$, $|\widetilde{k}_1|$ or  $|k_1| + |\widetilde{k}_1|$.
Now   $|\partial_r^s \partial_z Q(r,a,z) | \leq | Q(r,a,z) |$ for $r \geq 0$,
$z \in D_{1/2}$, and $s=0,1$.  Thus   $\partial_z w^{(I)}(\lambda,\sigma,z)  \in  \overline{\WW}^\#$ follows from
the corresponding estimate for $w^{(I)}$. To show \eqref{eq:analyticityinzwovereline} it suffices to consider the difference quotients of the
denominators $Q(r,a,z)$,
\begin{equation}
 \frac{1}{h} \left( Q(r,a,z+h) - Q(r,a,z)  \right) - \partial_z Q(r,a,z) = h P(r,a,z,h)  ,
\end{equation}
with
\begin{equation}
 P(r,a,z,h) :=  \frac{1}{(r + 2 + a+z)^2(r + 2 + a - z + h)} .
\end{equation}
If $h$ is sufficiently small than $| \partial_r^s P(r,a,z,h) | \leq 2 | Q(r,a,z) |$
for $r \geq 0$, $z \in D_{1/2}$, and $s=0,1$.
Thus \eqref{eq:analyticityinzwovereline}, and hence the lemma, now follow again from the corresponding estimates for $w^{(I)}$.
\end{proof}

(ii) of Theorem \ref{initial:thmmain} follows from the convergence established in estimates  \eqref{eq:feb2:1}--\eqref{eq:feb2:3},  which is
uniform in $(\sigma,z) \in [0,\infty)  \times D_{1/2}$,
 and the following lemma.

\begin{lemma}  \label{lem:contwI}  Let $|\lambda| < \mu_0$. Then the function
$(\sigma,z) \mapsto V_{\umm,\upp,\unn,\uqq} [w^{(I)}(\lambda,\sigma,z)]$ is a continuous
function on   $D_{1/2} \times [0,\infty)$ with respect to the $\| \cdot \|_2$-norm.
\end{lemma}

The proof uses essentially the same idea as the proof of Lemma  \ref{lem:newcont}.
\begin{proof}
First observe  that  the kernel  $V_{\umm,\upp,\unn,\uqq}$ is a multi-linear expression of integral kernels, thus to show c-continuity
we will use \eqref{eq:telescoping} and \eqref{eq:babyestimate1}. Moreover, we
 we will use   the estimate
\begin{eqnarray*}
\left[ \int_{B_1^{m+n} } \frac{dK^{(m,n)}}{|K^{(m,n)}|^2} \sup_{r \geq 0} \left\| \overline{W}[w]_{p,q}^{m,n}(r,K^{(m,n)} )\right\|^2_{\rm op}  \right]^{1/2} \leq  \| w_{m+p,n+q} \|^\sharp ,
\end{eqnarray*}
with
\begin{eqnarray*}
 \lefteqn { \| w_{m,n} \|^\sharp } \\
&& :=
\left[ \int
\frac{d K^{(m,n)}}{|K^{(m,n)}|^2}  \sup_{r \geq 0}
\left| w_{m,n}(r  , k^{(m)}, \widetilde{k}^{(n)})  \right|^2 ( 1 + r + \Sigma[k^{(m)}])^m  ( 1 + r +  \Sigma[\widetilde{k}^{(n)}])^n \right]^{1/2}. \nonumber
\end{eqnarray*}
which follows from Lemma \ref{kernelopestimate}.
Above estimates together with  the estimates in  Lemma \ref{initial:thmE} and the
limits \eqref{lem:contwIproofeq1} and  \eqref{lem:contwIproofeq2}, shown next, yield c-continuity.
From the explicit expression of $w^{(I)}$  it follows from dominated convergence that   for $m+n=2$
\begin{equation}
\lim_{  (z,\sigma) \to (z_0,\sigma_0) }  \left\| w_{m,n}^{(I)}(\lambda,\sigma_0,z_0) -  w_{m,n}^{(I)}(\lambda,\sigma,z) \right\|^\sharp =  0 ,\label{lem:contwIproofeq1}
\end{equation}
and it follows that
\begin{equation}
\lim_{  (z,\sigma) \to (z_0,\sigma_0) } \sup_{r \in [0,\infty)}
\left|
  \frac{\chib_1^2(r)}{t^{(I)}(\lambda,\sigma_0,z_0)(r)} -    \frac{\chib_1^2(r)}{t^{(I)}(\lambda,\sigma,z)(r)}   \right| =  0  \label{lem:contwIproofeq2} .
\end{equation}
\end{proof}

\section{Proving the Main Theorem}

\label{sec:prov}

In this section, we prove Theorem  \ref{thm:main2}, the main result of this paper.
The basic idea behind the proof of Theorem   \ref{thm:main2}, is that the ground state and its energy are given as
limits of  uniformly convergent sequences having terms which are analytic in $\lambda$ and continuous in $\sigma$, respectively.
A similar idea was  used in \cite{GH09} to show the analyticity in the coupling constant
of the ground state and ground state energy. The exposition  in \cite{GH09} is  different
due to the different representation of the spectral parameter.

\vspace{0.5cm}

\noindent
{\it Proof of Theorem \ref{thm:main2}.} Choose  $\rho, \xi, \epsilon_0$ such that Hypothesis (R) holds.
By Theorem \ref{initial:thmmain} we know that there exists a positive $\lambda_0$ which is less or equal than  $\mu_0$ such that for
all $ \lambda \in B_{\lambda_0}$ and $\sigma \geq0$ we have
$w^{(0)}(\lambda,\sigma)  := w^{(0)}(\lambda,\sigma,\cdot)  \in B(\epsilon_0/2,\epsilon_0/2,\epsilon_0/2)$,  and that the analyticity, continuity and
symmetry property as stated in (i)-(iii) of Theorem \ref{initial:thmmain} hold.
By Theorem  \ref{thm:bcfsmain}   and a twofold application of the Feshbach property, see Theorem \ref{thm:fesh}, it follows that
$E_\sigma(\lambda) = e_{(0,\infty)}[w^{(0)}(\lambda,\sigma)]$ is an eigenvalue
of $H_{\lambda,\sigma}$ with nonzero eigenvector $\psi_\sigma(\lambda) =  Q_{\lambda,\sigma} Q_{\lambda,\sigma}^{(I)} \psi_{(0,\infty)}[w^{(0)}(\lambda,\sigma)]$ where
\begin{align*}
Q_{\lambda,\sigma} &:= Q_{\chi^{(I)}}(H_{\lambda,\sigma} - E_\sigma(\lambda) ,\tau + H_f - E_\sigma(\lambda))  \\
Q_{\lambda,\sigma}^{(I)} &:=  Q_{\chi_1}(H^{(I)}_{\lambda,\sigma}(E_\sigma(\lambda))  , T_{\lambda,\sigma}^{(I)}(E_\sigma(\lambda))  )  .
\end{align*}
Note that the first two Feshbach transformations do not involve any transformation of the spectral parameter.
By Theorem \ref{thm:continsigma}  it follows that $E_\sigma(\lambda)$ and $\psi_{(0,\infty)}[w^{(0)}(\lambda,\sigma)]$ are analytic in $\lambda \in B_{\lambda_0}$ for all $\sigma \geq 0$ and continuous in $\sigma \in [0,\infty)$ for all $\lambda \in B_{\lambda_0}$.
It follows  using Theorem \ref{initial:thmA} that $(\lambda,\sigma) \mapsto Q_{\lambda,\sigma}$ is uniformly bounded on $B_{\lambda_0} \times [0,\infty)$, analytic in $\lambda$, and continuous in $\sigma$ where the continuity follows from  estimate (vi) in Lemma \ref{thm:estimates1}.
Similarly it follows using Theorem \ref{initial:thmC} that  $(\lambda,\sigma) \mapsto Q_{\lambda,\sigma}$ is uniformly bounded on $B_{\lambda_0} \times [0,\infty)$, analytic in $\lambda$, and continuous in $\sigma$ 
Now it follows that $\psi_\sigma(\lambda)$ is analytic for $\lambda \in B_{\lambda_0}$, continuous in $\sigma \geq 0$. Thus we have shown (ii).
Next we show that the expansion coefficients are bounded and continuous in $\sigma$ and use Cauchy's formula
$$
E_\sigma^{(n)} = \frac{1 }{2 \pi i} \int_{|\lambda|= r} \frac{E_\sigma(\lambda)}{\lambda^{n+1}} d\lambda  \quad , \quad
\psi_\sigma^{(n)} = \frac{1 }{2 \pi i} \int_{|\lambda|= r} \frac{\psi_\sigma(\lambda)}{\lambda^{n+1}} d\lambda  ,
$$
for some  positive $r$ which is less than $\lambda_0$.
Using Cauchys formula it follows that
 $\sup_{\sigma \geq 0 } \| \psi_\sigma^{(n)}\| $ is bounded because of the uniform bound $\|\psi_\sigma(\lambda)\| \leq 4 e^4$, see Theorem \ref{thm:bcfsmain}, and the boundedness
of $Q_{\lambda,\sigma}$ and   $Q_{\lambda,\sigma}^{(I)}$. Moreover $\sigma \mapsto \psi_\sigma^{(n)}$ is continuous, which follows from Cauchy's formula
and dominated convergence. Likewise it follows that $\sup_{\sigma \geq 0} | E_\sigma^{(n)}|$  is  bounded, since by definition $| E_\sigma( \lambda)| \leq 1/2$.
Again using dominated convergence and Cauchy's formula we see that $\sigma \mapsto  E_\sigma^{(n)}$ is continuous.
By  possibly choosing  $\lambda_0$ smaller but still positive one can ensure that  for all $\sigma \geq 0$ the projection
\begin{equation} \label{eq:projectionlambda}
P_\sigma(\lambda) := \frac{ \left|         \psi_\sigma({\lambda})      \right\rangle \left\langle    \psi_\sigma(\overline{\lambda})              \right| }{ \left\langle    \psi_\sigma(\overline{\lambda})       ,         \psi_\sigma({\lambda})      \right\rangle }
\end{equation}
is well defined
for $| \lambda | < {\lambda}_0$. To this end we need to show the uniformity in $\sigma$.
By choosing a phase and a suitable normalization  we can assume that  $\psi_\sigma(\lambda)  = \Omega_{\downarrow} + \sum_{n=1}^\infty
\psi_\sigma^{(n)}\lambda^n$ with radius of convergence greater or equal to $\lambda_0$,
where $\Omega_{\downarrow}$ is defined in \eqref{eq:perturbproj00}.
Since we have shown that   $\psi_\sigma^{(n)}$ can be estimated
uniformly in $\sigma \geq 0$, it follows that  the denominator in  \eqref{eq:projectionlambda} can be estimated from below by a positive constant uniformly in $\sigma \geq 0$.
In view of
\eqref{eq:projectionlambda}, it follows that $\lambda \mapsto P_\sigma(\lambda)$ is an analytic function on $B_{\lambda_0}$ and that $P_\sigma(\lambda)^* = P_\sigma(\overline{\lambda})$.
Thus we have shown (iii).
The continuity of $P_\sigma^{(n)}$ in $\sigma$ follows from
$$
P_\sigma^{(n)} = \frac{n! }{2 \pi i} \int_{|\lambda|= r} \frac{P_\sigma(\lambda)}{\lambda^{n+1}} d\lambda , \quad 0 < r < \lambda_0 ,
$$
 dominated convergence, and in view of \eqref{eq:projectionlambda} the continuity of $\psi_\sigma(\lambda)$ in $\sigma$.
To show (i) assume that $\lambda \in B_{\lambda_0}$ is real. Then $w^{(0)}(\lambda,\sigma)$ is a symmetric kernel, see Theorem \ref{initial:thmmain}.
It now follows from  Theorem \ref{thm:bcfsmain} that $H^{(0)}_{\lambda,\sigma}(z)$ is bounded invertible if $z \in (-\frac{1}{2}, E_\sigma(\lambda))$. Applying the Feshbach property twice
it follows that $H_\lambda-z$ is bounded invertible for $z \in (-\frac{1}{2}, E_\sigma(\lambda))$. For $ z \leq  - 1/2$ the bounded invertibility of $H_{\sigma,\lambda} - z$ follows from
the  estimate  $\| \lambda  \sigma_x \phi(f_\sigma) (H_f + \tau + \frac{1}{2} )^{-1} \| \leq 6 | \lambda | /( 8 \mu_0)   < 1 $, see Lemma \ref{thm:estimates1}.
Thus $E_\sigma(\lambda) = \inf \sigma(H_{\sigma,\lambda})$ for real $\lambda \in B_{\lambda_0} \cap \R$. The uniqueness follows from Theorem  \ref{thm:uniqueness}.
Finally observe that $(-1)^N H_{\lambda,\sigma}(-1)^N = H_{-\lambda,\sigma}$ where $N$ is the closed linear operator on $\FF$ with $N \upharpoonright S_n(\hh^{\otimes n}) = n$.
This implies that the ground state energy $E_{\sigma}(\lambda)$ cannot depend on odd powers of $\lambda$.
\qed

\section{Analytic Perturbation Theory}
\label{sec:ana}

In this section we discuss analytic perturbation theory.
We put the discussion in Section \ref{sec:mod} about analytic perturbation theory
on a sound mathematical footing and justify Equation  \eqref{eq:perturbproj0}, by proving Theorem
\ref{thm:perturb1} below. Moreover, we elaborate on  Remark   \ref{thm:formalperturb00} 
at the end of this section.

\begin{theorem} \label{thm:perturb1}
For $\sigma > 0$, there is a $\lambda_0(\sigma) > 0$ such that for all
$\lambda \in B_{\lambda_0(\sigma)}$, the Hamiltonian
$H_{\lambda,\sigma}$ has a non-degenerate eigenvalue
$\widehat{E}_\sigma(\lambda)$ with eigen-projection $\widehat{P}_\sigma(\lambda)$ such that
\begin{itemize}
\item[(i)] $\widehat{E}_{\sigma}(\lambda) = \inf \sigma(H_{\lambda,\sigma})$ for $\lambda \in B_{\lambda_0(\sigma)}$  and $\widehat{E}_\sigma(0) = 0$.
\item[(ii)] $\lambda \mapsto \widehat{E}_\sigma(\lambda)$ and $\lambda \mapsto \widehat{P}_\sigma(\lambda)$ are analytic functions on $B_{\lambda_0(\sigma)}$.
\item[(iii)] $\widehat{P}^*_\sigma(\lambda) = \widehat{P}_\sigma(\overline{\lambda})$ for all $\lambda \in B_{\lambda_0(\sigma)}$.
\end{itemize}
Moreover on $B_{\lambda_0(\sigma)}$ we have  a convergent power series expansion $\widehat{P}_\sigma(\lambda) = \sum_{n=0}^\infty \widehat{P}_\sigma^{(n)} \lambda^n$,
where $\widehat{P}^{(n)}_{\sigma}$ is given by   \eqref{eq:perturbproj0}.
\end{theorem}

We will use the notation $P_\downarrow := \left( \begin{array}{cc} 0 & 0 \\ 0 & 1 \end{array} \right)$.

\begin{proof} Define the subspaces of $\hh$,

$$
\hh_\sigma^{(+)} := L^2 (\{ k \in \R^3 | |k| \geq \sigma \}) \quad , \quad \hh_\sigma^{(-)} := L^2 (\{ k \in \R^3 | |k| <  \sigma \})
$$
and the associated Fock-spaces $\FF_{\sigma}^{(\pm)} := \FF(\hh_{\sigma}^{(\pm)})$, where we denote the vacua by $\Omega_\sigma^{(\pm)}$.
We consider the natural unitary isomorphism
$$
U : \FF_{\sigma}^{(+)} \otimes \FF_{\sigma}^{(-)} \to \FF_s \; ,
$$
which is uniquely characterized by
$$
U \left(  S_n ( h_1 \otimes  \cdots \otimes h_n )  \otimes  S_m ( g_1 \otimes  \cdots \otimes g_m )  \right)
=  S_{n+m} ( h_1 \otimes \cdots  \otimes h_n \otimes  g_1 \otimes  \cdots \otimes g_m ) ,
$$
for any $h_1,...,h_n  \in \hh_{\sigma}^{(+)}$ and $g_1,...,g_m \in \hh_{\sigma}^{(-)}$.
We denote the trivial extension of $U$ to  $\C^2 \otimes  \FF_{\sigma}^{(+)} \otimes \FF_{\sigma}^{(-)}$ by the same symbol.
We  write
$$
U^* H_\sigma(\lambda) U = \left( H_{0,\sigma}  + \lambda T_\sigma^{(+)} \right) \otimes 1 +  1 \otimes 1 \otimes H_{f,\sigma}^{(-)}  \; ,
$$
where we introduced the following operators acting on the corresponding spaces
\begin{align*}
H_{0,\sigma} &:= \tau  \otimes 1 + 1 \otimes H_{f,\sigma}^{(+)}  \\
T_\sigma^{(+)} &:= \sigma_x \otimes \phi(\chi_\sigma f) , \quad
H_{f,\sigma}^{(+)}   := d \Gamma (\chi_\sigma \omega) , \quad
H_{f,\sigma}^{(-)}   := d \Gamma ((1-\chi_\sigma) \omega)           .
\end{align*}
Now observe that  $H_{f,\sigma}^{(-)}$ has only one eigenvalue. That eigenvalue is zero, it is at the bottom
of the spectrum, it is non-degenerate and  and its
eigenvector is the vacuum of $\FF_{\sigma}^{(-)}$. This implies
that $H_\sigma(\lambda)$ and $H_{0,\sigma} + \lambda T^{(+)}_\sigma$ have the same eigenvalues and the corresponding
eigen-spaces are in bijective correspondence.
Next observe that $H_{0,\sigma}$ has at the bottom of its spectrum  an isolated non-degenerate eigenvalue.
In fact we have $\sigma(H_{0,\sigma}) = \{ 0 \} \cup [\sigma , \infty)$.
Moreover, $\lambda \mapsto H_{0,\sigma} + \lambda T_\sigma^{(+)}$ is
an analytic family, since the interaction term is infinitesimally bounded with respect to $H_{0,\sigma}$.
Now by analytic perturbation theory, it follows that there exists an $\epsilon > 0$
such that for $\lambda$ in a neighborhood of zero  the following operator
is well defined
\begin{align} \label{eq:orginialprojection}
P_{\sigma}^{(+)}(\lambda)  &:=  - \frac{1}{2 \pi i} \int_{|z| = \epsilon} (H_{0,\sigma}^{(+)} + \lambda T_\sigma^{(+)}  - z )^{-1} dz .
\end{align}
Moreover, the operator $P_{\sigma}^{(+)}(\lambda) $  projects
onto a one-dimensional space which is the eigen-space of $H_{0,\sigma}^{(+)} + \lambda T^{(+)}_\sigma$ with eigenvalue
$E_\sigma(\lambda)$. Furthermore, $P_{\sigma}^{(+)}(\lambda) $ and $E_\sigma(\lambda)$ depend analytically on $\lambda$
and $E_\sigma(0) = 0$. We conclude that $E_\sigma(\lambda)$  is a non-degenerate eigenvalue of
$H_{\lambda,\sigma}$ with corresponding eigen-projection
\begin{equation} \label{eq:orginialprojection1}
\widehat{P}_\sigma(\lambda) = U ( P_{\sigma}^{(+)}(\lambda)\otimes P_{\Omega_\sigma^{(-)}      } ) U^* ,
\end{equation}
and properties (i)-(iii) of the theorem are satisfied, where $P_{\Omega_\sigma^{(\pm)}}$ denotes the orthogonal projection in
$\FF_{\sigma}^{(\pm)}$
onto $\Omega_\sigma^{(\pm)}$.
Expanding in a Neumann series we find
\begin{align*}
P_{\sigma}^{(+)}(\lambda)
& =  - \sum_{n=0}^\infty \lambda^n  \frac{1}{2 \pi i} \int_{|z| = \epsilon} (H_{0,\sigma}^{(+)}  - z )^{-1} \left( T_\sigma^{(+)}  (H_{0,\sigma}^{(+)}  - z )^{-1} \right)^n  dz .
\end{align*}
Inserting the Laurent expansion of the resolvent about 0,
$$
(H_{0,\sigma}^{(+)}  - z )^{-1}  = \sum_{\nu=0}^\infty S_\sigma^{(\nu,+)}  z^{\nu-1}
$$
where
$$
S_\sigma^{(\nu,+)} = \left\{ \begin{array}{ll} -   P_\downarrow \otimes P_{\Omega_\sigma^{(+)}}  \quad &, \ \nu = 0 \\   \left(H_{0,\sigma}^{(+)}\right)^{-\nu} (1- P_\downarrow \otimes P_{\Omega_\sigma^{(+)}})

 \quad &, \ \nu \geq 1  , \end{array} \right.
$$
 and calculating the contour integral we arrive at $P_{\sigma}^{(+)}(\lambda) = \sum_{n=0}^\infty P_{\sigma}^{(n,+)}\lambda^n$, with
\begin{equation} \label{eq:perturbproj}
P^{(n,+)}_{\sigma} = -  \sum_{\nu_1 + ... + \nu_{n+1} = n, \ \nu_i \geq 0} S_\sigma^{(\nu_1,+)} T_\sigma^{(+)} S_\sigma^{(\nu_2,+)} ... T_\sigma^{(+)} S_\sigma^{(\nu_{n+1},+)} .
\end{equation}
Now using identity  \eqref{eq:orginialprojection1}, we find \eqref{eq:perturbproj0} noting that  that $U ( 1 \otimes P_{\Omega_\sigma^{(-)}}) U^* = Q_\sigma$ and
$U( S_\sigma^{(\nu,+)} \otimes P_{\Omega_\sigma^{(-)}}  ) U^* = S_\sigma^{(\nu)}  $.
\end{proof}

Next we will elaborate on the  statement of Remark \ref{thm:formalperturb00}.
We  calculate the first four coefficients of $P(\lambda)$ using
\eqref{eq:perturbproj0}.
We set  $b(f) = a(f_\sigma /\sqrt{\omega})$ and $b^*(f) = a^*(f_\sigma /\sqrt{\omega})$. To keep the notation
simple the  integration  symbol   $\int$ stands in this section for $\prod_{i} \int_{|k_i| \geq \sigma}$ and we
drop the $\sigma$-subscript of $\widehat{P}^{(n)}_\sigma$, $S^{(\nu)}_\sigma$, and $T_\sigma$.
By divergent we mean an expression which diverges in the limit $\sigma \downarrow 0$. We will write $P_{\downarrow}$ for
a short hand notation of $ P_{\downarrow} \otimes 1$ and define $P_\Omega := 1 \otimes | \Omega \rangle \langle \Omega |$.
We find
$$
\widehat{P}^{(0)} =  -  S^{(0)} =  P_{\downarrow} P_\Omega .
$$
Furthermore, we have
\begin{align*}
\widehat{P}^{(1)} = - S^{(1)} T S^{(0)} -  S^{(0)} T S^{(1)} =     ( 2 + H_f)^{-1} b^*(f) P_{\downarrow} P_\Omega +  P_\Omega P_{\downarrow}  b(f) ( 2 + H_f)^{-1} .
\end{align*}
To determine  higher order expressions, we note that terms of the form $S^{(0)} T S^{(0)}$ vanish. We find
\begin{align*}
 P^{(2)} &=  - S^{(0)} T S^{(1)} T S^{(1)} -   S^{(1)} T S^{(0)} T S^{(1)} -  S^{(0)} T S^{(1)} T S^{(1)}  -  S^{(0)} T S^{(2)} T S^{(0)}
\end{align*}
with
\begin{align*}
S^{(0)} T S^{(2)} T S^{(0)}  &=       \int \frac{d^3k |f(k)|^2}{(4 \pi)^2|k|(|k| + 2)^2}     P_\downarrow P_\Omega \\
S^{(1)} T S^{(1)} T S^{(0)}  &=       -  H_f^{-1} b^*(f) (H_f + 2)^{-1} b^*(f)    P_\downarrow P_\Omega \\
S^{(0)} T S^{(1)} T S^{(1)}  &=    -     P_\downarrow P_\Omega  b(f) (H_f + 2)^{-1} b(f)   H_f^{-1} \\
S^{(1)} T S^{(0)} T S^{(1)}  &=    -     P_\uparrow  (H_f + 2 )^{-1}   b^*(f)  P_\Omega  b(f) (H_f + 2)^{-1} P_\uparrow .
\end{align*}
One checks that all above terms are not divergent. Similarly, one finds that $P^{(3)}$ does not contain any divergent terms.
Next we consider the following terms occurring  in   $\widehat{P}^{(4)}$
$$
A := S^{(1)} T S^{(1)} T S^{(1)} T S^{(1)} T S^{(0)} , \quad B := S^{(2)}  T S^{(1)} T  S^{(0)} T S^{(1)} T S^{(0)} .
$$
Inserting the definition of $T$ we find
$
A = A_1 + A_2 + A_3
$
with
\begin{align*}
A_1 & = - H_f^{-1} b(f) (H_f + 2)^{-1} b^*(f) H_f^{-1} b^*(f) (H_f + 2 )^{-1} b^*(f) P_{\downarrow} P_\Omega \\
A_2 & =  -  H_f^{-1} b^*(f) (H_f + 2)^{-1} b(f) H_f^{-1} b^*(f) (H_f + 2 )^{-1} b^*(f) P_{\downarrow} P_\Omega \\
A_3 & =  -  H_f^{-1} b^*(f) (H_f + 2)^{-1} b^*(f) H_f^{-1} b^*(f) (H_f + 2 )^{-1} b^*(f) P_{\downarrow} P_\Omega .
\end{align*}
Above we used that the other contributions to $A$ vanish, which can be seen by using the following identity which holds
 for $k\geq 1$,
$$
S^{(k)} = H_f^{-k} P_\Omega^\perp P_\downarrow
+ ( 2 + H_f)^{-k} P_\Omega^\perp P_\uparrow + 2^{-k} P_\uparrow P_\Omega \ ,
$$
with $P_\uparrow :=  1 - P_\downarrow$ and $P_\Omega^\perp := 1 - P_{\Omega}$.
Using the pull-through formula and  the canonical commutation relations, we find using $dk^{(3)} = d^3k_1 d^3 k_2 d^3 k_3$,
\begin{align*}
 A_1
& =  - \int \frac{  dk^{(3)}  |f(k_1)|^2 f(k_2) f(k_3)}{(4 \pi)^4 |k_1| |k_2|^{1/2} |k_3|^{1/2}} \left[ \frac{1}{|k_2| + |k_3|   } \right]    \left[ \frac{1}{ |k_1| + |k_2| + |k_3|  +2}\right]   \\
& \times
\left\{  \frac{1}{|k_2| + |k_3|   }  \frac{1}{ |k_3|   +2} +    \frac{1}{|k_1|+|k_3| } \frac{1}{|k_3|+ 2  } +
        \frac{1}{|k_3|+|k_1| } \frac{1}{|k_1|+2 } \right\}
 a^*(k_2) a^*(k_3) P_\Omega P_{\downarrow} .
\end{align*}
Note that only the first term in the brackets $\{ \cdots \}$ yields a  divergent expression.  Similarly  one
finds
\begin{align*}
 A_2  & =  - \int \frac{ dk^{(3)}  |f(k_1)|^2 f(k_2) f(k_3)}{(4 \pi)^4 |k_1| |k_2|^{1/2} |k_3|^{1/2}}
\left[ \frac{1}{   |k_2|   + |k_3|    } \right]
\left[ \frac{1}{ |k_3|   +2}\right]
\left[\frac{1}{ |k_1|  + |k_3|   }\right] \\
&  \times  \left\{ \frac{1}{ |k_3|    +2} + \frac{1}{ |k_1|    +2} \right\}
 a^*(k_2) a^*(k_3) P_\Omega P_{\downarrow} .
\end{align*}
One sees that  $A_2$ is  not divergent. Likewise one checks that $A_3$ is not divergent.
For $B$ we obtain
\begin{align*}
 B
 = \int \frac{   dk^{(3)} |f(k_1)|^2 f(k_2) f(k_3)}{( 4 \pi)^4 |k_1| |k_2|^{1/2} |k_3|^{1/2}}
\left[ \frac{1}{ |k_2| + |k_3|}\right]^2
\left[\frac{1}{ |k_2|    + 2  }\right]
\left[\frac{1}{ |k_1|    +2} \right]
a^*(k_2) a^*(k_3) P_\Omega P_{\downarrow} .
\end{align*}
 Now one sees that the divergence in $B$ and the divergence in $A$ cancel, which follows using the  identity
\begin{align*}
 \frac{1}{ |k_1| + 2 } -
\frac{1}{ |k_1| + |k_2| + |k_3|    + 2  }
=
\frac{1}{ |k_1|    + 2}  \left( |k_2| + |k_3|  \right)  \frac{1}{ |k_1|  +|k_2| + |k_3|    +2} ,
\end{align*}
and the symmetry of the corresponding expressions with respect to  $k_2$ and $k_3$.

\section*{Acknowledgements}

D.H. acknowledges the summer research grant of the College of William \& Mary. D.H. wants to thank  ETH Zurich for hospitality.
D.H. wants to thank G.M. Graf, J. Fr\"ohlich, M. Griesemer, and A. Pizzo for interesting conversations.

\section*{Appendix A: Elementary Estimates and the Pull-through Formula}
\label{sec:appA}

To give a precise meaning to expressions which occur in  \eqref{eq:defhmnrig} and \eqref{eq:defhlinemn},  we introduce the following definition.
For $\psi$ having finitely many particles  we set
\beqn \label{eq:defofa}
\left[ a(k_1) \cdots a(k_l) \psi \right]_n(k_{l+1},...,k_{l+n}) :=  \sqrt{\frac{(l+n)!}{n!}} \psi_{l+n}(k_{1},...,k_{l+n}) ,
\eeqn
for $k_1,...,k_{l+n} \in \R^3$.
Using a theorem of Fubini it is elementary to see that for such $\psi$ the vector valued map
 $(k_1,...,k_l) \mapsto a(k_1) \cdots a(k_l) \psi$ is
an element of $L^2(\R^{3l}; \FF)$. We note that  definition \eqref{eq:defofa} is consistent with \eqref{eq:formala}, since
a straightforward computation shows that
\begin{eqnarray*}
\lefteqn{ \int \overline{f_1}(k_1) \cdots \overline{f_l}(k_l) \left[ a(k_1) \cdots a(k_l) \psi \right]_n(k_{l+1},...,k_{l+n})  d^3k_1 \cdots d^3k_l  }
\\ & \quad \quad  = \left[  a(f_1) \cdots a(f_l) \psi \right]_n(k_{l+1},...,k_{l+n}) .
\end{eqnarray*}

We will make repeated use of the well known pull-through formula which is stated in the following lemma.
\begin{lemma} \label{lem:pullthrough}
Let $f : \R_+ \to \C$ be a bounded measurable function. Then
$$
f(H_f) a^*(k) = a^*(k) f(H_f + \omega(k) ) , \quad a(k) f(H_f) = f(H_f + \omega(k) ) a(k) \; .
$$
\end{lemma}
The pull through formula can be shown as follows. Using definition \eqref{eq:defofa},
we have for any $\psi \in \FF$ with finitely many particles
\begin{align*}
[ f(H_f + \omega(k) ) a(k) \psi]_{n}(k_1,...,k_n) &= f(\omega(k_1) + \cdots + \omega(k_n) + \omega(k) ) \sqrt{n+1}  \psi_{n+1}(k,k_1,...,k_n) \\
&= [ a(k)  f(H_f  ) \psi]_{n}(k_1,...,k_n) .
\end{align*}
This shows the second identity, since vectors with finitely many particles are dense in $\FF$.
The first identity  follows by taking the adjoint of the second identity.

\begin{remark}
{ Note that for $\psi \in \HH_{\rm red}$, the statements of Lemma \ref{lem:multanihiest} and Corollary \ref{lem:multanihiest2}
also hold  in the case $X = B_1$.}
\end{remark}

\begin{lemma} \label{lem:multanihiest} Let $X = \R^3$ and $  P_\Omega^\perp = 1 - | \Omega \rangle \langle \Omega |$. Then for $n \geq 1$,
\begin{eqnarray}
\lefteqn{ \int_{X^n}  d k^{(n)} | k^{(n)} | \left\| \prod_{l=1}^n \left[  H_f + \Sigma[k^{(l)}] \right]^{-1/2} a(k^{(n)}) \psi \right\|^2 }   \nonumber  \\
&&=
\int_{X^n}  d k^{(n)} | k^{(n)} | \left\| \prod_{l=1}^n \left[ a(k_l) H_f^{-1/2} \right]  P_\Omega^\perp \psi \right\|^2   \nonumber \\
&&= \| P_\Omega^\perp \psi \|^2 . \nonumber
\end{eqnarray}
\end{lemma}
\begin{proof}
The first identity follows from the pull-through formula. The second identity follows from iteration of the following
equality where $\phi = P_\Omega^\perp \psi$,
$$
\int_{X} d^3k |k| \| a(k) H_f^{-1/2} \phi \|^2 = \langle H_f^{-1/2} \phi , H_f H_f^{-1/2} \phi \rangle = \| \phi \|^2 .
$$
\end{proof}

\begin{corollary} \label{lem:multanihiest2} Let $X = \R^3$. Then
\begin{eqnarray} \nonumber
\int_{X^n}  d k^{(n)} | k^{(n)} | \left\|  a(k^{(n)})  \psi \right\|^2  \leq \| H_f^{n/2} \psi \|^2 .
\end{eqnarray}
\end{corollary}
\begin{proof}
\begin{eqnarray*}
\lefteqn{ \int_{X^n}  d k^{(n)} | k^{(n)} | \left\|  a(k^{(n)})  \psi \right\|^2  } \nonumber \\
&&= \int_{X^n}  d k^{(n)} | k^{(n)} | \left\| (H_f + \Sigma[k^{(n)}] )^{-n/2}  a(k^{(n)})  H_f^{n/2} \psi \right\|^2 \nonumber \\
&&\leq  \int_{X^n}  d k^{(n)} | k^{(n)} | \left\| \prod_{l=1}^n \left[  H_f + \Sigma[k^{(l)}] \right]^{-1/2} a(k^{(n)} ) H_f^{n/2}\psi \right\|^2  \nonumber \\
&&= \| H_f^{n/2} \psi \|^2 ,
\end{eqnarray*}
where in the first equality we used the pull-through formula. The second line follows from an elementary operator inequality using
 $\Sigma[k^{(l)}] \leq \Sigma[k^{(n)}]$, if $l \leq n$.
The last equality follows from Lemma
\ref{lem:multanihiest}.
\end{proof}

\begin{lemma} \label{kernelopestimate}  For $\overline{H}_{m,n}(\cdot)$ as defined in \eqref{eq:defhlinemn}
we have
\begin{eqnarray}
\lefteqn{
\| \overline{H}_{m,n}(w_{m,n}) \|^2 } \nonumber \\
&&\leq  \int  \frac{d K^{(m,n)}}{|K^{(m,n)}|^2} \sup_{r \geq 0} |w_{m,n}(r,K^{(m,n)}) |^2  \prod_{l=1}^m \left[  r + \Sigma[k^{(l)}] \right]
 \prod_{\widetilde{l}=1}^n \left[  r + \Sigma[\widetilde{k}^{(\widetilde{l})}] \right] \nonumber \\
 &&
 \leq   \int  \frac{d K^{(m,n)}}{|K^{(m,n)}|^2} \sup_{r \geq 0} | w_{m,n}(r,K^{(m,n)}) |^2  \left[  r + \Sigma[k^{(m)}] \right]^m
  \left[  r + \Sigma[\widetilde{k}^{({n})}] \right]^n \nonumber
\end{eqnarray}
\end{lemma}
\begin{proof}
The first inequality is obtained by estimating
$$
| \langle \phi, \overline{H}_{m,n}(w_{m,n}) \psi \rangle | .
$$
To this end one inserts identities of the form  $1 =   \left[  H_f + \Sigma[k^{(l)}] \right]^{1/2}    \left[  H_f + \Sigma[k^{(l)}] \right]^{-1/2} $, uses the Cauchy-Schwarz
inequality, and  Lemma  \ref{lem:multanihiest}.  The second inequality follows from the first since $\Sigma[k^{(p)}] \leq \Sigma[k^{(q)}]$, if $p \leq q$.
\end{proof}

Now we collect some elementary estimates.

\begin{lemma} \label{thm:estimates1} We have the estimates
\begin{itemize}
\item[(i)] $\| a(f) \varphi \| \leq \| f/\sqrt{\omega} \| \| H_f^{1/2}  \varphi \|$
\item[(ii)] $ \|  H_f^{-1/2} P_\Omega^\perp  a^*(f)  \| =    \| a(f) H_f^{-1/2} P_\Omega^\perp  \| \leq \| f/\sqrt{\omega} \|$
\item[(iii)] $\| a^*(f) \varphi \|^2  \leq  \| f \|^2 \| \varphi \|^2  + \| f/\sqrt{\omega} \|^2   \| H_f^{1/2}  \varphi \|^2$
\item[(iv)] $\| a(f) (H_f + 1)^{-1/2} \| =      \|  (H_f + 1)^{-1/2} a^*(f) \|     \leq   \| f/\sqrt{\omega} \| $
\item[(v)] $ \|  (H_f + 1)^{-1/2} a(f) \|  =   \| a^*(f) (H_f + 1)^{-1/2} \|    \leq   (\| f/\sqrt{\omega} \|^2  + \| f\|^2)^{1/2} $
\item[(vi)] $ \|  (H_f + 1)^{-1/2} \phi(f) \| =  \| \phi(f) (H_f + 1)^{-1/2} \|  \leq  \left\| \frac{f }{ 4 \pi \sqrt{\omega}}       \right\| +
 \left[\left\| \frac{f}{4 \pi \omega  }  \right\|^2  + \left\| \frac{f}{ 4 \pi  \sqrt{\omega}  } \right\|^2  \right]^{1/2}$
\end{itemize}
\end{lemma}
\begin{proof}
(i) follows from Corollary \ref{lem:multanihiest2} with $n=1$ and the Cauchy-Schwarz inequality, (ii) follows from (i), (iii) follows from (i) and the canonical
commutation relations of the creation and annihilation operators, (iv) follows from (ii), (v) follows from (iii), and (vi) follows from (iv) and (v).
\end{proof}

\section*{Appendix B: Wick's Theorem}
\label{sec:appB}

Let $(\sigma_1, \sigma_2, ... , \sigma_n \} \in \{ + , - \}^n$. For any subset $Z\subset N_n := \{ 1,2, ... , n\}$, $n \in \N$,
we set $Z_{\pm} := \{ j \in Z | \sigma_j = \pm \}$.
We use the notation $a^+(k) = a^*(k)$ and $a^-(k)= a(k)$. We define the Wick-ordered product $: \cdot :$ by
$$
: \prod_{j \in Z} a^{\sigma_j}(k_j) : := \prod_{j \in Z_+} a^{\sigma_j}(k_j)  \prod_{j \in Z_-} a^{\sigma_j}(k_j)
$$
We introduce a notation for  the vacuum expectation
of an operator $A$ on Fock space by setting
$$
\langle A \rangle :=  \langle \Omega | A \Omega \rangle .
$$
\begin{lemma} \label{lem:wick} (Wick's Theorem) For any $(\sigma_1, \sigma_2, ... , \sigma_n) \in \{ + , - \}^n$
$$
\prod_{j \in N_n} a^{\sigma_j}(k_j) = \sum_{Z \subset N_n } \left\langle \prod_{j \in N_n \setminus Z} a^{\sigma_j} (k_j) \right\rangle : \prod_{j \in Z} a^{\sigma_j}(k_j):
$$
\end{lemma}
See for example  \cite{Sch61} (Chapter 13, Eq. (106) and references therein).
The next lemma is from \cite{BFS98}.
\begin{lemma}   \label{lem:wick2} Let $f_j$ be measurable functions on $\R_+$. Then
\begin{eqnarray*}
\lefteqn{ \prod_{j=1}^n \{ a^{\sigma_j}(k_j)f_j(H_f)\} } \\
 &= \sum_{Z \subset N_n} \prod_{j \in Z_+} a^+(k_j) \Big\langle \prod_{j=1}^n \Big\{ \left[ a^{\sigma_j}(k_j) \right]^{\chi_{Z^c}(j)}
  f_j\left( H_f + r +  \omega_j \right) \Big\} \Big\rangle \Big|_{r = H_f} \prod_{j \in Z_-} a^-(k_j) ,
\end{eqnarray*}
where we set
$$
\omega_j = \sum_{\substack{i=1 \\ i \in Z_-}}^{j} |k_i| + \sum_{\substack{i=j+1 \\ i \in Z_+}}^{n} |k_i|
$$
\end{lemma}
\begin{proof}
To prove the lemma it is convenient to extend $f_j$ to the real line by setting $f_j(r)=0$ for any $r < 0$. That way we can
use the pull-through formula backwards. This and Wick's theorem yield
\begin{align*}
\prod_{j=1}^n \{ a^{\sigma_j}(k_j)f_j(H_f)\}
& = \prod_{j=1}^n a^{\sigma_j}(k_j) \prod_{j=1}^n f_j(H_f + \sum_{i =j+1}^n \sigma_i |k_i | ) \\
& = \sum_{Z \subset N_n } \left\langle \prod_{j \in N_n \setminus Z} a^{\sigma_j} (k_j) \right\rangle : \prod_{j \in Z} a^{\sigma_j}(k_j): \prod_{j=1}^n f_j(H_f + \sum_{i =j+1}^n \sigma_i |k_i | )
\end{align*}
Now using again the pull-through formula to bring the $f_j$'s to the desired position yields the  claim, noting that
$$\omega_j = \sum_{i=j+1}^n \sigma_i |k_i| + \sum_{j \in Z^-} |k_i| - \sum_{\substack{ i=j+1 \\ i \in Z^c}}^n \sigma_i |k_i| .$$
\end{proof}

\vspace{0.5cm}

\noindent
{\it Proof of Theorem \ref{thm:wicktheorem}}. For $m + n \geq 1$
we set
$$
\widehat{W}_{m,n} = \int a^*(z^{(m)}) w(H_f, Z^{(m,n)}) a(\widetilde{z}^{(n)}) \frac{d Z^{(m,n)}}{|Z^{(m,n)}|^{1/2}}.
$$
Using Lemma \ref{lem:wick2} and the observation that there are
$$\left( \begin{array}{c} m_l + p_l \\  p_l \end{array} \right)$$
ways to choose $p_l$ elements out of a set containing $m_l + p_l$ elements, we find
\begin{align*}
 & F_0(H_f) \widehat{W}_{M_1,N_1} F_1(H_f) \widehat{W}_{M_2,N_2}  \cdots  \widehat{W}_{M_L,N_L} F_L(H_f)  \\
& \quad = \sum_{\substack{p_1,...,p_L,  m_1,...,m_L, \\ q_1,...,q_L,n_1,...,n_L : \\ p_i + m_i = M_i,q_i + n_i = N_i, \\ p_l + q_l + m_l + n_l \geq 1 }}
\left( \begin{array}{c} m_l + p_l \\ p_l \end{array} \right) \left( \begin{array}{c} n_l + q_l \\ q_l \end{array} \right)  \int a^*(k_1^{(m_1)}) \cdots a^*(k_L^{(m_L)})
\\
& \quad
\langle \Omega   , F_0(H_f + r + \widetilde{r}_0 )  \widehat{W}_1( r+r_1, K_1^{(m_1,n_1)}) F_1(H_f + r + \widetilde{r}_1 ) \widehat{W}_2(r+r_2, K_2^{(m_2,n_2)}) \cdots \\
& \cdots F_{L-1}(H_f +  r + \widetilde{r}_{L-1}) \widehat{W}_L(r + r_L , K^{(m_L,n_L)}_L )    F_L(H_f + r + \widetilde{r}_L)      \Omega \rangle  \big|_{r =H_f}
\\
& \quad  a(\widetilde{k}_1^{(n_1)}) \cdots a(\widetilde{k}_L^{(n_L)}) \prod_{j=1}^L \frac{d K_j^{(m_j,n_j)}}{|K_j^{(m_j,n_j)}|^{1/2}}
\end{align*}
where we have set $\widehat{W}_i = \widehat{W}_{p_i,q_i}^{m_i,n_i}$ with
\begin{eqnarray*}
\lefteqn{  \widehat{W}_{p_l,q_l}^{m_l,n_l}[w](r, K_l^{(m_l,n_l)}) } \\
&& = P_{\rm red} \int_{B_1^{p_l + q_l} } \frac{d X^{(p_l,q_l)} }{|X^{(p_l,q_l)}|^{1/2}}
 a^*(x^{(p_l)}) w_{p_l + m_l , q_l + n_l }(H_f + r , x^{(p_l)},k_l^{(m_l)},   \widetilde{x}^{(q_l)}, \widetilde{k_l}^{(n_l)} ) a(\widetilde{x}^{(q_l)}) P_{\rm red}  .
\end{eqnarray*}
With this notation we have used the permutation symmetry of $w_{m,n}$ and
$$
d Z_l^{(M_l,N_l)} = d X^{(p_l, q_l )} d K_l^{(m_l,n_l )} .
$$
Summing over all $(M_i, N_i)$, it is easy to read off the kernels of the resulting operator, $H(\widetilde{w}^{\rm sym})$.
\qed


\begin{thebibliography}{30}






\bibitem{BCFS03} V. Bach, T. Chen, J. Fr\"ohlich, I.M. Sigal,
{\em Smooth Feshbach map and operator-theoretic renormalization group methods},
 J. Funct. Anal. 203 (2003),  44--92.

\bibitem{BFP06} V. Bach, J. Fr\"ohlich, A. Pizzo, {\it  Infrared-finite algorithms in QED: the groundstate of an atom interacting with the quantized radiation field},
Comm. Math. Phys. 264 (2006), no. 1, 145--165.

\bibitem{BFP09} V. Bach, J. Fr\"ohlich, A. Pizzo, {\it  Infrared-finite algorithms in QED. II. The expansion of the groundstate of an atom interacting with the quantized radiation field},
Adv. Math. 220 (2009), no. 4, 1023--1074.



\bibitem{BFS98}  V. Bach,  J. Fr\"ohlich, I.M. Sigal,
{\em Renormalization group analysis of spectral problems in quantum field theory},
 Adv. Math.  137 (1998),  205--298.


\bibitem{BFS99} V. Bach, J. Fr\"ohlich, I.M.  Sigal, {\it  Spectral analysis for systems of atoms and molecules coupled to the quantized radiation field},
 Comm. Math. Phys. 207 (1999), no. 2, 249--290.

\bibitem{BCVV09} J-M. Barbaroux, T. Chen, S. Vugalter, V. Vougalter
{\it  Quantitative estimates on the Hydrogen ground state energy in non-relativistic QED }. {\tt mp\_arc 09-48 }

\bibitem{CH04} I. Catto, C. Hainzl {\it Self-energy of one electron in non-relativistic QED}
 J. Funct. Anal. 207 (2004),  68--110.


\bibitem{Dix} J. Dixmier, \emph{Von Neumann algebras}, Translated from the second French edition by F. Jellett. North-Holland Mathematical Library, 27. North-Holland Publishing Co., Amsterdam-New York, 1981.

\bibitem{GLL01}
M. Griesemer, E. Lieb, M. Loss, {\it
Ground states in non-relativistic quantum electrodynamics}.
Invent. Math. 145 (2001), no. 3, 557--595.

\bibitem{G00} C. G\'{e}rard, {\it
On the existence of ground states for massless Pauli-Fierz Hamiltonians}.
 Ann. Henri  Poincar\'{e}  1 (2000),  443--459.


\bibitem{GH08} M. Griesemer, D. Hasler, {\em  On the smooth Feshbach-Schur Map},
J. Funct. Anal. 254 (2008),  2329--2335.

\bibitem{GH09} M. Griesemer, D. Hasler, {\em Analytic Perturbation Theory and Renormalization Analysis of Matter Coupled to Quantized Radiation},
Ann. Henri Poincar\'{e}.

\bibitem{HHS05} C. Hainzl, M. Hirokawa, H. Spohn, {\em
Binding energy for hydrogen-like atoms in the Nelson model without cutoffs},
J. Funct. Anal. 220 (2005),  424--459.

\bibitem{HS95}  M. H\"ubner, H. Spohn, {\em Spectral properties of the spin-boson Hamiltonian.}
 Ann. Inst. H. Poincar\'{e}   \ Phys. Tho\'{e}r. 62 (1995),  289--323.

\bibitem{K} T. Kato, \emph{Perturbation theory for linear operators}, Springer Verlag, New York, 1966, pp.75-80.

\bibitem{reesim4} M. Reed and B. Simon, {\it Methods of modern mathematical physics. IV. Analysis of operators},
Academic Press, New York-London, 1978.


\bibitem{Sch61}  S. Schweber, {\em An introduction to relativistic quantum field theory.} Foreword by Hans A. Bethe. Row, Peterson and Company, Evanston, Ill.-Elmsford, N.Y. 1961.

\bibitem{sim74} B. Simon, {\it The $P(\phi )_{2}$ Euclidean (quantum) field theory},
Princeton Series in Physics. Princeton University Press, Princeton, N.J., 1974.

\bibitem{spohn89} H. Spohn, {\it
Ground state(s) of the spin-boson Hamiltonian},
Comm. Math. Phys. 123 (1989),  277--304.










\end{thebibliography}
\end{document}